\documentclass{lmcs}

\keywords{dynamic logic, coalgebra, many-valued logic, bisimulation safety, completeness}

\usepackage{amsmath,amsthm,amssymb,tikz}
\usepackage{graphicx}
\usepackage[all,2cell]{xy}
\usepackage{tikz-cd}
\usepackage{hyperref}
\usepackage{leftindex}
\usepackage{comment}
\usepackage{hyperref}
\usepackage{euflag}

\newcommand{\alg}[1]{\mathbf{#1}}
\newcommand{\func}[1]{\mathsf{#1}}

\newcommand{\Coalg}[1]{\mathsf{Coalg}({#1})}
\newcommand{\sem}[1]{[\![{#1}]\!]}
\newcommand{\diam}[1]{\langle {#1} \rangle}
\newcommand{\boxm}[1]{[ {#1} ]}
\renewcommand{\phi}{\varphi}
\newcommand{\lra}{\leftrightarrow}

\DeclareMathOperator{\seq}{;}

\newcommand{{\twocontra}}{2^{\mathrm{(-)}}}
\newcommand{{\Acontra}}{A^{\mathrm{(-)}}}
\newcommand{\lucas}{\text{\bf\L}}
\newcommand{{\FLew}}{\mathsf{FL}_\mathrm{ew}}
\newcommand{\LT}{\mathcal{L}_{\func{F},\Pred}}
\newcommand{\LTDyn}{\mathcal{L}_{\func{F}, \Pred}^\mathsf{Dyn}}

\newcommand{\Set}{\mathsf{Set}}
\newcommand{\car}{\mathsf{U}}
\newcommand{\id}{\mathsf{id}}
\newcommand{\arity}{\mathsf{ar}}
\newcommand{\ev}{\mathsf{ev}}

\newcommand{\test}{\mathsf{test}}
\newcommand{\TeOp}{\mathsf{Te}}
\newcommand{\CoOp}{\mathsf{Op}}
\newcommand{\Pred}{\Lambda}
\newcommand{\Act}{\mathsf{Act}}
\newcommand{\AtAct}{\mathsf{At}}
\newcommand{\Prop}{\mathsf{Prop}}
\newcommand{\Formu}{\mathsf{Form}}
\newcommand{\nkTempl}{\mathsf{Templ}_{n,k}}
\newcommand{\sign}{\mathsf{sign}}

\newcommand{\ob}{\mathrm{ob}}
\newcommand{\pw}{\mathrm{pw}}
\newcommand{\dual}{(\cdot)^\partial}

\newcommand{\proves}{\mathrel{|}\joinrel\mkern-.5mu\mathrel{-}}

\newcommand{\Pow}{\mathcal{P}}
\newcommand{\PA}{\mathcal{P}_\alg{A}}

\newcommand{\kNA}{k\mathcal{N}_\alg{A}}
 
\newcommand{\M}{\mathcal{M}} 
\newcommand{\NA}{\mathcal{N}_\alg{A}}
\newcommand{\MA}{\mathcal{M}_\alg{A}}
\newcommand{\NMH}{\mathcal{H}} 

\begin{document}

\title[Many-valued coalgebraic dynamic logics]{Many-valued coalgebraic dynamic logics:\\ Safety and strong completeness via reducibility}
\thanks{We thank the three anonymous reviewers for their valuable comments and feedback.}	
\author[H.H.~Hansen]{Helle Hvid Hansen\lmcsorcid{0000-0001-7061-1219}}[a]
\author[W.~Poiger]{Wolfgang Poiger\lmcsorcid{0009-0002-8485-5905}}[b]
\address{Bernoulli Institute for Mathematics, Computer Science and Artificial Intelligence, University of Groningen, Netherlands}
\email{h.h.hansen@rug.nl}  
\address{Institute of Computer Science of the Czech Academy of Sciences, Prague, Czech Republic}	
\email{poiger@cs.cas.cz}  

\begin{abstract}
We present a coalgebraic framework for studying generalisations of dynamic modal logics such as PDL and game logic in which both the propositions and the semantic structures can take values in an algebra $\alg{A}$ of truth-degrees.  
More precisely, we work with coalgebraic modal logic via $\alg{A}$-valued predicate liftings and interpret actions (abstracting programs and games) as $\func{F}$-coalgebras where the functor $\func{F}$ represents some type of $\alg{A}$-weighted system. We also allow combinations of crisp propositions with $\alg{A}$-weighted systems and vice versa.
We introduce coalgebra operations and tests, with a focus on operations that are reducible in the sense that modalities for composed actions can be reduced to compositions of modalities for the constituent actions.
We prove that reducible operations are safe for bisimulation and behavioural equivalence, and prove a general strong completeness result, from which we obtain new strong completeness results for $\alg{2}$-valued iteration-free PDL with $\alg{A}$-weighted accessibility relations when $\alg{A}$ is a finite chain, and for many-valued iteration-free game logic with many-valued strategies based on finite \L ukasiewicz logic.       
\end{abstract}
\maketitle
\section*{Introduction}\label{section:Introduction}
\emph{Propositional dynamic logic} (PDL) \cite{FischerLadner79,HarelKozTiu00:dyna-book} is a well-known modal logic for reasoning about non-deterministic programs. In PDL, programs are an explicit part of the syntax, and PDL thus combines modal reasoning with a programming language. 
Parikh's \emph{game logic} \cite{Parikh83,PaulyParikh03} can be seen as a generalisation of PDL from programs to 2-player games, which allows for reasoning about program correctness in a distributed setting where the environment is viewed as an opponent. A variant of game logic has been applied to reason about hybrid systems \cite{Platzer15:diffGL} and neural networks \cite{TeuberMP24:diffGL-NN-verif}. 
Both PDL and game logic combine modal reasoning with an algebra of operations on their semantic structures.
This view was abstracted into a coalgebraic framework, called \emph{coalgebraic dynamic logic}~\cite{HansenKupkeLeal2014,HansenKupke2015},
which extends the approach to coalgebraic modal logic via predicate liftings. 
Coalgebraic modal logic \cite{Pattinson2003,KupkePattinson2011} allows for a uniform treatment of modal logics
over a wide range of state-based structures including Kripke models, neighbourhood models, multigraphs, weighted automata and probabilistic transition systems by modelling these structures as $\func{F}$-coalgebras for some $\Set$-endofunctor $\func{F}$.

In this paper, we generalise the coalgebraic dynamic framework and results from \cite{HansenKupkeLeal2014} to the many-valued setting.
Many-valued logic, where Boolean truth-values are replaced by an algebra $\alg{A}$ of truth-degrees, is more suitable than classical logic when reasoning about uncertainty, vagueness, and fuzzy systems, where propositions cannot be determined as being either true or false. 
There are numerous approaches to many-valued modal logic, 
each with their own motivations \cite{Fitting1991,Fitting1992,CaicedoRodriguez2010,HansoulTeheux2013}.
These approaches can generally be parametrised in the algebra $\alg{A}$ (which determines the propositional base)
and whether the semantics is given by crisp relations or $\alg{A}$-weighted (fuzzy) relations.
Our framework is also parametric in $\alg{A}$, which we assume to be (at least) a $\FLew$-algebra, \emph{i.e.}, a certain kind of residuated lattice prevalent in many-valued/fuzzy logic (Definition~\ref{definition:FLew-algebra}). 
The choice between crisp or fuzzy relations is subsumed and generalised by the choice of functor $\func{F}$. 
The coalgebraic setting further adds a choice of modalities (predicate liftings) and
the dynamic aspect amounts to a choice of operations on $\func{F}$-coalgebras. 
These parameters allow us to cover a wide range of many-valued dynamic logics (see Examples~\ref{example:CoalgebraicDynamicLogics-PDL-Crisp}--\ref{example:CoalgebraicDynamicLogics-IPDL}) in a single framework, and facilitate uniform proofs which eliminates the need to reprove concrete results for each new combination of parameters. 
The abstract setting also helps to identify which conditions ensure that certain results hold. 
Our framework accommodates iteration constructs, which are defined in  Section~\ref{section:CoalgebraicDynmicLogic} and shown to be bisimulation safe in Section~\ref{section:SafetyAndInvariance}. 
However, since iteration is not reducible, it is omitted in Section~\ref{section:ReducibleOperations} and also excluded when proving (strong) completeness in Section~\ref{section:Completeness}.
As additional techniques are needed to prove (weak) completeness including iteration, this is left for future work.

As our main contributions, we
(i) define a notion of reducibility for coalgebra operations and tests which entails the soundness of certain reduction axioms (Propositions~\ref{proposition:soundness-ReductionAxioms} \&~\ref{proposition:SoundnessTests});
(ii) prove that reducible coalgebra operations and tests are safe for bisimulation and behavioural equivalence (Propositions~\ref{proposition:ReducibleImpliesSafe} \&~\ref{proposition:ReducibleTestsAreSafe});
(iii) prove for finite $\alg{A}$ that one-step completeness of the underlying coalgebraic modal logic implies strong completeness of its dynamic version under certain conditions (Theorem~\ref{theorem:FinitelyOneStepComplete-CanModel}).
We note that (i) generalises \cite{HansenKupkeLeal2014} already in the Boolean setting, \emph{e.g.}, we do not require $\func{F}$ to be a monad in order to axiomatise sequential composition, which allows us to include Instantial PDL \cite{vanBenthemBezhanishviliEnqvist2019} in our framework.
By instantiating (iii) for concrete $\func{F}$, we obtain for any finite, linear $\FLew$-algebra $\alg{A}$, strong completeness for 2-valued iteration-free PDL with $\alg{A}$-weighted accessibility relations (Example~\ref{example: TWOPDL-labeled}), and 
if furthermore negation is an involution on $\alg{A}$  
(\emph{e.g.}, \L ukasiewicz logic), we obtain
strong completeness for an $\alg{A}$-valued version of iteration-free game logic (Example~\ref{example: MVGL}).
To our knowledge, both of these results are new.

\textbf{Related work.}
Our work can be seen as a coalgebraic generalisation of the iteration-free part of the axiomatisation for many-valued PDL with many-valued relations (for any finite $\FLew$-algebra $\alg{A}$) in \cite{Sedlar2020},  
and as a dynamic generalisation of the completeness theorems (via one-step completeness) for many-valued coalgebraic modal logics in \cite[Theorem~23]{LinLiau2023} (where $\alg{A}$ is a finite $\FLew$-algebra) and \cite{KurzPoigerTeheux2024} (where $\alg{A}$ is a semi-primal algebra). 
Completeness and decidability for variants of many-valued PDL with crisp relations were proved in  
\cite{Sedlar2016,Teheux2014}.
Expressivity of many-valued coalgebraic logics has been studied in \cite{BilkovaDostal2013,BilkovaDostal2016} 
and in connection with behavioural distances in, \emph{e.g.}, \cite{Beohar2024}.

This article is an extended version of the article \cite{HansenPoiger2025}. To improve the overall presentation, we provide more detailed examples and proofs (\emph{e.g.}, Subsection~\ref{subsection:example-applications}, Examples~\ref{example:PredicateLiftings} \&~\ref{example:ReducibleTests}, Lemma~\ref{lemma:SafetyTestOperations}, Proposition~\ref{proposition:soundness-ReductionAxioms}). Furthermore, we include additional examples which generalise and extend the operations of Instantial PDL to double monads,
\emph{i.e.}, compositions $\func{F} \circ \func{F}$ where $\func{F}$ is a monad, but $\func{F} \circ \func{F}$ is not.
In particular, the coalgebraic semantics of Instantial PDL \cite{vanBenthemBezhanishviliEnqvist2019} (see Example~\ref{example:CoalgebraicDynamicLogics-IPDL}) is given by the double monad $\mathcal{P} \circ \mathcal{P}$, which is known not to be a monad \cite{KlinSalamanca2018}.
The operations for double monads are described in Examples~\ref{example:Generalcompositions}--\ref{example:CoalgebraOperations-CounterDomain} \&~\ref{example:TestOperations}(4).
We also add a new result stating that iteration constructs (generally defined in Example~\ref{example:CoalgebraOperations-Iteration}) are always safe (Proposition~\ref{proposition:IterationIsSafe}).                  

\section{Preliminaries}\label{section:Preliminaries}

In this section, we first recall some basics of PDL and game logic (Subsections~\ref{subsection:PDL} \& \ref{subsection:GL}) to motivate the general framework (for a more thorough introduction to PDL and related concepts, we refer to the textbook \cite{HarelKozTiu00:dyna-book}). We then give brief overviews of many-valued modal logic (Subsection~\ref{subsection:Many-valuedModalLogic}) and coalgebraic logic via predicate liftings (Subsection~\ref{subsection:PredicateLiftings}).    

\subsection{Propositional Dynamic Logic}\label{subsection:PDL}

\emph{Propositional dynamic logic} (PDL) \cite{FischerLadner79,HarelKozTiu00:dyna-book} is a modal logic for reasoning about nondeterministic, imperative programs. The behaviour of a program is modelled as a relation on a state space where states represent the value of program variables and edges represent the possible state changes effected by executing the program. Syntactically, complex programs $\alpha$ are built from atomic programs (\emph{e.g.}, variable assignment statements) 
and the program constructs of nondeterministic choice ($\cup$), sequential composition ($\seq$), iteration ($^*$) and tests ($\phi?$), where $\phi?$ is the program that checks whether $\phi$ holds in the current state and aborts if this test fails, otherwise execution continues. 
A modal formula $\boxm{\alpha}\phi$ has the informal reading `\emph{after all successful executions of program $\alpha$, the property $\phi$ holds}'.

\begin{defi}[\textsf{PDL syntax}]\label{def:PDL-syntax}
Let $\Prop$ be a set of propositional variables and let $\AtAct$ be a set of atomic programs.
The sets of \emph{formulas and programs of PDL} are defined by a mutual induction by the following grammar: 
\[
\begin{array}[t]{lcl}
\text{formulas } \ \phantom{i} \phi   & \Coloneqq & p \in \Prop \mid \top \mid \lnot\phi \mid \phi\land\phi\mid \boxm{\alpha}\phi\\[.3em]
\text{programs } \ \alpha  & \Coloneqq & a \in \AtAct \mid \alpha \seq \alpha \mid \alpha\cup\alpha \mid \alpha^* \mid \phi?
\end{array}
\]
\end{defi}

The semantics of PDL is multi-relational Kripke semantics where the relations representing complex programs have fixed interpretations in order to capture the intended meaning of program constructs. 
Briefly stated, PDL semantics is `modal logic plus relation algebra'.

\begin{defi}[\textsf{PDL semantics}]\label{def:PDL-semantics}
A \emph{standard PDL model}  $M = (W, \{R_a\}_{a \in \AtAct}\}, \sem{\cdot})$ consists of a set $X$ (of states),
a binary relation $R_a \subseteq X \times X$ for every atomic program $a \in \AtAct$, 
and a valuation $\sem{\cdot} \colon \Prop \to \Pow(X)$ of propositional variables.

The semantics of formulas and complex programs is given by mutual induction:
\[ 
\begin{array}[t]{rcl}
\sem{\top} &=& X\\
\sem{\neg \phi} &=& X \setminus\sem{\phi}\\
\end{array}
\qquad
\begin{array}[t]{rcl}
\sem{\phi \land \psi} &=& \sem{\phi} \cap \sem{\psi}\\
\sem{\boxm{\alpha}\phi} &=& \{ x \in X \mid \forall y \colon\, (x,y) \in R_\alpha \Rightarrow y \in \sem{\phi} \}\\
\end{array}
\]
\[ \begin{array}{rcl}
R_{\alpha\seq\beta} &=& R_\alpha \seq R_\beta = \{ (x,z) \mid \exists y.\, (x,y) \in R_\alpha \land (y,z) \in R_\beta\}\;\; \text{(relation composition)} \\
R_{\alpha\cup\beta} &=& R_\alpha \cup R_\beta \;\; \text{(relation union)}\\
R_{\alpha^*} &=& R_{\alpha}^* \;\; \text{(reflexive, transitive closure)}\\ 
R_{\phi?} &=& \{ (x,x) \mid x \in \sem{\phi}\}
\end{array} 
\]
\end{defi}

PDL programs can encode simple \texttt{while}-programs as follows:
\begin{align*}
 \texttt{if }\phi\texttt{ then }\alpha\texttt{ else }\beta  & \ \equiv \ (\phi?\seq\alpha) \cup (\neg\phi?\seq\beta)   \\
 \texttt{while }\phi \texttt{ do }\alpha & \ \equiv \ (\phi?\seq\alpha)^*\seq\neg\phi? \\
 \texttt{skip} & \ \equiv \ \top?
\end{align*}

The PDL satisfiability problem is \textsc{exptime}-complete \cite{FischerLadner79,Pratt1980:exptimePDL}.
A complete axtiomatisation for validity over standard PDL models was first provided in \cite{Parikh78:pdl-compl}, 
see also \cite{KozenParikh1981:PDLcompl}.
We follow the presentation in \cite[Section~5.5]{HarelKozTiu00:dyna-book}.

\begin{thm}
The following Hilbert system derives a formula $\chi$ if and only if $\chi$ is valid in all standard PDL models: 

\begin{minipage}[t]{0.45\textwidth}
\begin{enumerate}[(i)]
\item All propositional tautologies.
\item $\boxm{\alpha}(\phi \to \psi) \to (\boxm{\alpha}\phi \to \boxm{\alpha}\psi)$   
\item $\boxm{\alpha}(\phi \land \psi) \lra (\boxm{\alpha}\phi \land \boxm{\alpha}\psi)$
\item $\boxm{\alpha\seq\beta}\phi \lra \boxm{\alpha}\boxm{\beta}\phi$
\item $\boxm{\alpha\cup\beta}\phi \lra \boxm{\alpha}\phi\land\boxm{\beta}\phi$
\item $\boxm{\psi?}\phi \lra (\psi \to \phi)$
\item $\boxm{\alpha^*}\phi \lra \phi\land\boxm{\alpha}\boxm{\alpha^*}\phi$
\item $\phi\land\boxm{\alpha^*}(\phi \to \boxm{\alpha}\phi) \to \boxm{\alpha^*}\phi$
\end{enumerate}
\end{minipage}
\begin{minipage}[t]{0.4\textwidth}
\[
\textrm{(MP)\;}\begin{array}{c}
\phi \quad \phi \to \psi\\\hline
\psi\\
\end{array}
\]
\[
\textrm{(Nec)\;}\begin{array}{c}
\phi\\\hline
\boxm{\alpha}\phi
\end{array}
\]
\bigskip

where $\alpha, \beta$ are arbitrary programs and $\phi, \psi$ are arbitrary formulas.
\end{minipage}
\end{thm}

We note here that items (i)--(iii) and rules (MP) and (Nec) axiomatise multi-relational normal modal logic \textbf{K}.  
Axioms (iv)--(vi) are called reduction axioms, since they show how properties of complex programs can be expressed in terms of simpler programs. Axiom (vii) expresses that $\boxm{\alpha^*}\phi$ is a fixpoint (in $\mu$-calculus, the formula $\boxm{\alpha^*}\phi$ can be expressed as $\boxm{\alpha^*}\phi \equiv \nu x . \phi \land \boxm{\alpha}x$), 
and axiom (viii) is called  the induction axiom.
The PDL axiomatisation can thus be described as (multi-relational) normal modal logic \textbf{K} extended with axioms for the program constructs. This is the view that motivated the general framework of coalgebraic dynamic logic in \cite{HansenKupkeLeal2014,HansenKupke2015}.

It is well-known that PDL is not compact, and hence not strongly complete.
Failure of compactness can be seen by
considering the infinite set $\Phi = \{p, [a]p, [a][a]p, \ldots , \neg\boxm{a^*}p\}$. Every finite subset of $\Phi$ is satisfiable, but $\Phi$ is not satisfiable. The deducibility problem for PDL was shown to be $\Pi^1_1$-complete (\emph{i.e.}, highly undecidable) in \cite{MeyerStreett:PDL-deduc}.
 
\subsection{Game Logic}\label{subsection:GL}

\emph{Game logic} \cite{Parikh83,PaulyParikh03} can be seen as generalising PDL from one player (the program) to two players (the program and an adversarial environment). The players are often referred to as Angel and Demon. 
A formula $\diam{\alpha}\phi$ is now informally read as `\emph{Angel has a strategy in the game $\alpha$ to ensure that the outcome state satisfies $\phi$}'.
Moving from programs to games amounts to moving from Kripke semantics to 
\emph{monotone neighbourhood semantics} \cite{Chellas1981,Hansen2003}. 
In particular, a game $\alpha$ is interpreted as a monotone neighbourhood function 
$E_\alpha \colon X \to \M(X)$ where $\M(X) = \{ N \in \Pow(\Pow(X)) \mid U \in N \And U \subseteq V \Rightarrow V \in N\}$.
That is, $\M(X)$ is the set of all upwards-closed subsets of $\Pow(X)$ with respect to the inclusion order.
Playing game $\alpha$ at state $x$ can be understood as Angel choosing some $U \in E_\alpha(x)$ and Demon choosing an element of $U$, which becomes the next state. The elements of $E_\alpha(x)$ can thus be seen as the strategy choices of Angel.
The syntax of game logic is obtained by extending the syntax of PDL ($\alpha$ is now a game instead of program) with the game construct \emph{dual} $\dual$. The game $\alpha^\partial$ is the same as the game $\alpha$ but with the roles of the two players interchanged. Sequential composition $\alpha \seq \beta$ is the game `\emph{first play $\alpha$ then play $\beta$}'.
In the game $\alpha \cup \beta$ Angel chooses whether they will play $\alpha$ or $\beta$.
In $\alpha^*$, the game $\alpha$ is played repeatedly, and Angel decides when they stop.  
The game $\phi?$ is called an angelic test, since it is the game in which Angel immediately loses if $\phi$ is false at the current state; otherwise, play continues. Demonic choice, iteration and tests can be defined by
$\alpha \cap \beta \coloneq (\alpha^\partial \cup \beta^\partial)^\partial$,
$\alpha^\times \coloneq ((\alpha^\partial)^*)^\partial$, and 
$\alpha! \coloneq (\neg\alpha?)^\partial$, respectively. We refer to \cite{PaulyParikh03} for details.
Parikh \cite{Parikh83} proposed an axiomatisation for game logic extending monotonic modal logic \textbf{M} \cite{Chellas1981, Hansen2003} with the following reduction axioms:

\begin{enumerate}[(i)]  
\item $\diam{\alpha\seq\beta}\phi \lra \diam{\alpha}\diam{\beta}\phi$
\item $\diam{\alpha\cup\beta}\phi \lra \diam{\alpha}\phi\lor\diam{\beta}\phi$
\item $\diam{\psi?}\phi \lra (\psi \land \phi)$
\item $\diam{\alpha^\partial} \lra \neg \diam{\alpha}\neg\phi$
\end{enumerate}
and rules expressing that $\diam{\alpha^*}\phi$ is the least fixed point of 
$U \mapsto \sem{\phi} \cup \{ x \in X \mid U \in E_\alpha(x) \}$.
However, the completeness of this Hilbert system is still open. 
The proof presented in \cite{EnqvistHKMV2019:game-logic-compl} relied on the completeness of a cut-free system for the 
modal $\mu$-calculcus in \cite{AfshariLeigh2017}, but this system later turned out not to be complete \cite{Kloibhofer23:Clo-incompl}. We refer to \cite{EnqvistHKMV2019:game-logic-compl} for a discussion of the challenges of completeness for game logic.

\subsection{Many-valued modal and dynamic logics}\label{subsection:Many-valuedModalLogic}
Standard modal logic (see, \emph{e.g.}, \cite{Blackburn2001} for a modern introduction) is based on classical propositional logic, which can be described as formulas being evaluated in the \emph{two-element Boolean algebra} $\alg{2} = \langle \{ 0,1 \}, \wedge, \vee, \rightarrow, 0, 1 \rangle$. The basic idea of \emph{many-valued modal logic} is to replace the two-element Boolean algebra $\alg{2}$ by another algebra $\alg{A}$ of \emph{truth-values} (or \emph{truth-degrees}). Most commonly, this algebra $\alg{A}$ is chosen to be (an expansion of) a \emph{$\FLew$-algebra}.  

\begin{defi}[\textsf{$\FLew$-algebras}]\label{definition:FLew-algebra}
An algebra $\alg{A} = \langle A, \wedge, \vee, \odot, \rightarrow, 0, 1 \rangle$ is a \emph{(complete) $\FLew$-algebra} if $\langle A, \wedge, \vee, 0, 1 \rangle$ is a (complete) bounded  lattice, $\langle A, \odot, 1 \rangle$ is a commutative monoid and $\alg{A}$ satisfies the \emph{residuation law} $x \odot y \leq z \Leftrightarrow x \leq y \rightarrow z$ (where $\leq$ is the lattice-order).
\end{defi}

Synonymously, $\FLew$-algebras are \emph{commutative integral residuated lattices}, \emph{i.e.}, residuated lattices where $\odot$ is commutative and the multiplicative unit $1$ coincides with the top element of the lattice (\emph{e.g.}, see \cite[Section 2.2]{GalatosJipsen2007}). The name `\emph{$\FLew$-algebra}' stems from the fact that these algebras correspond to the Full Lambek Calculus with exchange and weakening.

The choice of $\alg{A}$ determines the \emph{many-valued propositional logic} underlying the many-valued modal (coalgebraic) logic. \emph{Negation} and \emph{bi-implication} are defined as in classical logic via $\neg x \coloneq x \rightarrow 0$ and $x \leftrightarrow y \coloneq (x \rightarrow y) \wedge (y \rightarrow x)$, respectively. Different choices of a monoid operation $\odot$ on the real interval $[0,1] \subseteq \mathbb{R}$ with its usual order ($\wedge = \mathrm{min}$, $\vee = \mathrm{max}$) give rise to different \emph{fuzzy logics} \cite{Hajek1998} in (3)--(5) below.      

\begin{exa}\label{example:FLew-algebras}
The following are examples of $\FLew$-algebras.
\begin{enumerate}
\item \emph{Heyting algebras} are $\FLew$-algebras with $\odot = \wedge$, including the next two examples.  
\item The \emph{two-element Boolean algebra} $\alg{2}$ corresponds to \emph{classical propositional logic}. 
\item The \emph{standard Gödel chain} $\alg{[0,1]_G} = \langle [0,1], \wedge, \vee, \rightarrow_G, 0, 1 \rangle$ is the complete $\FLew$-algebra with $\odot = \wedge$ and $\rightarrow_\text{G}$ being the Gödel implication
$$
 x\rightarrow_\text{G} y = \begin{cases}
 1 & \text{ if } x \leq y, \\
 y & \text{otherwise.}
 \end{cases}
$$
This algebra gives rise to \emph{infinite Gödel logic}, and its finite subalgebras $\alg{G}_n$ (up to isomorphism this is uniquely defined by any $n$-element subuniverse) yield \emph{finite Gödel logics} (for an overview of Gödel logics see, \emph{e.g.}, \cite{Preining2010}).    
\item The \emph{standard $\mathsf{MV}$-chain} $\alg{[0,1]_\lucas} = \langle [0,1], \wedge, \vee, \odot, \rightarrow_\text{Ł} \rangle$ is the complete $\FLew$-algebra with the Łukasiewicz connectives 
$$
x \odot y = \mathrm{max}\{0, (x+y)-1\}, \quad x \rightarrow_\text{Ł} y = \mathrm{min}\{ 1, 1-x+y \}.
$$
This algebra gives rise to \emph{infinite Łukasiewicz logic}, and its finite subalgebras $\lucas_n$ (defined by the subuniverses $ \{ 0, \tfrac{1}{n}, \dots, \tfrac{n-1}{n}, 1 \}$) yield \emph{finite \L ukasiewicz logics} (for an overview of these logics see, \emph{e.g.}, \cite{Cignoli2000}). 
\item The \emph{standard product chain} $\alg{[0,1]_P} = \langle [0,1], \wedge, \vee, \cdot,  \rightarrow_\text{P} \rangle$, with $\odot = \cdot$ being the usual product of reals and the product implication $x \rightarrow_\text{P} y = \mathrm{min}\{\tfrac{y}{x}, 1\}$ (where we let $\tfrac{y}{0} \coloneq 1$). This algebra gives rise to \emph{product logic} (see, \emph{e.g.}, \cite{Hajek1996}). 
\item Every commutative integral \emph{quantale} $\alg{Q} = \langle Q, \bigvee, \odot, 1 \rangle$ can be identified with a complete $\FLew$-algebra via $x \rightarrow y \coloneq \bigvee\{ z \mid x \odot z \leq y  \}$.
\item In \cite{GalatosJipsen2017} an exhaustive (graphical) list of finite residuated lattices up to size six is provided. Among them, the $\FLew$-algebras are the ones where $\top = 1$ and every element is central (graphically indicated by filled-in circles $\bullet$ and squares $\text{\tiny $\blacksquare$}$ in this list). \hfill $\blacksquare$
\end{enumerate}
\end{exa}

\begin{rem}[\textsf{Semi-primal $\FLew$-algebras}]\label{remark:SemiPrimal}
The finite $\mathsf{MV}$-algebras $\lucas_n$ of Example~\ref{example:FLew-algebras}(4) are examples of $\FLew$-algebras which are \emph{semi-primal}, meaning that they satisfy a term-expressivity property generalising the well-known \emph{functional completeness} (\emph{primality}) of the Boolean algebra $\alg{2}$ (\emph{i.e.}, every map $2^n \to 2$ is term-definable in $\alg{2}$). A finite algebra $\alg{A}$ is \emph{semi-primal} \cite{Foster1964} if every map $f\colon A^n \to A$ which preserves subalgebras (\emph{i.e.}, $f(S) \subseteq S$ for every subalgebra $\alg{S} \leq \alg{A}$) is term-definable in $\alg{A}$ (for more information on semi-primality we refer the reader to \cite{KurzPoigerTeheux2024b} and references therein). For a finite $\FLew$-algebra $\alg{A}$, semi-primality is equivalent to the property that for all subsets $P \subseteq A$, the characteristic function $\chi_P \colon A \to \{0,1\}$ is term-definable in $\alg{A}$, that is, there exists a unary term-function $t_P(x) \colon A \to A$ with $t_P(a) = 1$ if $a \in P$ else $t_P(a) = 0$ \cite[Proposition 2.8]{KurzPoigerTeheux2024b}. Any finite $\FLew$-algebra has a semi-primal expansion, obtained by adding as primitive operation the characteristic functions $\chi_{ \{d\}}$ of all singleton subsets $\{ d \} \subseteq A$. Semi-primality is the key property in the study of many-valued coalgebraic logics in \cite{KurzPoigerTeheux2024, KurzPoiger2023}. In the present paper, the algebra $\alg{A}$ is \emph{not} required to be semi-primal. However, semi-primality can be convenient later on to find \emph{reduction axioms} (\emph{e.g.}, Example~\ref{example:ReducibleTests}(1)) or to ensure \emph{one-step completeness} of coalgebraic logics (\emph{e.g.}, Example~\ref{example: MVPDL-crisp}). 
\hfill $\blacksquare$                   
\end{rem}  

Every algebraic connective $c$ of $\alg{A}$ can be extended to products $A^X = \{ \sigma \colon X \to A \}$ in a pointwise manner and throughout the paper we use $c^\mathrm{pw}$ to denote the resulting operation, \emph{e.g.}, for the monoid operation $\odot$ we have $\sigma_1 \odot^\mathrm{pw} \sigma_2 (x) = \sigma_1(x) \odot \sigma_2(x)$.     

Let $\alg{A}$ be a complete $\FLew$-algebra. Formulas of \emph{$\alg{A}$-valued modal logic} are inductively defined from a countable set $\Prop$ of \emph{propositional variables}, the algebraic signature of $\alg{A}$ and the unary modalities $\Box$ and $\Diamond$. A \emph{crisp Kripke model} $(X,R, \sem{\cdot})$ consists of a Kripke frame (\emph{i.e.}, $R \subseteq X^2$) together with a propositional \emph{$\alg{A}$-valuation} map $\sem{\cdot} \colon \Prop \to A^X$, which is inductively extended to all formulas by the pointwise extension of the corresponding $\alg{A}$-operations and for modal formulas via 
\begin{equation}\label{eq:mv-mod-crisp}
    \sem{\Box\varphi}(x) = \bigwedge_{xRx'} \sem{\varphi}(x') \quad\text{ and } \quad\sem{\Diamond\varphi}(x) = \bigvee_{xRx'} \sem{\varphi}(x').
\end{equation}
More generally, in an \emph{$\alg{A}$-weighted Kripke model} $(X,R,\sem{\cdot})$ the accessibility relation is \emph{$\alg{A}$-weighted}, that is, a map $R \colon X^2 \to A$. The propositional \emph{$\alg{A}$-valuation} map $\sem{\cdot} \colon \Prop \to A^X$ is then extended to modal formulas via
\begin{equation}\label{eq:mv-mod-fuzzy}
    \sem{\Box\varphi}(x) = \bigwedge_{x' \in X} R(x,x') \rightarrow \sem{\varphi}(x') \quad\text{ and }\quad \sem{\Diamond\varphi}(x) = \bigvee_{x' \in X} R(x,x') \odot \sem{\varphi}(x').
\end{equation} 

Strong completeness results for Gödel modal logic are found in \cite{CaicedoRodriguez2010} for crisp and \cite{RodriguezVidal2021} for weighted accessibility relations. For \L ukasiewicz modal logics such results are found in \cite{HansoulTeheux2013} for crisp and \cite{BouEstevaGodo2011} for weighted relations. The latter also includes strong completeness results for $\alg{A}$-valued modal logic (crisp and weighted) for any finite $\FLew$-algebra expanded with truth-constants (the crisp case also requires a unique coatom). Let us also mention that completeness for many-valued modal logic with \emph{neighbourhood semantics} over arbitrary $\FLew$-algebras is discussed in \cite{CintulaNoguera2018}. 

Prior work on \emph{many-valued propositional dynamic logic} includes the following. Using \emph{crisp accessibility} relations, the papers \cite{Teheux2014} and \cite{Sedlar2016} contain decidability and completeness proofs for PDL based on finite \L ukasiewicz logic (also see Example~\ref{example:CoalgebraicDynamicLogics-PDL-Crisp}) and on four-valued Belnap-Dunn logic, respectively. For many-valued PDL with \emph{$\alg{A}$-weighted accessibility} relations, the article \cite{Sedlar2020} contains decidability and completeness results for any finite $\FLew$-algebra $\alg{A}$ expanded with truth-constants, and \cite{Sedlar2021} offers additional complexity results for the case of $\alg{A} = \lucas_n$ being a finite \L ukasiewicz chain (also see Example~\ref{example:CoalgebraicDynamicLogics-PDL-Labelled1}).             

\subsection{Coalgebraic logic via predicate liftings}
\label{subsection:PredicateLiftings}

We assume familiarity with basic category theory and coalgebraic logic, and only briefly recall relevant definitions and fix notation.
For overviews of the topics, we refer to \cite{MacLane1997} for elementary category theory, to \cite{Rutten00} for coalgebras and \cite{KupkePattinson2011} for an overview of coalgebraic logic.
Unless stated otherwise, $\alg{A}$ is an arbitrary $\FLew$-algebra.

\begin{defi}[\textsf{Coalgebras}]\label{definition:coalgebras}
Let $\func{F}$ be a $\Set$-endofunctor. An \emph{$\func{F}$-coalgebra} is a map $\gamma\colon X \to \func{F}X$. 
A map $f\colon X \to Y$ is a \emph{coalgebra morphism} from $\gamma\colon X \to \func{F}X$ to $\gamma'\colon Y \to \func{F}Y$, denoted $f \colon \gamma \to \gamma'$,
if $\func{F}f \circ \gamma = \gamma' \circ f$, that is, the below diagram commutes.
$$ 
\begin{tikzcd}[row sep = large, column sep = large]
X \arrow[r, "\gamma"] \arrow[d, "f"']
& \func{F}X \arrow[d, "\func{F}f"] \\
Y \arrow[r,"\gamma'"']
&  \func{F}Y
\end{tikzcd}  
$$
$\func{F}$-coalgebras with these morphisms form the \emph{category of $\func{F}$-coalgebras}, denoted by $\Coalg{\func{F}}$.
\end{defi} 

Coalgebras provide a uniform way to model \emph{state-based transition systems}. From the logical perspective, these give rise to the structures on which formulas are interpreted semantically, as illustrated in the following.   

\begin{exa}\label{example:FunctorsAndTheirCoalgebras/Kleisli}
The following are examples of (categories of) coalgebras. In (2), we require $\alg{A}$ to be complete.
\begin{enumerate}
\item It is well-known that $\Coalg{\mathcal{P}}$ for the \emph{covariant powerset functor} $\mathcal{P}$ is isomorphic to the category of \emph{Kripke frames} (with \emph{bounded morphisms}). More specifically, a Kripke frame $(X, R)$ corresponds to the coalgebra $\gamma \colon X \to \mathcal{P}X$ via $xRy \Leftrightarrow y \in \gamma(x)$. 
\item More generally, if $\alg{A}$ is a complete $\FLew$-algebra, the \emph{covariant $\alg{A}$-powerset functor} $\PA$ is defined 
by $\PA X = A^X$ on objects, and sends a map $f\colon X \to Y$ to $\PA f \colon A^X \to A^Y$ defined by $(\PA f)(\sigma) \colon y \mapsto \bigvee \{ \sigma(x) \mid f(x) = y \}$. Coalgebras for the functor $\PA$ are Kripke frames with \emph{$\alg{A}$-weighted accessibility relations} (or \emph{$\alg{A}$-weighted graphs}) as discussed in Subsection~\ref{subsection:Many-valuedModalLogic}.  
\item The \emph{$\alg{A}$-neighbourhood functor} $\NA$ is given by $\NA \coloneq \Acontra \circ \Acontra$, where we use $\Acontra$ to abbreviate the contravariant hom-functor $\Set(-,A)$. The \emph{monotone $\alg{A}$-neighbourhood functor} $\MA$ is the subfunctor of $\NA$ obtained by $\MA X = \{ N \in \NA X \mid \forall \sigma_1,\sigma_2 \in A^X \colon\sigma_1 \leq \sigma_2 \Rightarrow N(\sigma_1) \leq N(\sigma_2) \}$. In the case $\alg{A} = \alg{2}$, the corresponding coalgebras are \emph{(monotone) neighbourhood frames} (see, \emph{e.g.}, \cite{HansenKupke2004,Pacuit2017}). In general,
$\NA$-coalgebras correspond to \emph{$\alg{A}$-neighbourhood frames} in \cite{CintulaNoguera2018}.   
\item The \emph{instantial neighbourhood functor} is the double (covariant) powerset functor $\mathcal{P} \circ \mathcal{P}$ and the corresponding coalgebras are \emph{instantial neighbourhood frames} considered in \cite{vanBenthemBezhanishviliEnqvistYu2017,vanBenthemBezhanishviliEnqvist2019}. Note that on the object level these coincide with neighbourhood frames, however their corresponding coalgebra morphisms are different. \hfill $\blacksquare$
\end{enumerate}
\end{exa}

\begin{rem}[\textsf{Pointed Monads}]\label{remark:PointedMonads}
Prior work on coalgebraic dynamic logic \cite{HansenKupkeLeal2014,HansenKupke2015} is based on the assumption that the coalgebraic signature functor is a \emph{monad} $(\func{F}, \eta, \mu)$ (for the theory of monads see, \emph{e.g.}, \cite[Chapter VI]{MacLane1997}). In order to define tests, the monad $\func{F}$ is furthermore required to be \emph{pointed} \cite[Section 2.3]{HansenKupkeLeal2014}, meaning that for every set $X$ there exists $\bot_{\func{F}X} \in \func{F}X$ such that $\func{F}f(\bot_{\func{F}X}) = \bot_{\func{F}Y}$ for every map $f \colon X \to Y$. In the present article we do \emph{not} require $\func{F}$ to be a monad or to be pointed. However, most of our examples are based on pointed monads $\func{F}$ or \emph{double (pointed) monads} $\func{F} \circ \func{F}$, since these structures give rise to uniform notions, for example of sequential composition (see Examples~\ref{example:CoalgebraOperations-Kleisli} \& ~\ref{example:Generalcompositions}) and tests (see Example~\ref{example:TestOperations}).          
\hfill $\blacksquare$
\end{rem}

Formal reasoning about classes of coalgebras is the subject of \emph{coalgebraic logic}. We follow the \emph{predicate lifting} approach (see, \emph{e.g.}, \cite{Pattinson2003,Schroeder2008};  for an overview of all approaches \cite{KupkePattinson2011}), that is, our modalities are given semantics by certain \emph{natural transformations}. The following is a straightforward generalisation of Boolean-valued predicate liftings to the many-valued setting (introduced in \cite{PattinsonSchroeder2011} and also used, \emph{e.g.}, in \cite{LinLiau2023}).  
     
\begin{defi}[\textsf{Predicate liftings}]\label{definition:PredicateLifting}
A $k$-ary \emph{$A$-predicate lifting} for $\func{F}$ (where $k \geq 1$) is a natural transformation $\lambda\colon (\Acontra)^k \Rightarrow \Acontra \circ \func{F}$. 
\end{defi}  

Note that for any set $X$, via transposition, we have that $\lambda_X \colon (A^X)^k \rightarrow  A^{\func{F}X}$
corresponds to $\widehat{\lambda}_X \colon \func{F}X  \rightarrow A^{(A^X)^k}$.
Hence, a $k$-ary $A$-predicate lifting $\lambda$ equivalently corresponds to a natural transformation $\widehat{\lambda} \colon \func{F} \Rightarrow \kNA$, where $\kNA$ is defined by $X \mapsto A^{(A^X)^k}$ and similarly to the $\alg{A}$-neighbourhood functor $\NA$ (Example~\ref{example:FunctorsAndTheirCoalgebras/Kleisli}(3)) on morphisms. We frequently employ this correspondence throughout this paper. 

\emph{Formulas} of coalgebraic logic are defined inductively from a countable set $\Prop$ of propositional variables, the algebraic operations $o \in \sign(\alg{A})$ in the signature of $\alg{A}$ (which at the least includes $\top, \bot, \wedge, \vee, \odot, \rightarrow$) and a collection of predicate liftings $\lambda \in \Lambda$.     
\begin{align*}
&&\varphi \Coloneqq p \in \mathsf{Prop} 
&\text{ $\big\vert$ } o(\varphi_1, \dots, \varphi_{\arity(o)})\phantom{.} 
\text{ $\big\vert$ } \heartsuit^\lambda (\varphi_1,\dots,\varphi_{\arity(\lambda)}) 
\end{align*}
 
A \emph{model} consists of a coalgebra $\gamma \colon X \to \func{F}X$ and a propositional $\alg{A}$-valuation $\sem{\cdot} \colon \Prop \to A^X$, which is inductively extended to all formulas via the pointwise extension of the corresponding $\alg{A}$-operations and to modal formulas via
$$
\sem{\heartsuit^\lambda (\varphi_1, \dots, \varphi_{\arity(\lambda)}} = \lambda_X(\sem{\varphi_1}, \dots, \sem{\varphi_{\arity(\lambda)}}) \circ \gamma.
$$

A key feature of coalgebraic logic is that modal truth is \emph{invariant under coalgebra morphisms}, in particular under \emph{bisimulation and behavioural equivalence} (cf.~\cite[Definition~4.1 \& Corollary~4.4]{KupkePattinson2011}). In order to have an \emph{expressive} modal language, \emph{i.e.}, where modal equivalence implies behavioural equivalence, the collection of predicate liftings needs to be able to distinguish elements in $\func{F}X$ (see, \emph{e.g.}, \cite{Pattinson2004, Schroeder2008}).
Later on in Section~\ref{section:Completeness} we require separation on \emph{finite} $X$ in order to obtain \emph{quasi-canonical models} \cite{SchroederPattinson2009}.
  
\begin{defi}[\textsf{Separating}]\label{definition:SeparatingPredLifts}
A collection $\Pred$ of $A$-predicate liftings is \emph{separating} if the collection of transposes $\{ \widehat{\lambda} \mid \lambda \in \Pred \}$ is jointly monic, \emph{i.e.}, whenever $t_1 \neq t_2$ holds for two elements $t_1,t_2 \in \func{F}X$, then there exists some $\lambda \in \Pred$ for which $\widehat{\lambda}_X(t_1) \neq \widehat{\lambda}_X(t_2)$ holds. 
The collection $\Lambda$ is called \emph{separating on finite sets} if the above holds for all finite sets $X$. 
\end{defi}

We now present how predicate liftings give rise to semantics of various modal logics, in particular (crisp and weighted) many-valued modal logic (see Subsection~\ref{subsection:Many-valuedModalLogic}). 

\begin{exa}\label{example:PredicateLiftings}
The following are examples of $A$-predicate liftings. In (1)--(3), we require $\alg{A}$ to be complete.
\begin{enumerate}
\item $\mathcal{P}$-coalgebras are crisp Kripke frames and
the semantics of the $\alg{A}$-valued $\Box$ and $\Diamond$ modalities (Eq.~\eqref{eq:mv-mod-crisp}) correspond to unary $A$-predicate liftings whose components $\lambda_X^\Box, \lambda_X^\Diamond \colon A^X \to A^{\mathcal{P}X}$ are defined on $\sigma \in A^X$ and $Z \in \mathcal{P}X$ by 
$$
\lambda_X^\Box(\sigma) \colon Z \mapsto  \bigwedge \sigma[Z] \quad \text{ and } \quad \lambda_X^\Diamond(\sigma) \colon Z \mapsto  \bigvee \sigma[Z].
$$  
Both $\{ \lambda^\Box \}$ and $\{ \lambda^\Diamond \}$ are separating, we show the former to illustrate. Let $Z_1,Z_2 \in \mathcal{P}X$ be distinct and without loss of generality pick $z \in Z_1{\setminus}Z_2$. Let $\sigma \in A^X$ be defined by $\sigma(z) = 0$ while $\sigma(x) = 1$ for all $x \neq z$. Then $\lambda^\Box_X(\sigma)(Z_1) = 0$ while $\lambda^\Box_X(\sigma)(Z_2) = 1$. 
\item For the functor $\PA$, whose coalgebras are $\alg{A}$-weighted Kripke frames, the $\alg{A}$-valued $\Box$ and $\Diamond$ modalities on $\alg{A}$-weighted frames (Eq.~\eqref{eq:mv-mod-fuzzy}) correspond to the unary $A$-predicate liftings with components $\lambda_X^\Box, \lambda_X^\Diamond \colon A^X \to A^{\PA X}$ defined on $\sigma \in A^X$ and $\rho \in \PA X$ by
$$
\lambda_X^\Box(\sigma) \colon \rho \mapsto  \bigwedge_{x\in X} \rho(x) \rightarrow \sigma(x) \quad \text{and} \quad \lambda_X^\Diamond(\sigma) \colon \rho \mapsto  \bigvee_{x\in X} \rho(x) \odot \sigma(x).
$$  
Again, it can be checked that $\{ \lambda^\Box \}$ and $\{ \lambda^\Diamond \}$ are separating, we prove the latter to illustrate. Let $\rho_1(z) \neq \rho_2(z)$ for $\rho_1,\rho_2 \in A^X$ and set $\sigma(z) = 1$ while $\sigma(x) = 0$ for all $x \neq z$. Then $\lambda_X^\Diamond(\sigma)(\rho_1) = \rho_1(z)$ while $\lambda_X^\Diamond(\sigma)(\rho_2) = \rho_2(z)$. 
\item For $\PA$, we also consider a Boolean-valued logic over $\alg{A}$-weighted frames, with $2$-predicate liftings $\lambda^r_X \colon 2^X \to 2^{A^X}$ indexed by elements $r\in \alg{A}$ which are defined by 
$$
\lambda^r_X(S) = \{ \rho \in A^X \mid  \bigvee \rho(S) \geq r \}.
$$ 
The collection $\{ \lambda^r \}_{r \in A}$ is separating. To see this, let $\rho_1(z) \neq \rho_2(z)$ and without loss of generality assume $\rho_1(z) \not\leq \rho_2(z)$. Set $r = \rho_1(z)$ and $S = \{ z \}$, then $\lambda^r(S)(\rho_1) = 1$ while $\lambda^r(S)(\rho_2) = 0$. Note that if the lattice-order $\leq$ of $\alg{A}$ is \emph{linear} (\emph{i.e.}, $\forall x,y\colon x \leq y \vee y \leq x$) the collection $\{ \lambda^r \}_{r \in D}$ is already separating if $D \subseteq A{\setminus}\{ 0 \}$ is an \emph{order-dense} subset, as the proof can be easily adapted. 
\item For $\NA$ and $\MA$, the unary $A$-predicate lifting $\lambda^\ev$ which evaluates at its input, is defined by 
$$
\lambda^\ev_X(\sigma) = \ev_\sigma \colon N \mapsto N(\sigma)
$$
and corresponds to \emph{$\alg{A}$-valued neighbourhood semantics} \cite{CintulaNoguera2018}.  
One easily checks that $\{ \lambda^\ev \}$ is separating for both $\NA$ and $\MA$.
\item For the instantial neighbourhood functor $\mathcal{P} \circ \mathcal{P}$ and any $k \in \mathbb{N}$, we define the $(k+1)$-ary $2$-predicate lifting $\lambda^{k+1}$ by 
$$
\lambda^{k+1}_X(S_1, \dots, S_k, S) = \{ N \subseteq \mathcal{P}X \mid \exists Z \in N\colon Z \subseteq S \wedge (\forall i \colon  Z \cap S_i \neq \varnothing)\}.$$ 
The collection $\{ \lambda^{k+1} \}_{k \in \mathbb{N}}$ corresponds to the modalities of \emph{instantial neighbourhood semantics} \cite{vanBenthemBezhanishviliEnqvistYu2017,vanBenthemBezhanishviliEnqvist2019}.
This collection is not separating (\emph{e.g.} set $X = \mathbb{N}$, then $N_1 = \mathcal{P}(\mathbb{N})$ and $N_2 = \mathcal{P}_{\mathrm{fin}}(\mathbb{N})$ can not be separated). However, this collection is separating on finite sets $X$, since for any subset $Z =\{ z_1, \dots, z_k \} \subseteq X$ it holds that $Z \in N \Leftrightarrow N \in \lambda^{k+1}(\{z_1\}, \dots, \{z_k\}, Z)$. \hfill $\blacksquare$
\end{enumerate}
\end{exa}

Our aim in the following section is to extend this coalgebraic framework to the domain of many-valued dynamic logics. 

\section{A general coalgebraic framework for dynamic logics}\label{section:CoalgebraicDynmicLogic}
Before establishing syntax and semantics of general coalgebraic dynamic logics (Subsection~\ref{subsection:SyntaxSemantics}), we introduce coalgebra operations and tests to provide semantics for the \emph{dynamic} part in the following subsection.
\subsection{Coalgebra operations and tests}\label{subsection:Coalg/Test-Operations}
For any $\Set$-endofunctor $\func{F}$, we use $\car_\func{F}\colon \Coalg{\func{F}} \rightarrow \Set$ to denote the forgetful functor sending a coalgebra $X \to \func{F}X$ to its \emph{carrier set} $X$. We use $\func{F}^n X = \func{F}X \times \cdots \times \func{F}X$ to denote the $n$-fold product (not composition of $\func{F}$) and often use $\vec{\gamma}$ to denote an $\func{F}^n$-coalgebra, 
identifying it with an $n$-tuple $(\gamma_1, \dots, \gamma_n)$ of $\func{F}$-coalgebras in the obvious way. Lastly, for a category $\mathcal{C}$ we denote by $\mathcal{C}_\ob$ its corresponding \emph{discrete category}, that is, the corresponding category having the same objects but only identity morphisms, and (since no ambiguity arises) we also use $\car_\func{F}$ for the analogous forgetful functor $\Coalg{\func{F}}_\ob \to \Set$. Our uniform coalgebraic generalisation of operations such as sequential composition, non-deterministic choice and iteration used in various dynamic logics is subsumed by the following notion.

\begin{defi}[\textsf{Coalgebra operations}]\label{definition:CoalgebraOperation}
Let $\func{F}\colon \Set \to \Set$ be an endofunctor. An \emph{$\func{F}$-coalgebra operation} of arity $n \geq 1$ is a functor $O \colon \Coalg{\func{F}^n}_\ob \to \Coalg{\func{F}}_{\ob}$ 
satisfying $\car_\func{F} \circ O = \car_{\func{F}^n}$, that is, for which the following diagram commutes:
$$ 
\begin{tikzcd}[row sep = small]
\Coalg{\func{F}^n}_\ob \arrow[rr, "O"] \arrow[dr, "\car_{\func{F}^n}"']
&& \Coalg{{\func{F}}}_\ob \arrow[dl, "\car_{\func{F}}"] \\
& \Set &
\end{tikzcd}  
$$
\end{defi}  

In other words, an $n$-ary $\func{F}$-coalgebra operation takes $n$-many coalgebras $\gamma_1, \dots, \gamma_n \colon X \to \func{F}X$ and returns a composed coalgebra $O(\gamma_1, \dots, \gamma_n) \colon X \to \func{F}X$ on the same carrier set. Note that at this point we do not require any compatibility with \emph{coalgebra morphisms} here. However, we will require it later in the context of \emph{safety} (Section~\ref{section:SafetyAndInvariance}).

In the following examples, we illustrate how operations used in various instances of dynamic logics are described (and generalised) as coalgebra operations. 

\begin{exa}[\textsf{Pointwise induced coalgebra operations}]
\label{example:CoalgebraOperations-Induced}
A $\Set$-indexed family of maps $\theta_X \colon \func{F}^n X \to \func{F}X$ (which need not form a natural transformation) yields the \emph{induced $n$-ary coalgebra operation} $O^\theta \colon \Coalg{\func{F}^n}_\ob \to \Coalg{\func{F}}_\ob$ pointwise via 
$$
O^\theta(\vec{\gamma}) = \theta_X \circ \vec{\gamma}.
$$
The following are specific examples of induced coalgebra operations.  
\begin{enumerate}[(i)]
\item For $\mathcal{P}$, set-theoretical operations like union $\cup$ and intersection $\cap$ induce binary coalgebra operations.  
\item More generally for $\PA$, algebraic operations of $\alg{A}$ induce coalgebra operations. For example $\vee, \wedge, \odot \colon \PA^2 X\to \PA X$ (taken pointwise, \emph{e.g.}, $(\sigma_1 \wedge \sigma_2)(x) = \sigma_1(x) \wedge \sigma_2(x)$) all induce binary coalgebra operations.
\item Similarly for $\NA$ and $\MA$, any (monotone) $\alg{A}$-operation, such as $\wedge, \vee, \odot \colon \NA^2 X \rightarrow \NA X$ (again taken pointwise, \emph{e.g.}, $(N_1 \wedge N_1)(\sigma) = N_1(\sigma) \wedge N_2(\sigma)$) induces a coalgebra operation.  
\item For $\NA$ and $\MA$, the \emph{dual} operation $\dual$ is defined by $N^\partial(\sigma) = \neg N(\neg \sigma)$ and induces a unary coalgebra operation. Here, $\neg \sigma$ for $\sigma \in A^X$ is defined pointwise.
\item The \emph{neighbourhood-wise} union $N_1 \Cup N_2 = \{ Z_1 \cup Z_2 \subseteq X \mid Z_i \in N_i\}$ induces a binary coalgebra operation for the instantial neighbourhood functor $\mathcal{P} \circ \mathcal{P}$ (in \cite{vanBenthemBezhanishviliEnqvist2019} this is the \emph{parallel composition} denoted by `$\cap$'). \hfill $\blacksquare$ 
\end{enumerate} 
\end{exa}

\begin{exa}[\textsf{Kleisli composition}]
\label{example:CoalgebraOperations-Kleisli}
Let $(\func{F}, \eta, \mu)$ be a \emph{monad} (see, \emph{e.g.}, \cite[Chapter VI]{MacLane1997}). As in \cite{HansenKupkeLeal2014}, a natural choice of \emph{sequential composition} is the binary coalgebra operation of \emph{Kleisli composition}, given by 
$$
\gamma_1 {;} \gamma_2 = \mu_X \circ \func{F}\gamma_2 \circ \gamma_1.
$$
This includes the following examples in particular. 
\begin{enumerate}[(i)]
\item For the covariant powerset functor $\mathcal{P}$, one obtains the usual composition of relations 
$$
x(R_1{;}R_2)z \Leftrightarrow \exists y \colon x R_1 y R_2 z.
$$ 
\item More generally for $\PA$, one obtains the composition of $\alg{A}$-weighted relations defined by 
$$
(R_1{;}R_2)(x,z) = \bigvee \{ R_1(x,y) \odot R_2(y,z) \mid y\in X\}.
$$ 
\item $\NA$ and $\MA$ are monads with Kleisli composition 
$$
(\xi_1{;}\xi_2)(x)(\sigma) = \xi_1(x)(\ev_\sigma \circ \xi_2).
$$ 
For $\alg{A} = \alg{2}$, this yields $S \in (N_1{;}N_2)(x) \Leftrightarrow \{ y \in X \mid S \in N_2(y) \} \in N_1(x)$, the sequential composition used in game logic. \hfill $\blacksquare$   
\end{enumerate}    
\end{exa}

\begin{exa}[\textsf{General Compositions \& Double Monads}]
\label{example:Generalcompositions}
Note that the definition of Kleisli composition only requires a natural transformation $M \colon \func{FF} \Rightarrow \func{F}$, indeed \emph{every} such natural transformation yields a binary coalgebra operation via $\gamma_1 {;} \gamma_2 = M_X \circ \func{F}\gamma_2 \circ \gamma_1$.

For the instantial neighbourhood functor $\mathcal{P} \circ \mathcal{P}$ (which is not a monad \cite{KlinSalamanca2018}), one recovers the \emph{sequential composition of instantial neighbourhoods} $;$ used in \cite[Section 3]{vanBenthemBezhanishviliEnqvist2019} (therein denoted `$\circ$') in this way with $M \colon \mathcal{PPPP} \Rightarrow \mathcal{PP}$ defined on $\Xi \in \mathcal{PPPP}X$ by
$$
M_X(\Xi) = \{ \bigcup G \mid \exists \mathcal{F} \in \Xi \colon G \subseteq \bigcup\mathcal{F}  \And  \forall N \in \mathcal{F} \colon G \cap N \neq \varnothing \}.
$$
To see this note that, given two coalgebras $\gamma_1, \gamma_2 \colon X \to \mathcal{PP}X$, instantiating the above at $\Xi = (\mathcal{PP}\gamma_2 \circ \gamma_1)(x) =  \{ \gamma_2[Z] \mid Z \in \gamma_1(x) \}$ yields 
$$
(\gamma_1 {;} \gamma_2)(x) = \big\lbrace \bigcup G \mid \exists Z \in \gamma_1(x)\colon G \subseteq \bigcup \gamma_2[Z] \And \forall z \in Z \colon \gamma_2(z) \cap G \neq \varnothing  \big\rbrace 
$$
(writing $z R_2 Y \Leftrightarrow Y \in \gamma_2(z)$, we have that the Egli-Milner lifted relation $Z \overline{R_2} G$ holds if and only if $\forall Y \in G \colon \exists z \in Z \colon z R_2 Y$ and $\forall z \in Z \colon  \exists Y \in G \colon z R_2 Y$ -- note that $G \subseteq \bigcup \gamma_2[Z]$ is an equivalent reformulation of the former and $\forall z \in Z \colon \gamma_2(z) \cap G \neq \varnothing$ is an equivalent reformulation of the latter).

 Let $(\func{F}, \eta, \mu)$ be a monad. Then there are various ways to define sequential compositions for the corresponding \emph{double monad}\footnote{ Our choice of the term `\emph{double monad}' derives from the commonly used term `\emph{double powerset}'.} $\func{FF}$. For instance, the natural transformations $M^\bullet, M^\star \colon \func{FFFF} \Rightarrow \func{FF}$ given by 
$$
M^\bullet = \func{F}\mu \circ \func{FF}\mu \quad \text{and} \quad M^\star = \mu_\func{F} \circ \mu_\func{FF} 
$$
give rise to binary coalgebra operations which we denote by `$\bullet$' and `$\star$'. In particular, for $\mathcal{P} \circ \mathcal{P}$ this yields  
\begin{align*}
(\gamma_1 \bullet \gamma_2)(x) &= \big\{ \bigcup\{ \bigcup \gamma_2(z) \mid z \in Z \} \big| Z \in \gamma_1(x) \big\} \quad \text{and}\\
(\gamma_1 \star \gamma_2)(x) &= \{ Z \subseteq X \mid \exists Y \in \gamma_1(x) \colon \exists y \in Y \colon Z \in \gamma_2(y) \}    
\end{align*}
as alternative compositions of instantial neighbourhoods (not considered in \cite{vanBenthemBezhanishviliEnqvist2019}, but studied in \cite[Examples~5.5 \& 6.7]{MotamedOR26:coalg-social-sys} in the context of analysing social roles). 
\hfill $\blacksquare$
\end{exa}

\begin{exa}[\textsf{Iteration constructs}]
\label{example:CoalgebraOperations-Iteration}
For monads $\func{F}$, we define \emph{Kleisli iteration} $(\cdot)^\ast$ similarly to \cite[Section 2.2]{HansenKupke2015}. For this, we assume that there is natural transformation $\bigsqcup \colon \mathcal{P}\func{F} \Rightarrow \func{F}$ (\emph{i.e.}, $\func{F}f(\bigsqcup_{i \in I} t_i) = \bigsqcup_{i \in I} \func{F}f(t_i)$ for all $f\colon X \to Y$) such that for all $X$, $(\func{F}X, \bigsqcup)$ is a complete sup-semilattice structure.
For a coalgebra $\gamma \colon X \to \func{F}X$ we define   
$$
\gamma^{[0]} = \eta_X, \quad \gamma^{[n+1]} = \gamma;\gamma^{[n]}, \quad \gamma^\ast = \bigsqcup_{n \in \omega} \gamma^{[n]}.
$$
For $\mathcal{P}$, using $\bigsqcup = \bigcup$ we recover iteration of PDL. 
Similarly, one can define iteration for $\PA$ and $\MA$ using $\bigsqcup = \bigvee$ being the (corresponding) pointwise join.\footnote{ Naturality for $\PA$ follows from the definition of $\PA f$ via joins, recall Example~\ref{example:FunctorsAndTheirCoalgebras/Kleisli}(2). Naturality for $\MA$ holds because $\MA f(\bigvee \nu_i)(\sigma) = (\bigvee\nu_i)(\sigma \circ f) = \bigvee\nu_i(\sigma \circ f) = \bigvee \MA f(\nu_i)(\sigma) = (\bigvee \MA f(\nu_i))(\sigma)$.}  

Note that the definition of Kleisli iteration only requires a natural transformation $H \colon \textsf{Id}_{\textsf{Set}} \Rightarrow \func{F}$, a binary coalgebra operation $;$ and natural transformation $\bigsqcup$ as above. With this, we can define general iterations via 
$$
\gamma^{[0]} = H_X, \quad \gamma^{[n+1]} = \gamma;\gamma^{[n]}, \quad \gamma^\ast = \bigsqcup_{n \in \omega} \gamma^{[n]}
$$
In particular, for $\mathcal{P}\circ \mathcal{P}$ we recover \emph{instantial neighbourhood iteration} used in \cite[Section 4.1]{vanBenthemBezhanishviliEnqvist2019} using $H_X(x) = (\mathcal{P}\eta_X \circ \eta_X)(x) = \big\lbrace\!\lbrace x \rbrace\!\big\rbrace$ and instantial neighbourhood composition $;$. Similarly, the alternative compositions $\bullet$ and $\star$ of Example~\ref{example:Generalcompositions} yield corresponding alternative iterations $(\cdot)^\bullet$ and $(\cdot)^\star$, respectively. \hfill $\blacksquare$
\end{exa}

\begin{exa}[\textsf{Counter-domain generalised}]
\label{example:CoalgebraOperations-CounterDomain}
For the covariant powerset functor $\mathcal{P}$, the \emph{counter-domain} is a unary coalgebra operation $\sim$ defined by 
$$
({\sim}\gamma) (x) = \begin{cases}
\{ x \} & \text{ if } \gamma(x) = \varnothing, \\ 
\varnothing & \text{ if } \gamma(x) \neq \varnothing.
\end{cases}
$$ 
It was shown in \cite{vanBenthem1998} that every bisimulation-invariant first-order definable operation of binary relations is a combination of the counter-domain together with the PDL program constructs and propositional tests. 

The counter-domain can be generalised to a unary coalgebra operation $\sim$ for any $\func{F}$ or $\func{F \circ F}$ with a \emph{pointed monad} $\func{F}$ (see Remark~\ref{remark:PointedMonads}) via 
$$
({\sim}\gamma) (x) = \begin{cases}
\eta_X(x) & \text{ if } \gamma(x) = \bot_{\func{F}X} \\ 
\bot_{\func{F}X} & \text{ if } \gamma(x) \neq \bot_{\func{F}X}
\end{cases}
\quad \text{and} \quad
({\sim}\gamma) (x) = \begin{cases}
\func{F}\eta_X\circ \eta_X(x) & \text{ if } \gamma(x) = \bot_{\func{FF}X} \\ 
\bot_{\func{FF}X} & \text{ if } \gamma(x) \neq \bot_{\func{FF}X}
\end{cases}
$$   
respectively. For $\PA$ we call the resulting coalgebra operation the \emph{counter-support} and for $\mathcal{P} \circ \mathcal{P}$ we call the resulting operation the \emph{instantial counter-domain}. \hfill $\blacksquare$
\end{exa}

The prior coalgebraic framework for (two-valued) dynamic logics found in \cite{HansenKupkeLeal2014,HansenKupke2015} covers operations induced pointwise by natural transformations (Example~\ref{example:CoalgebraOperations-Induced}) as well as Kleisli composition and Kleisli iteration (Examples~\ref{example:CoalgebraOperations-Kleisli} \& \ref{example:CoalgebraOperations-Iteration}), but does not include the operations in Examples~\ref{example:Generalcompositions} \& \ref{example:CoalgebraOperations-CounterDomain}.  

We now define \emph{tests} in a manner similar to coalgebra operations. In the following, we use $\func{A}$ to denote the \emph{constant functor} sending every set to $A$ and every map to $\id_A$ (note that $\func{A}$-coalgebras can be identified with $A$-predicates). 

\begin{defi}[\textsf{Tests}]\label{definition:TestOperation}
An \emph{$A$-test} for $\func{F}$ is a functor $\test \colon \Coalg{\func{A}}_\ob \to \Coalg{\func{F}}_\ob$ 
such that $\car_\func{F} \circ \test = \car_{\func{A}}$, that is,
for which the following diagram commutes:  
$$ 
\begin{tikzcd}[row sep = small]
\Coalg{\func{A}}_\ob \arrow[rr, "\test"] \arrow[dr, "\car_{\func{A}}"']
&& \Coalg{{\func{F}}}_\ob \arrow[dl, "\car_{\func{F}}"] \\
& \Set &
\end{tikzcd}   
$$ 
\end{defi}

In other words, an $A$-test for $\func{F}$ turns a formula $\varphi$ corresponding to a predicate $\sem{\varphi} \in A^X$ into a coalgebra $\test(\sem{\varphi}) \colon X \to \func{F}X$ (intuitively speaking, the transition system for $\varphi?$). 

In the following, we illustrate how the tests of previously considered dynamic logics arise in this way. Again, our framework generalises the one established in \cite{HansenKupkeLeal2014}, where only tests of the form in (1) below (with $P = \{ 1 \}$) are considered. 

\begin{exa}\label{example:TestOperations} 
The following are examples of tests (here, $\sigma \colon X \to A$ and $S \colon X \to 2$). 
\begin{enumerate}
\item Suppose that $\func{F}$ is a \emph{pointed monad} (see Remark~\ref{remark:PointedMonads}). Let $P \subseteq A$ be an arbitrary subset (defining some `property' or `desirable values' in $A$). Then we define the $A$-test $\test_P$ via 
$$
\test_P(\sigma)(x) = \begin{cases}
\eta_X(x) & \text{ if } \sigma(x)\in P, \\
\bot_{\func{F}X} & \text{ otherwise}.
\end{cases}
$$
Informally, $\test_P(\sigma)$ should be understood as `\emph{if $\sigma$ evaluates to a value in $P$ then continue, else abort}'. In particular, with $P = \{ 1 \}$ this yields the usual test for truth used in \cite{HansenKupkeLeal2014} (where $\alg{A} = \alg{2}$) and \cite{Teheux2014} (where $\alg{A} = \lucas_n$ and $\func{F} = \mathcal{P}$).   
\item For $\PA$, another $A$-test is given by the pointwise monoid operation 
$$
\test(\sigma)(x) = e_x \odot^\mathrm{pw} \sigma,
$$
with the `unit vector' $e_x(x) = 1$, otherwise $e_x(x') = 0$. This yields the $\alg{A}$-weighted relation $R(x,x) = \sigma(x)$, and $R(x,x') = 0$ for $x \neq x'$ which has been proposed as adequate test for many-valued PDL (\emph{e.g.}, in \cite{Sedlar2020}).   
\item For $\MA$, defining 
$$
\test(\sigma)(x) = \ev_x\big( (\text{-}) \odot^\mathrm{pw} \sigma\big)
$$ 
generalises angelic tests of game logic, cf.~Subsection~\ref{subsection:GL}.

For example, in \L ukasiewicz logic, $\test(\sem{\psi})(x)(\sem{\phi}) = \sem{\psi}(x) \odot \sem{\phi}(x)$ reduces the value of $\sem{\phi}(x)$ by $(1-\sem{\psi}(x))$.
Demonic tests are obtained as $\test(\neg\sigma)^\partial$, which yields $\test(\sem{\neg\psi})^\partial(x)(\sem{\phi}) = \sem{\psi}(x) \oplus \sem{\phi}(x)$, \emph{i.e.}, adding $\sem{\psi}(x)$ to $\sem{\phi}(x)$. 
\item For the instantial neighbourhood functor $\mathcal{P} \circ \mathcal{P}$, the $2$-test used in Instantial PDL \cite{vanBenthemBezhanishviliEnqvist2019} is defined via 
$$
\test(S)(x) = \begin{cases}
\big\lbrace\!\lbrace x \rbrace\!\big\rbrace & \text{if } x \in S, \\
\varnothing & \text{otherwise}. 
\end{cases}
$$
More generally, for any double pointed monad $\func{F} \circ \func{F}$ and $P \subseteq A$, there is an $A$-test defined by 
$$
\test_P(\sigma)(x) = \begin{cases}
\func{F}\eta_X \circ \eta_X(x) & \text{if } \sigma(x) \in P, \\
\bot_{\func{FF}X} & \text{otherwise}. 
\end{cases}
$$  
This is similar to (1), and we recover the Instantial PDL test with $P = \{ 1 \}$. 
\hfill $\blacksquare$
\end{enumerate}
\end{exa}

We are now ready to define a general syntax and semantics for many-valued coalgebraic dynamic logics in the following subsection.

\subsection{Syntax and semantics of coalgebraic dynamic logics}\label{subsection:SyntaxSemantics}

For the remainder of the article, let $\alg{A} = \langle A, \wedge, \vee,  \odot, \rightarrow, 0, 1 \rangle$ be a \emph{$\FLew$-algebra}, possibly extended with further operations.
We write $\sign(\alg{A})$ for the \emph{algebraic signature} of $\alg{A}$.
Furthermore, let $\Pred$ be a countable collection of \emph{$A$-predicate liftings} for the \emph{$\Set$-endofunctor} $\func{F} \colon \Set \to \Set$. Let $\CoOp$ be a finite collection of \emph{coalgebra operations} (Definition~\ref{definition:CoalgebraOperation}) and let $\TeOp$ be a finite collection of \emph{$A$-tests} for $\func{F}$ (Definition~\ref{definition:TestOperation}). We use the common notation $\arity(o)$, $\arity(\lambda)$ and $\arity(O)$ to denote the finite \emph{arity} of any $o\in \sign(\alg{A})$, $\lambda \in \Pred$ and $O \in \CoOp$, respectively. Lastly, let $\Prop$ be a countably infinite set of \emph{propositional variables} and let $\AtAct$ be a countable collection of \emph{atomic actions} (we use `\emph{action}' as a generalisation of `\emph{program}' and `\emph{game}').  

From this data, first we define general formulas and actions. 

\begin{defi}[\textsf{Formulas \& Actions}]\label{definition:ActionsFormulas}
We define the set of \emph{dynamic formulas} $\varphi \in \Formu$ and the set of \emph{actions} $a \in \Act$ by simultaneous induction as follows.
\begin{align*}
&&\varphi &\ \Coloneqq \ p \in \mathsf{Prop} 
\text{ $\big\vert$ } o(\varphi_1, \dots, \varphi_{\arity(o)})\phantom{.} 
\text{ $\big\vert$ } \diam{a}^\lambda (\varphi_1,\dots,\varphi_{\arity(\lambda)}) 
  \\
&&a & \ \Coloneqq \ a_0 \in \AtAct 
\text{ $\big\vert$ } O(a_1, \dots, a_{\arity(O)}) 
\text{ $\big\vert$ } \test(\varphi)    
\end{align*} 
Here, $o \in \mathsf{sign}(\alg{A})$ ranges over (propositional) connectives of $\alg{A}$, modal formulas correspond to predicate liftings $\lambda\in\Pred$ and complex actions are obtained from coalgebra operations $O \in \CoOp$ and tests $\test \in \TeOp$.
\end{defi} 
    
Note that the set of dynamic formulas $\Form$ always includes formulas of the form $\varphi \wedge \varphi$, $\varphi \vee \varphi$, $\varphi \odot \varphi$ and $\varphi \rightarrow \varphi$, as well as $\top \coloneq 1$, $\bot \coloneq 0$ and $\neg \varphi \coloneq \varphi \rightarrow \bot$, since $\alg{A}$ is (an extension of) a $\FLew$-algebra. 

The coalgebraic dynamic \emph{semantics} are as follows.

\begin{defi}[\textsf{Models}]\label{definition:model}
A \emph{(dynamic) model} is a map $\gamma \colon \Act \to (\func{F}X)^X$ together with a \emph{propositional $\alg{A}$-valuation} $\sem{\cdot} \colon \Prop \to A^X$. 
\end{defi} 

For a model as above we write $\gamma_a \colon X \to \func{F}X$ rather than $\gamma(a)$. As usual, the propositional valuation is inductively extended to all formulas in $\Formu$ via the rules 
\begin{align*}
&& \sem{o(\varphi_1, \dots, \varphi_{\arity(o)})} &= o^\mathrm{pw}\big(\sem{\varphi_1}, \dots, \sem{\varphi_{\arity(o)}}\big), \\
&& \sem{\diam{a}^\lambda(\varphi_1,\dots, \varphi_{\arity(\lambda)})} &= \lambda_X\big(\sem{\varphi_1},\dots, \sem{\varphi_{\arity(\lambda)}}\big) \circ \gamma_a.
\end{align*}

Note that, in order to fully determine a model $\gamma$, it suffices to provide a map $\gamma^0 \colon \AtAct \to (\func{F}X)^X$ defined only on \emph{atomic} actions together with the propositional valuation $\sem{\cdot} \colon \Prop \to A^X$, since one can extend the valuation to all formulas as above and simultaneously extend $\gamma^0$ to \emph{all} actions via 
$$
\gamma_{O(a_1, \dots, a_{\arity(O)})} = O(\gamma_{a_1},\dots, \gamma_{a_{\arity(O)}}), \quad \gamma_{\test(\varphi)} = \test(\sem{\varphi}).
$$  
We are mainly interested in models which arise in this way which, as in \cite{HansenKupkeLeal2014}, we refer to as \emph{standard} models.

\begin{defi}[\textsf{Standard models}]\label{definition:StandardModel}
A model $\gamma \colon \Act \to (\func{F}X)^X$ is \emph{standard} if it is obtained from some $\gamma^0 \colon \AtAct \to (\func{F}X)^X$ together with a propositional valuation as described above. 
\end{defi} 
A formula $\phi$ is \emph{true} at a state $x$ if  $\sem{\varphi}(x) = 1$,
and \emph{local semantic entailment} $\Gamma \models \varphi$ holds if and only if 
every state $x$ of every \emph{standard} model satisfies $\sem{\varphi}(x) = 1$ whenever $\sem{\psi}(x) = 1$ for all $\psi \in \Gamma$. 

We now present (in Examples~\ref{example:CoalgebraicDynamicLogics-PDL-Crisp}--\ref{example:CoalgebraicDynamicLogics-IPDL} below) some coalgebraic dynamic logics which will serve as our main examples throughout this article. To the best of our knowledge, no instances of Examples~\ref{example:CoalgebraicDynamicLogics-PDL-Labelled2} and \ref{example:CoalgebraicDynamicLogics-GL} previously appeared in the literature.

\begin{exa}[\textsf{Many-valued PDL with crisp accessibility relations}]
\label{example:CoalgebraicDynamicLogics-PDL-Crisp}
With the coalgebraic signature functor $\func{F} = \mathcal{P}$ being the covariant powerset functor, a dynamic model $\gamma \colon \Act \to (\mathcal{P}X)^X$ can be identified with a multi-relational Kripke frame $(X, \{R_a\}_{a \in \Act})$ containing a (crisp) relation $R_a \subseteq X^2$ for every action $a \in \Act$. The coalgebra operations $\CoOp = \{ \cup, ;  , (\cdot)^\ast\}$ and test(s) $\TeOp \subseteq \{ \test_P \mid P \subseteq A \}$ are interpreted by the rules
\begin{align*}
x R_{a\cup b} y & \ \Longleftrightarrow \ xR_a y \text{ or } x R_b y, \\
x R_{a;b} y & \ \Longleftrightarrow \ \exists v\colon xR_a v R_b y, \\
x R_{a^\ast}y & \ \Longleftrightarrow \ x = y \text{ or }\exists v_1, \dots, v_n : xR_a v_1 R_a v_2R_a \dots v_{n-1}R_a v_n R_a y \\
x R_{\test_P (\psi)} y &\ \Longleftrightarrow \ x = y \text{ and } \sem{\psi}(x) \in P
\end{align*}
in every standard model. The $A$-predicate liftings $\Lambda \subseteq \{ \lambda^\Box, \lambda^\Diamond \}$ correspond to $\alg{A}$-valued modal logic on crisp Kripke frames (Example~\ref{example:PredicateLiftings}(1)) and yield the interpretations of modalities (where we abbreviate $\boxm{a} \coloneq \diam{a}^{\lambda^\Box}$ and $\diam{a} \coloneq \diam{a}^{\lambda^\Diamond}$) via 
$$
\sem{\boxm{a}\varphi}(x) = \bigwedge_{x R_a y} \sem{\varphi}(y) \quad \text{and} \quad
\sem{\diam{a}\varphi}(x) = \bigvee_{x R_a y} \sem{\varphi}(y).  
$$
With $\alg{A} = \lucas_n$ a finite \L ukasiewicz chain, $\Pred = \{\lambda^\Box \}$ and $\TeOp = \{ \test_{\{ 1 \} } \}$,  this logic is considered in \cite{Teheux2014} (used, \emph{e.g.}, to model \emph{searching games with errors}).
\hfill $\blacksquare$
\end{exa}

\begin{exa}[\textsf{Many-valued PDL with weighted accessibility relations}]
\label{example:CoalgebraicDynamicLogics-PDL-Labelled1}
With the coalgebraic signature functor $\func{F} = \PA$, a dynamic model can be identified with $(X, \{ R_a \}_{a \in \Act})$, where the  $R_a \colon X^2 \to A$ are $\alg{A}$-weighted accessibility relations for every $a \in \Act$. The coalgebra operations $\CoOp = \{ \vee, ;, (\cdot)^\ast \}$ and the test of Example~\ref{example:TestOperations}(2) correspond to the rules
\begin{align*}
R_{a\vee b}(x,y) & \ = \ R_a(x,y) \vee R_b(x,y), \\
R_{a;b}(x, y) & \ = \ \bigvee_{v \in X} R_a(x,v) \odot R_b(v,y), \\
R_{a^\ast}(x,y) & \ = \ e_x(y) \vee \bigvee_{v_i \in  X} R_a(x,v_1) \odot R_a(v_1,v_2) \odot \dots  R_a(v_{n-1}, v_n) \odot R_a(v_n, y),\\
R_{\test(\psi)}(x,y) & \ = \ \begin{cases}
\sem{\psi}(x) &\text{ if }x = y \\
0 &\text{ if } x \neq y
\end{cases}
\end{align*}
in every standard model. The predicate liftings $\Lambda \subseteq \{ \lambda^\Box, \lambda^\Diamond \}$ correspond to $\alg{A}$-valued modal logic on $\alg{A}$-weighted frames (Example~\ref{example:PredicateLiftings}(2)) and yield the interpretations of modalities (where we abbreviate $\boxm{a} \coloneq \diam{a}^{\lambda^\Box}$ and $\diam{a} \coloneq \diam{a}^{\lambda^\Diamond}$) via 
$$
\sem{\boxm{a}\varphi}(x) = \bigwedge_{y \in X} R_a(x,y) \rightarrow \sem{\varphi}(y), \quad \text{ and } \quad
\sem{\diam{a}\varphi}(x) = \bigvee_{y \in X} R_a(x,y) \odot \sem{\varphi}(y).  
$$      
With $\Pred = \{\lambda^\Box \}$, this is similar to the finitely-valued PDL considered in \cite{Sedlar2020}\footnote{ Therein, $(\cdot)^+$ instead of $(\cdot)^\ast$ is used and tests are excluded.} to argue, among others, about \emph{program costs}. For example, the formula $\diam{\pi_1 \vee \pi_2}\top \rightarrow \diam{\pi_2}\top$ can intuitively be read as `\emph{the maximal cost of executing $\pi_2$ is greater than that of $\pi_1$}'.
\hfill $\blacksquare$
\end{exa}

\begin{exa}[\textsf{Two-valued PDL with weighted accessibility relations}]
\label{example:CoalgebraicDynamicLogics-PDL-Labelled2}
We use the same coalgebraic signature functor $\func{F} = \PA$, coalgebra operations $\CoOp = \{ \vee, ; (\cdot)^\ast\}$ and test as in the previous example. However, now we consider the Boolean-valued logic defined by the collection of $2$-predicate liftings $\Lambda = \{ \lambda^r \mid r\in A{\setminus}\{ 0 \} \}$ of Example~\ref{example:PredicateLiftings}(3). This gives rise to the interpretation of modalities (where we abbreviate $\diam{a}^r \coloneq \diam{a}^{\lambda^r}$) via 
$$
x \Vdash \diam{a}^r\varphi \Longleftrightarrow \bigvee_{x' \Vdash \varphi} R_a(x,x') \geq r.  
$$    
For example, if $\alg{A}$ is again thought of as describing \emph{costs} of programs, the formula $\diam{\pi}^r \varphi$ can intuitively be read as `\emph{reaching $\varphi$ after executing $\pi$ may cost more than $r$}'.
\hfill $\blacksquare$
\end{exa}

\begin{exa}[\textsf{Many-valued game logic}]
\label{example:CoalgebraicDynamicLogics-GL}
With the coalgebraic signature functor $\func{F} = \MA$, a dynamic model $\gamma \colon \Act \to (\MA  X)^X$ can be identified with a structure $(X, \{\nu_a\}_{a\in\Act})$ containing an $\alg{A}$-neighbourhood assignment $\nu_a \colon X \to A^{A^X}$ for every action $a \in \Act$ (and these assignments are monotone in the sense that $\sigma_1 \leq^\pw \sigma_2 \Rightarrow \nu_a(x)(\sigma_1) \leq \nu_a(x)(\sigma_2)$). The operations $\CoOp = \{ \vee, \wedge, \dual, ; , (\cdot)^\ast \}$ and the test of Example~\ref{example:TestOperations}(3) correspond to the rules
\begin{align*}
\nu_{a\vee b}(x)(\sigma) & \ = \ \nu_a(x)(\sigma) \vee \nu_b(x)(\sigma), \\
\nu_{a\wedge b}(x)(\sigma) & \ = \ \nu_a(x)(\sigma) \wedge \nu_b(x)(\sigma), \\
\nu_{a^\partial}(x)(\sigma) & \ = \ \neg \nu_a(x)(\neg\sigma), \\
\nu_{a;b}(x)(\sigma) & \ = \ \nu_a(x)(\ev_\sigma \circ \nu_b), \\
\nu_{a^\ast}(x)(\sigma) & \ = \ \sigma(x) \vee \bigvee_{n\geq 1} \nu^{[n]}_a(x)(\sigma), \\
\nu_{\test(\psi)}(x)(\sigma) & \ = \ \sigma(x) \odot \sem{\psi}(x)
\end{align*}
in every standard model (note that $\ev_\sigma \colon A^{A^X} \to A$ so the fourth rule is well-defined). The predicate lifting $\Lambda = \{ \lambda^\ev\}$ of Example~\ref{example:PredicateLiftings}(4) yields the interpretations of modalities (where we abbreviate $\diam{a} \coloneq \diam{a}^{\lambda^\ev}$) via 
\begin{align*}
\sem{\diam{a}\varphi}(x) &= \nu_a(x)(\sem{\varphi}).  
\end{align*}  
This logic can be seen as describing $\alg{A}$-valued strategic ability in two-player games (say between Angel and Demon) as follows. 
Angel wants to maximise truth-values, while Demon wants to minimise them.
Suppose that for each game $\alpha$, we have a monotone neighbourhood function $S_\alpha \colon X \to \M X$ where $U \in S_\alpha(x)$ 
means that Angel has a strategy when playing $\alpha$ in state $x$ to ensure that the outcome is in $U$. 
Let 
$$
\nu_\alpha(x)(\sigma) = \bigvee_{U \in S_\alpha(x)} \bigwedge_{y \in U} \sigma(y)
$$
for all states $x \in X$ and all $\sigma \in A^X$ (it is easy to verify monotonicity using that $S_\alpha$ is monotonic).
One then obtains $\sem{\diam{\alpha}\phi}(x) = \bigvee_{U \in S_\alpha(x)} \bigwedge_{y \in U} \sem{\phi}(y)$, that is, $\sem{\diam{\alpha}\phi}(x)=r$ means that at state $x$, Angel has a strategy in $\alpha$ to ensure $\sem{\phi}(y) \geq r$ for any outcome $y$.
In an angelic test $\diam{\psi?}\phi$, the truth value of $\psi$ constrains the overall truth value, and hence can be seen as a challenge for Angel (also see Example~\ref{example:TestOperations}(3)).
The operations $;$ , $\dual$ and $\vee$ are consistent with the above interpretation and the corresponding operations on $S$ in classic game logic \cite{PaulyParikh03}. In particular, for all $\sigma \in A^X$, $\nu_{\alpha^\partial}(x)(\sigma) = \neg \nu_{\alpha}(x)(\neg \sigma)$ is 
the least $\sigma$-value that Demon can ensure when playing $\alpha$ at $x$. 
If $\alg{A}$ is not linear or $\neg$ is not involutive then $\dual$ can no longer be understood as introducing the other player (negation is involutive, \emph{e.g.}, in \L ukasiewicz logic, but not in G\"odel and product logic). 
\hfill $\blacksquare$
\end{exa}

\begin{exa}[\textsf{Instantial PDL}]
\label{example:CoalgebraicDynamicLogics-IPDL}
With the coalgebraic signature functor
$\func{F} = \mathcal{P} \circ \mathcal{P}$, a dynamic model $\gamma \colon \Act \to (\func{F}X)^X$ can be identified with a structure $(X, \{N_a\}_{a\in\Act})$ containing a neighbourhood assignment $N_a \colon X \to \mathcal{PP}X$ for every action $a \in \Act$. The coalgebra operations $\CoOp = \{ \Cup, ;,(\cdot)^\ast\}$ and the test of Example~\ref{example:TestOperations}(4) correspond to the rules
\begin{align*}
N_{a \Cup b}(x) & \ = \ \{ Z \mid \exists Z' \in N_a(x), Z'' \in N_b(x) \colon Z = Z' \cup Z''  \}, \\
N_{a;b}(x) & \ = \ \big\lbrace \bigcup G \mid \exists Z \in N_a(x) \colon G \subseteq \bigcup_{z \in Z} N_b(z) \wedge \forall z \in Z: G \cap N_b(z) \neq \varnothing\big\rbrace, \\
N_{a^\ast}(x) & \ = \ \big\lbrace\!\{ x \}\!\big\rbrace\cup \bigcup_{n \geq 1} N_a^{[n]}(x)  \\
N_{\test(\psi)}(x) & \ = \ \begin{cases}
\big\lbrace\!\lbrace x \rbrace\!\big\rbrace & \text{if } x \Vdash \psi \\
\varnothing & \text{otherwise}
\end{cases} 
\end{align*}
in every standard model. The predicate liftings $\Lambda = \{ \lambda^{k+1} \mid k \in \mathbb{N} \}$ of Example~\ref{example:PredicateLiftings}(5) give rise to the semantics of modal formulas (where we abbreviate $\diam{a}^{k+1} \coloneq \diam{a}^{\lambda^{k+1}}$) via 
$$
x \Vdash \diam{a}^{k+1}(\varphi_1, \dots, \varphi_k, \varphi) \Leftrightarrow \exists Z \in N_a(x) \colon Z \subseteq \sem{\varphi} \text{ and } Z \cap \sem{\varphi_i} \neq \varnothing \text{ for all } i = 1, \dots, k.  
$$
This logic is called \emph{Instantial PDL} and was introduced in \cite{vanBenthemBezhanishviliEnqvist2019} as a modal logic for \emph{computation in open systems}.
\hfill $\blacksquare$
\end{exa}
Note that Example~\ref{example: INSTPDL} will show that our framework generalises the one established in \cite{HansenKupkeLeal2014, HansenKupke2015} already in the two-valued case, as it uses a coalgebraic signature functor which is not a monad, as well as countably many predicate liftings of arbitrary finite arities.      

In the next section, we discuss the topic of safety with respect to bisimulation and behavioural equivalence in this general coalgebraic setting.

\section{Safety and logical invariance}\label{section:SafetyAndInvariance}
A fundamental compositional aspect of dynamic logics such as PDL and game logic is that operations and tests are safe for bisimilarity \cite{vanBenthem1998, Pauly2000}, meaning that it suffices to check bisimilarity for the atomic actions to conclude bisimilarity for all actions. 
In our coalgebraic setting (where bisimilarity and behavioral equivalence may differ), 
safety can conveniently be defined as requiring that the operation preserves coalgebra morphisms. That is, while in Definition~\ref{definition:CoalgebraOperation} we defined coalgebra operations as functors $O \colon \Coalg{\func{F}^n}_\ob \to \Coalg{\func{F}}_{\ob}$ on the discrete categories of coalgebras only, safe operations can be seen as similar functors on the full categories of coalgebras.

\begin{defi}[\textsf{Safe coalgebra operations}]\label{definition:SafeCoalgebraOperation}
A $n$-ary coalgebra operation $O$ is \emph{safe} if it also defines (via $Of = f$) a functor $O \colon \Coalg{\func{F}^n} \to \Coalg{\func{F}}$ with $\car_\func{F} \circ O = \car_{\func{F}^n}$, that is, for which the following diagram similarly to Definition~\ref{definition:CoalgebraOperation} commutes: 
$$ 
\begin{tikzcd}[row sep = small]
\Coalg{\func{F}^n} \arrow[rr, "O"] \arrow[dr, "\car_{\func{F}^n}"']
&& \Coalg{{\func{F}}} \arrow[dl, "\car_{\func{F}}"] \\
& \Set &
\end{tikzcd}  
$$
\end{defi}      

In other words, safety ensures for coalgebras $\gamma_1,\dots,\gamma_n\colon X\to\func{F}X$ and $\gamma'_1,\dots,\gamma'_n\colon Y\to\func{F}Y$ that if $f\colon X \to Y$ is a joint coalgebra morphism $\gamma_i \to \gamma'_i$, then $f$ is also a coalgebra morphism $O(\vec{\gamma}) \to O(\vec{\gamma}')$.

For instance, the coalgebra operations of Examples~\ref{example:CoalgebraOperations-Kleisli}--\ref{example:CoalgebraOperations-CounterDomain} are all safe, which can be shown directly (in particular, for iteration constructs see Proposition~\ref{proposition:IterationIsSafe}), or as a consequence of Proposition~\ref{proposition:ReducibleImpliesSafe} later on. On the contrary, not all pointwise induced operations of Example~\ref{example:CoalgebraOperations-Induced} are safe. As an example, for the functor $\PA$ it holds that $\vee$ induces a safe coalgebra operation but $\wedge$ does not, the reason behind this being as follows. 

\begin{lem}\label{lemma:SafeInducedOperations}
Let $O^\theta\colon \Coalg{\func{F}^n}_\ob \to \Coalg{\func{F}}_\ob$ be induced by the $\Set$-indexed collection of maps $\theta_X \colon \func{F}^n X \to \func{F}X$. Then $O^\theta$ is safe if and only if $\theta$ is a natural transformation $\func{F}^n \Rightarrow \func{F}$. 
\end{lem}

\begin{proof}
Assuming $\theta$ is natural, we show that the outer diagram below commutes for any joint coalgebra morphism $f \colon \vec{\gamma} \to\vec{\gamma}'$. The left square commutes since $f$ is a joint coalgebra morphism and the right square commutes by naturality of $\theta$.   
$$ 
\begin{tikzcd}[row sep = large, column sep = large]
X \arrow[r, "{\vec{\gamma}}"] \arrow[d, "f"']
& \func{F}^n X \arrow[r, "{\theta_X}"] \arrow[d, "\func{F}^n f"]
&\func{F}X \arrow[d, "\func{F}f"] \\
Y \arrow[r,"{\vec{\gamma}'}"']
&  \func{F}^n Y \arrow[r, "\theta_Y"']
& \func{F}Y
\end{tikzcd}  
$$
Conversely, if $\theta$ does not define a natural transformation, there is a map $f\colon X \to Y$ and a tuple $(t_1, \dots, t_n) \in \func{F}^n X$ such that $\func{F}f(\theta_X(t_1,\dots,t_n)) \neq \theta_Y(\func{F}f(t_1), \dots, \func{F}f(t_n))$. With the constant coalgebras $\gamma_i(x) = t_i$ and $\gamma_i'(y) = \func{F}f(t_i)$, this $f$ witnesses that $O^\theta$ is not safe. 
\end{proof} 

For further illustration, we show that the \emph{compositions} of Examples~\ref{example:CoalgebraOperations-Kleisli} \&~\ref{example:Generalcompositions} as well as \emph{iterations} of Example~\ref{example:CoalgebraOperations-Iteration} are always safe.   

\begin{prop}\label{proposition:IterationIsSafe}
For any monad, Kleisli composition and Kleisli iteration are safe. More generally, all compositions and iterations obtained as in Examples~\ref{example:Generalcompositions} \& \ref{example:CoalgebraOperations-Iteration} are safe.  
\end{prop}
\begin{proof}
To see that Kleisli composition of a monad $(\func{F}, \eta, \mu)$ is safe, let $f\colon X \to Y$ be a joint coalgebra morphism $(\gamma_1, \gamma_2) \to (\gamma_1', \gamma_2')$. Then $f$ also is a coalgebra morphism $f\colon\gamma_1;\gamma_2 \to \gamma'_1;\gamma'_2$ since the outer diagram
$$ 
\begin{tikzcd}[row sep = large, column sep = large]
X \arrow[r, "{\gamma_1}"] \arrow[d, "f"']
& \func{F}X \arrow[r, "{\func{F}\gamma_2}"] \arrow[d, "\func{F} f"']
& \func{FF}X \arrow[d, "\func{FF}f"] \arrow[r, "\mu_X"]
& \func{F}X \arrow[d, "\func{F}"]
\\
Y \arrow[r,"\gamma'_1"']
& \func{F}Y \arrow[r, "\func{F}\gamma'_2"']
& \func{FF}Y \arrow[r, "\mu_Y"']
& \func{F}Y
\end{tikzcd}  
$$  
commutes (because $f \colon \gamma_i \to \gamma'_i$ and $\mu$ is a natural transformation). 

For the case of Kleisli iteration, let $f\colon X \to Y$ be a coalgebra morphism $\gamma \to \tilde{\gamma}$. In order to show that $f$ is also a coalgebra morphism $\gamma^\ast \to \tilde{\gamma}^\ast$, first we use naturality of $\bigsqcup$ to find 
$$
\func{F}f\big(\gamma^\ast(x)\big) = \func{F}f\big(\bigsqcup_{n\in\omega} \gamma^{[n]}(x)\big) = \bigsqcup_{n\in\omega} \func{F}f\big(\gamma^{[n]}(x)\big).
$$  
Thus, we are done once we show $\func{F}f \circ \gamma^{[n]} = \tilde{\gamma}^{[n]} \circ f$ (\emph{i.e.}, $f \colon \gamma^{[n]} \to \tilde{\gamma}^{[n]}$) for all $n$, which we do by induction. The case $n = 0$ is simply naturality of $\eta$ (recall $\gamma^{[0]} = \eta_X$ and $\tilde{\gamma}^{[0]} = \eta_Y$). If $f\colon \gamma^{[n]} \to \tilde{\gamma}^{[n]}$, then we also get $f \colon \gamma;\gamma^{[n]} \to \tilde{\gamma};\tilde{\gamma}^{[n]}$ by our (initial) assumption $f\colon \gamma \to \tilde{\gamma}$ together with safety of $;$ as shown above.  

Since the above proofs only depend on the fact that the components involved are defined via various natural transformations, it is straightforward to show the same for the general compositions and iterations of Examples~\ref{example:Generalcompositions} \& \ref{example:CoalgebraOperations-Iteration}.
\end{proof}

Next we define safety of \emph{tests} (Definition~\ref{definition:TestOperation}) by functoriality on morphisms analogous to safety of coalgebra operations. 

\begin{defi}[\textsf{Safe tests}]\label{definition:SafeTestOperation}
An $A$-test for $\func{F}$ is \emph{safe} if it also defines (via $\test f = f$) a functor $\test \colon \Coalg{\func{A}} \to \Coalg{\func{F}}$ with $\car_\func{A} \circ O = \car_\func{F}$, that is, for which the following diagram similarly to Definition~\ref{definition:TestOperation} commutes:
$$ 
\begin{tikzcd}[row sep = small]
\Coalg{\func{A}} \arrow[rr, "\test"] \arrow[dr, "\car_{\func{A}}"']
&& \Coalg{{\func{F}}} \arrow[dl, "\car_{\func{F}}"] \\
& \Set &
\end{tikzcd}   
$$ 
\end{defi}

Similarly to Lemma~\ref{lemma:SafeInducedOperations}, we proceed to show that safety of tests can also be characterized by a certain \emph{naturality condition}. In the following, we use $\func{F}(\text{-})^{A^{(\text{-})}}$ to denote the $\Set$-functor defined by $X \mapsto (\func{F}X)^{(A^X)}$ on objects and $f \mapsto \func{F}f \circ (-) \circ A^f$ on morphisms.

\begin{lem}\label{lemma:SafetyTestOperations}
Let $\test$ be an $A$-test for $\func{F}$. Then the following are equivalent. 
\begin{enumerate}[(i)]
\item The functor $\test$ is a safe test. 
\item The transposes $\widehat{\test}_X \colon X \to (\func{F}X)^{(A^X)}$ define a natural transformation $\func{Id} \Rightarrow \func{F}(\text{-})^{A^{(\text{-})}}$.
\item For every $f \colon X \to Y$, the following diagram commutes:
$$ 
\begin{tikzcd}[row sep = tiny, column sep = large]
A^Y \arrow[r, "\test_Y"] \arrow[dd, "(\text{--})\circ f"'] & (\func{F}Y)^Y \arrow[dr, near start, "(\text{--})\circ f"] & 
\\
& & (\func{F}Y)^X 
\\
A^X \arrow[r,"\test_X"'] &  (\func{F}X)^X \arrow[ur, near start, "\func{F}f \circ (\text{--})"'] &
\end{tikzcd}  
$$ 
\end{enumerate}  
\end{lem}
\begin{proof}
The equivalence (ii)$\Leftrightarrow$(iii) is immediate by definition of natural transformation and the functor $\func{F}(\text{-})^{A^{(\text{-})}}$ (indeed, spelling out this definition (ii) precisely yields $\func{F}f \circ \test(\sigma' \circ f) = \test(\sigma') \circ f$ for every $f\colon X \to Y$ and $\sigma' \in A^Y$). 

For (iii)$\Rightarrow$(i), first note that $\func{A}$-coalgebras are predicates, say $\sigma \colon X \to A$ and $\sigma'\colon Y \to A$ and a coalgebra morphism $\sigma \to \sigma'$ is simply a map $f \colon X \to Y$ with $\sigma' \circ f = \sigma$. From (iii) together with this fact we get 
$$
\test(\sigma') \circ f = \func{F}f \circ \test(\sigma' \circ f) = \func{F}f \circ \test(\sigma),
$$        
which shows that $f$ is an $\func{F}$-coalgebra morphism $\test(\sigma) \to \test(\sigma')$ as desired.

Similarly for (i)$\Rightarrow$(iii) note that any $f \colon X \to Y$ is an $\func{A}$-coalgebra morphism $\sigma' \circ f \to \sigma'$, thus by safety it is also an $\func{F}$-coalgebra morphism $\test(\sigma' \circ f) \to \test(\sigma')$, which means $\func{F}f \circ \test(\sigma' \circ f) = \test(\sigma') \circ f$ as desired.  
\end{proof} 

Using Lemma~\ref{lemma:SafetyTestOperations}, one can show that all tests of Example~\ref{example:TestOperations} are safe, alternatively this can also be obtained as a consequence of Proposition~\ref{proposition:ReducibleTestsAreSafe} later on.

Them main result of this section is the following general compositionality property via safety for coalgebraic dynamic logics. 

\begin{prop}\label{proposition:CoalgMorphismInvariance}
Let all operations in $\CoOp$ and all tests in $\TeOp$ be safe and let $\gamma \colon \Act \to (\func{F}X)^X$ and $\gamma' \colon \Act \to (\func{F}Y)^Y$ be standard models. If $f \colon \gamma_{a_0} \to \gamma'_{a_0}$ is a coalgebra morphism for all atomic actions $a_0\in\AtAct$ and $\sem{p}' \circ f = \sem{p}$ for all propositional variables $p \in \Prop$, then $f \colon \gamma_{a} \to \gamma'_{a}$ is a coalgebra morphism for \emph{all} actions $a \in \Act$ and $\sem{\varphi}' \circ f = \sem{\varphi}$ for \emph{all} formulas $\varphi \in \Formu$. 
\end{prop}

\begin{proof}
By simultaneous induction on the structure of the action $a \in \Act$ and the formula $\varphi \in \Formu$, we show that $f$ is a coalgebra morphism $\gamma_a \to \gamma'_a$ and that $\sem{\varphi}' \circ f = \sem{\varphi}$ holds.  

If $a = O(a_1, \dots, a_n)$ is obtained using some $n$-ary coalgebra operation, then safety of $O$ together with the inductive hypothesis that $f$ is a coalgebra morphism $\gamma_{a_i} \to \gamma'_{a_i}$ for all $i = 1, \dots, n$ immediately yields that it also is a coalgebra morphism $\gamma_{O(a_1, \dots, a_n)} \to \gamma'_{O(a_1, \dots, a_n)}$. If $a = \test(\psi)$ is obtained via some test, then the inductive hypothesis $\sem{\psi}' \circ f = \sem{\psi}$ together with safety of $\test$ yields 
$$
\func{F}f \circ \gamma_a = \func{F}f \circ \test(\sem{\psi}) = \func{F}f \circ \test(\sem{\psi}' \circ f) = \test(\sem{\psi}') \circ f = \gamma'_a \circ f
$$ 
as desired.

If $\varphi = o(\varphi_1, \dots, \varphi_m)$ is obtained using some $m$-ary algebraic operation $o \in \sign(\alg{A})$, using the inductive hypothesis $\sem{\varphi_j}' \circ f = \sem{\varphi_j}$ for all $j = 1, \dots , m$ we simply compute 
$$
o^\pw\big( \sem{\varphi_1}, \dots, \sem{\varphi_m} \big) = o^\pw\big( \sem{\varphi_1}' \circ f, \dots, \sem{\varphi_m}' \circ f \big) = o^\pw\big( \sem{\varphi_1}', \dots, \sem{\varphi_m}' \big) \circ f
$$  
as desired. 

If $\varphi = \diam{a}^\lambda(\varphi_1, \dots, \varphi_k)$ for a $k$-ary predicate lifting $\lambda \in \Pred$, we use the inductive hypothesis $f\colon \gamma_a \to \gamma'_a$ and $\sem{\varphi_j}' \circ f = \sem{\varphi}$ together with naturality of $\lambda$ to find
$$
\lambda_X\big( \sem{\varphi_1}' \circ f, \dots, \sem{\varphi_k}' \circ f \big) \circ \gamma_a = \lambda_Y\big( \sem{\varphi_1}', \dots, \sem{\varphi_k}\big) \circ \func{F}f \circ  \gamma_a = \lambda_Y\big( \sem{\varphi_1}', \dots, \sem{\varphi_k}\big) \circ \gamma'_a \circ f
$$
as desired, finishing the proof. 
\end{proof}

Since both \emph{bisimulation} and \emph{behavioural equivalence} are defined via joint coalgebra morphisms (see, \emph{e.g.}, \cite[Definitions 3.5.1 \& 1.3.2]{Pattinson2003a}), the following is an immediate consequence of Proposition~\ref{proposition:CoalgMorphismInvariance}. 

\begin{cor}\label{corollary:BisimSafety}
For models $\gamma$ and $\gamma'$ as in Proposition~\ref{proposition:CoalgMorphismInvariance}, if $x\in X$ and $y\in Y$ are bisimilar (behaviourally equivalent) for \emph{atomic} actions, then they are bisimilar (behaviourally equivalent) for \emph{all} actions and $\sem{\varphi}(x) = \sem{\varphi}'(y)$ holds for all formulas $\phi \in \Formu$.   
\end{cor} 

Regarding \emph{expressivity} of many-valued coalgebraic logics, in \cite[Theorem 3]{BilkovaDostal2016} it is shown that modal equivalence implies bisimilarity for finitary, weak-pullback preserving $\func{F}_\omega$, given that $\Pred$ is \emph{$\alg{A}$-separating} \cite[Definition 4]{BilkovaDostal2016}, meaning that $t_1 \neq t_2$ in $\func{F}(A^n)$ implies $\lambda_{A^n}(\sigma)(t_1) \neq \lambda_{A^n}(\sigma)(t_2)$ for some $\lambda \in \Lambda$ and some \emph{term-definable} $\sigma\colon A^n \to A$ . Similarly, under those circumstances one may show that in a standard model $\gamma \colon \Act \to (\func{F}X)^X$ arising from $\gamma^0 \colon \AtAct \to (\func{F}_\omega X)^X$, modal equivalence implies bisimilarity with respect to all atomic actions, and therefore with respect to all actions given that every operation and test is safe. 
Let us mention that, in the many-valued setting, \emph{behavioural distances} rather than bisimilarity are often more interesting (many-valued modal characterisations for behavioural distances can for instance be found in \cite{KonigMichalski2018,WildSchroeder2022:char-log-beh-hemimetrics}). Studying these in the context of many-valued coalgebraic dynamic logics is left open for future research (also see Section~\ref{section:Conclusion}).

\section{Reducible operations and soundness}\label{section:ReducibleOperations}

Recall from PDL that non-deterministic choice, sequential composition and tests can be soundly reduced to their constituents in a uniform, term-definable manner due to the following validities  \cite{FischerLadner79}.  
\begin{align*}
[a \cup b]\varphi &\leftrightarrow [a]\varphi \wedge [b]\varphi  &  \diam{a \cup b}\varphi &\leftrightarrow \diam{a}\varphi \vee \diam{b}\varphi\\
[a{;}b]\varphi &\leftrightarrow [a][b]\varphi & \diam{a;b}\varphi &\leftrightarrow \diam{a}\diam{b}\varphi  \\
\boxm{\psi?}\varphi &\leftrightarrow(\psi \rightarrow \varphi) & \diam{\psi?}\varphi & \leftrightarrow (\psi \wedge \varphi)
\end{align*}

In this section, we investigate the special class of such \emph{reducible} operations and tests in our general coalgebraic framework (iteration is not reducible, hence excluded from the analysis below). This forms the basis of our uniform strong completeness proof in Section~\ref{section:Completeness} and we expect it to also be relevant for future research on coalgebraic dynamic logic (\emph{e.g.}, in order to provide a general definition of \emph{Fischer-Ladner closure} \cite[Section 6.1]{HarelKozTiu00:dyna-book}).     

\begin{rem}\label{remark:Section4Intro}
If all operations and tests are reducible then all dynamic formulas are provably equivalent to non-dynamic formulas, which implies that completeness of the dynamic logic is directly obtained from completeness of the underlying non-dynamic logic. For the purpose of completeness, the material in this section may seem needlessly elaborate. We include it in order to provide a formal definition of reducibility, proofs of soundness (Theorem~\ref{proposition:soundness-ReductionAxioms} \& Proposition~\ref{proposition:SoundnessTests}) and safety (Propositions~\ref{proposition:ReducibleImpliesSafe} \& \ref{proposition:ReducibleTestsAreSafe}).
Moreover, this section contains extensive examples (also see Subsection~\ref{subsection:example-applications}). \hfill $\blacksquare$ 
\end{rem}

\subsection{Templates and reducible coalgebra operations}\label{subsection:ReducingCoOps}
Intuitively speaking, we call an $n$-ary coalgebra operation $O$ \emph{reducible} with respect to a predicate lifting $\lambda$ if $\widehat{\lambda} \circ O(\gamma_{a_1}, \dots, \gamma_{a_n})$ can be directly obtained from the logical information of $\{ \widehat{\lambda'} \circ \gamma_{a_j} \mid \lambda' \in \Pred, j = 1,\dots, n \}$. To make this precise, we introduce the following.

\begin{defi}[\textsf{Templates}]\label{definition:templates}
The collection of \emph{$(n,k)$-templates} (where $n,k \geq 1$), denoted $\nkTempl$, is the set of formulas over $n$ modalities and $k$ variables inductively defined by 
\begin{align*}
&& \tau \ \Coloneqq \ \omega_i  
&\text{ $\big\vert$ } o(\tau_1, \dots, \tau_{\arity(o)})\phantom{.} 
\text{ $\big\vert$ } \diam{j}^\lambda (\tau_1,\dots,\tau_{\arity(\lambda)}) 
\end{align*} 
Here, the variables $\omega_i \in \{\omega_1, \dots, \omega_k\}$ are `\emph{placeholders for formulas}' and the indices $j \in \{ 1, \dots, n \}$ are `\emph{placeholders for actions}' (and, as usual, $o \in \mathsf{sign}(\alg{A})$ and $\lambda\in\Pred$). 
\end{defi}

Given some $(n,k)$-template $\tau$ together with $n$-many actions $\vec{a} = (a_1, \dots, a_n)$ and $k$-many formulas $\vec{\varphi} = (\varphi_1, \dots, \varphi_k)$, there is a corresponding formula $\tau[\vec{a}, \vec{\varphi}] \in \Formu$ obtained by substituting all occurrences of $\omega_i$ for $\varphi_i$ and all occurrences of $\diam{j}$ for $\diam{a_j}$. 

Similarly, from an $(n,k)$-template $\tau$ together with $\vec{\gamma} = (\gamma_1, \dots, \gamma_n) \colon X \to \func{F}^n X$ and $\vec{\sigma} = (\sigma_1, \dots, \sigma_k) \in (A^X)^k$, we inductively define $\tau[\vec{\gamma}, \vec{\sigma}] \in A^X$ via the rules  
\begin{align*}
&& \omega_i[\vec{\gamma}, \vec{\sigma}] & \ = \ \sigma_i,\\
&& o(\tau_1, \dots, \tau_{\arity(o)})[\vec{\gamma}, \vec{\sigma}] & \ = \ o^{\mathrm{pw}}\big(\tau_1[\vec{\gamma}, \vec{\sigma}], \dots, \tau_{\arity(o)}[\vec{\gamma}, \vec{\sigma}]\big),\\
&& \diam{j}^\lambda\big(\tau_1,\dots,\tau_{\arity(\lambda)})[\vec{\gamma}, \vec{\sigma}] & \ = \ \lambda_X\big(\tau_1[\vec{\gamma}, \vec{\sigma}], \dots, \tau_{\arity(\lambda)}[\vec{\gamma}, \vec{\sigma}]\big) \circ \gamma_j.
\end{align*}
This allows us to abstractly define the concept of reducibility as follows.     

\begin{defi}[\textsf{Reducible coalgebra operations}]\label{definition:ReducibleCoalgOperation}
Let $O\in \CoOp$ be an $n$-ary coalgebra operation and $\lambda \in \Pred$ be a $k$-ary predicate lifting. We call $O$ \emph{reducible with respect to $\lambda$} if there exists an $(n, k)$-template $\tau \in \nkTempl$ such that 
$$
\lambda_X(\vec{\sigma}) \circ O(\vec{\gamma}) = \tau[\vec{\gamma}, \vec{\sigma}]
$$   
holds for all $\vec{\gamma}\colon X {\to} \func{F}^n X$ and $\vec{\sigma} \in  (A^X)^k $. We simply call $O$ \emph{reducible} if it is reducible with respect to \emph{every} $\lambda \in \Pred$. 
\end{defi} 

By design, a template witnessing reducibility yields a sound \emph{reduction axiom} as follows.

\begin{prop}[\textsf{Soundness}]\label{proposition:soundness-ReductionAxioms}
Let $\tau \in \nkTempl$ be a $(n,k)$-template witnessing that $O$ is reducible with respect to $\lambda$. Then the formula
$$
\diam{ O(a_1, \dots, a_n) }^\lambda (\varphi_1, \dots, \varphi_k) \leftrightarrow \tau[\vec{a}, \vec{\varphi}]
$$ 
is true at every state of any standard model.  
\end{prop} 

\begin{proof}
This follows immediately from Definition~\ref{definition:ReducibleCoalgOperation} once we show that in every standard model $\gamma$, for all $(n,k)$-templates $\tau$ it holds that 
$$
\sem{\tau[\vec{a},\vec{\varphi}] } = \tau[\vec{\gamma}_a,\sem{\vec{\varphi}}],
$$ 
where $\vec{\gamma}_a = (\gamma_{a_1}, \dots, \gamma_{a_n})$ and $\sem{\vec{\varphi}} = (\sem{\varphi_1}, \dots, \sem{\varphi_k})$. We proceed to show this by induction on the structure of the template $\tau$. 

In the case $\tau = \omega_i$, we directly get $\sem{\tau[\vec{a}, \vec{\varphi}]} = \sem{\varphi_i} = \tau[\vec{\gamma}_a, \sem{\vec{\varphi}}]$. For the case $\tau = o(\tau_1, \dots, \tau_{\arity(o)})$ we use the inductive hypothesis $\sem{\tau_i[\vec{a},\vec{\varphi}] } = \tau_i[\vec{\gamma}_a,\sem{\vec{\varphi}}]$ for all $i = 1, \dots, \arity(o)$ to similarly get 
\begin{align*}
\sem{\tau[\vec{a}, \vec{\varphi}]} & = \sem{o(\tau_1[\vec{a}, \vec{\varphi}], \dots, \tau_{\arity(o)}[\vec{a}, \vec{\varphi}])}  \\
& = o^\pw\big(\sem{\tau_1[\vec{a}, \vec{\varphi}]}, \dots, \sem{\tau_{\arity(o)}[\vec{a}, \vec{\varphi}]}\big) \\
& = o^\pw\big(\tau_1[\vec{\gamma}_a,\sem{\vec{\varphi}}], \dots, \tau_{\arity(o)}[\vec{\gamma}_a,\sem{\vec{\varphi}}]\big)= \tau[\vec{\gamma}_a, \sem{\vec{\varphi}}].
\end{align*}          
Lastly, for the case $\tau = \diam{j}^\lambda(\tau_1, \dots, \tau_{\arity(\lambda)})$ we use the inductive hypothesis $\sem{\tau_i[\vec{a},\vec{\varphi}] } = \tau_i[\vec{\gamma}_a,\sem{\vec{\varphi}}]$ for all $i = 1, \dots, \arity(\lambda)$ and compute 
\begin{align*}
\sem{\tau[\vec{a}, \vec{\varphi}]} & = \sem{\diam{a_j}(\tau_1[\vec{a},\vec{\varphi}], \dots, \tau_{\arity(\lambda)}[\vec{a},\vec{\varphi})]}  \\
& = \lambda_X\big(\sem{\tau_1[\vec{a},\vec{\varphi}]}, \dots, \sem{\tau_{\arity(\lambda)}[\vec{a},\vec{\varphi})]}\big) \circ \gamma_{a_j} \\
& = \lambda_X\big(\tau_1[\vec{\gamma}_a,\sem{\vec{\varphi}}], \dots, \tau_{\arity(\lambda)}[\vec{\gamma}_a,\sem{\vec{\varphi}}]\big) \circ \gamma_{a_j} = \tau[\vec{\gamma}_a,\sem{\vec{\varphi}}]
\end{align*}          
as desired, finishing the proof.
\end{proof}

Recall that in Section~\ref{section:SafetyAndInvariance} we discussed the importance of \emph{safe} coalgebra operations (Definition~\ref{definition:SafeCoalgebraOperation}). In the following, we show that reducibility implies safety under the (mild) assumption that the collection $\Pred$ of predicate liftings is \emph{separating} (recall Definition~\ref{definition:SeparatingPredLifts}).

\begin{prop}\label{proposition:ReducibleImpliesSafe}
Let the collection of predicate liftings $\Pred$ be separating and let $O$ be a coalgebra operation. If $O$ is reducible, then it is safe.
\end{prop} 

\begin{proof}
Let $\arity(O) =: n$ and let $\vec{\gamma} \colon X \to \func{F}^n X$ and $\vec{\gamma}' \colon Y \to \func{F}^n Y$ be coalgebras. We want to show that 
if $f \colon \vec{\gamma} \to \vec{\gamma}'$ is a joint coalgebra morphism, then it is also a coalgebra morphism $f\colon O(\vec{\gamma}) \to O(\vec{\gamma}')$.  
Since $\Pred$ is separating, for the latter it suffices to show that 
$$
\widehat{\lambda}_Y \circ \func{F}f \circ O(\vec{\gamma}) = \widehat{\lambda}_Y \circ O(\vec{\gamma}') \circ f
$$ 
holds for all $\lambda \in \Pred$. 

Let $\arity(\lambda) =: k$ and choose some $(n,k)$-template $\tau \in \nkTempl$ which witnesses that $O$ is reducible with respect to $\lambda$. Starting on the left-hand side of the above equation, for every $x \in X$ and $\sigma_1, \dots, \sigma_k \in A^Y$, we get  
\begin{align*}
\big(\widehat{\lambda}_Y \circ \func{F}f \circ O(\vec{\gamma})\big)(x)(\vec{\sigma}) &= \big(\kNA f \circ \widehat{\lambda}_X \circ O(\vec{\gamma})\big)(x)(\vec{\sigma}) && \text{(Naturality $\widehat{\lambda}$)} \\
&= \big(\widehat{\lambda}_X \circ O(\vec{\gamma})\big)(x)(\vec{\sigma} \circ f^k) && \text{(Definition $\kNA f$)}\\ 
& = \big(\lambda_X(\vec{\sigma} \circ f^k) \circ O(\vec{\gamma})\big)(x) = \tau[\vec{\gamma}, \vec{\sigma} \circ f^k](x). &&\text{(Reducibility)}
\end{align*}
Starting on the right-hand side, we use reducibility as well to get
$$
\big(\widehat{\lambda}_Y \circ O(\vec{\gamma}') \circ f\big) (x)(\vec{\sigma}) = \big(\lambda_Y(\vec{\sigma}) \circ O(\vec{\gamma}')\big) \big( f (x)\big) = \tau[\vec{\gamma}', \vec{\sigma}]\big(f(x)\big).
$$
Thus, we are done once we prove $\tau[\vec{\gamma}, \vec{\sigma} \circ f^k] = \tau[\vec{\gamma}', \vec{\sigma}]\circ f$  for all $\tau \in \nkTempl$, which we do by induction on the structure of the template $\tau$. The cases $\tau = \omega_i$ and $\tau = o(\tau_1, \dots, \tau_{\arity(o)})$ are straightforward. For the case $\tau = \diam{j}^{\lambda'}(\tau_1, \dots, \tau_{\arity(\lambda')})$, we compute
\begin{align*}
\tau[\vec{\gamma}, \vec{\sigma} \circ f^k] &= \lambda'_X\big(\tau_1[\vec{\gamma}, \vec{\sigma} \circ f^k], \dots, \tau_{\arity(\lambda')}[\vec{\gamma}, \vec{\sigma} \circ f^k]\big) \circ \gamma_j && \text{(Definition)}\\
& = \lambda'_X\big(\tau_1[\vec{\gamma}', \vec{\sigma}]\circ f, \dots, \tau_{\arity(\lambda')}[\vec{\gamma}', \vec{\sigma}]\circ f\big) \circ \gamma_j && \text{(Ind.Hyp. on $\tau_i$'s)}\\ 
& = \lambda'_Y\big(\tau_1[\vec{\gamma}', \vec{\sigma}], \dots, \tau_{\arity(\lambda')}[\vec{\gamma}', \vec{\sigma}]\big) \circ \func{F}f \circ \gamma_j && \text{(Naturality $\lambda$)}\\ 
& = \lambda'_Y\big(\tau_1[\vec{\gamma}', \vec{\sigma}], \dots, \tau_{\arity(\lambda')}[\vec{\gamma}', \vec{\sigma}]\big) \circ \gamma'_j \circ f = \tau[\vec{\gamma}', \vec{\sigma}] \circ f && \text{($f \colon \gamma_j \to \gamma_j'$)}
\end{align*}
as desired, finishing the proof.  
\end{proof}

The converse of the above proposition is not necessarily true. For example, iteration (\emph{e.g.}, of PDL or game logic) is safe but not reducible. As first examples of reducible operations we consider \emph{sequential compositions}. For Kleisli composition (Example~\ref{example:CoalgebraOperations-Kleisli}), the following is already shown in \cite{HansenKupkeLeal2014}, the generalisation from $\alg{2}$ to $\alg{A}$ being straightforward.

\begin{exa}\label{example:ReductionKleisliMonadMorphism}
Let $\func{F}$ be a monad and $\lambda \in \Pred$ be a unary predicate lifting such that $\widehat{\lambda} \colon \func{F} \to \NA$ is a \emph{monad morphism}. As in \cite[Lemma 10]{HansenKupkeLeal2014}, one shows 
$
\lambda_X (\sigma) \circ (\gamma_1 {;} \gamma_2) = \lambda_X(\lambda_X(\sigma) \circ \gamma_2) \circ \gamma_1,
$
in other words, Kleisli composition $;$ is reducible with respect to $\lambda$ via the $(2,1)$-template $\tau = \diam{1}^\lambda \diam{2}^\lambda \omega$. Thus, the familiar reduction axiom 
$$
\diam{a_1 {;} a_2}^\lambda \varphi \leftrightarrow \diam{a_1}^\lambda \diam{a_2}^\lambda \varphi
$$
holds in all standard models. \hfill $\blacksquare$
\end{exa} 

The above applies to the coalgebraic dynamic logics of Examples~\ref{example:CoalgebraicDynamicLogics-PDL-Crisp}, \ref{example:CoalgebraicDynamicLogics-PDL-Labelled1} \& \ref{example:CoalgebraicDynamicLogics-GL}. The next example shows that our present framework also accommodates a template for $;$ in Example~\ref{example:CoalgebraicDynamicLogics-PDL-Labelled2}, where $\func{F}$ still is a monad, but the predicate liftings $\lambda^r$ do \emph{not} induce monad morphisms.

\begin{exa}\label{example:ReductionKleisli2-FiniteLinear}
Consider the $\alg{2}$-valued coalgebraic dynamic logic for $\PA$ with $\Pred = \{ \lambda^r\}_{ r\in A}$. If $\alg{A}$ is \emph{finite and linear}, we have 
$
\lambda_X^r(S)\circ (\gamma_1{;}\gamma_2) = \bigvee_{r_1 \odot r_2 \geq r}\lambda_X^{r_1}(\lambda_X^{r_2}(S) \circ \gamma_2) \circ \gamma_1,
$
in other words, the $(2,1)$-template $\tau = \bigvee_{r_1 \odot r_2 \geq r} \diam{1}^{r_1} \diam{2}^{r_2} \omega$ witnesses Kleisli composition $;$ being reducible with respect to $\lambda^r$. Thus, the corresponding reduction axioms 
$$
\diam{a_1 {;} a_2}^{r} \varphi \leftrightarrow \bigvee_{r_1 \odot r_2 \geq r} \diam{a_1}^{r_1} \diam{a_2}^{r_2} \varphi
$$  
hold in all standard models (note that since $\alg{A}$ is finite, the right-hand sides of the reduction axioms are all well-defined formulas). 

For the proof of this fact, first we note that linearity and finiteness of $\alg{A}$ yield that $\bigvee _{i\in I} s_i \geq s \Leftrightarrow \exists i\in I \colon s_i \geq s$ always holds. Let $S \in 2^X$ and $\gamma_1, \gamma_2 \colon X \to A^X$ be $\PA$-coalgebras. Then we have the chain of equivalences 
\begin{align*}
&\lambda_X^r(S)\circ (\gamma_1{;}\gamma_2)(x) = 1 \Longleftrightarrow \bigvee_{z\in S} \gamma_1 {;} \gamma_2 (x)(z) \geq r &&\text{(Def. of $\lambda^r$)}\\ 
& \Longleftrightarrow \bigvee_{z\in S} \bigvee_{y \in X} \gamma_1(x)(y) \odot \gamma_2(y)(z) \geq r &&\text{(Def. of $;$ )}\\
& \Longleftrightarrow \exists z\in S, y\in X : \gamma_1(x)(y) \odot \gamma_2(y)(z) \geq r &&\text{(Linearity of $\alg{A}$)}\\
& \Longleftrightarrow \exists z\in S, y\in X, r_1,r_2 \in A  : r_1 \odot r_2 \geq r, \gamma_1(x)(y) \geq r_1, \gamma_2(y)(z) \geq r_2 &&\text{($\odot$ is monotone)}\\
& \Longleftrightarrow \exists y\in X, r_1, r_2 \in A: r_1 \odot r_2 \geq r,  \gamma_1(x)(y) \geq r_1, \bigvee_{z \in S} \gamma_2(y)(z) \geq r_2 &&\text{(Linearity of $\alg{A}$)}\\
& \Longleftrightarrow \exists y\in X, r_1, r_2 \in A : r_1 \odot r_2 \geq r, \gamma_1(x)(y) \geq r_1, \lambda_X^{r_2}(S)(\gamma_2(y)) = 1 &&\text{(Def. of $\lambda^{r_2}$)}\\
& \Longleftrightarrow \exists r_1, r_2 \in A : r_1 \odot r_2 \geq r, \bigvee_{\lambda_X^{r_2}(S)(\gamma_2(y)) = 1} \gamma_1(x)(y) \geq r_1 &&\text{(Linearity of $\alg{A}$)}\\
&\Longleftrightarrow \exists r_1, r_2 \in A : r_1 \odot r_2 \geq r, \lambda_X^{r_1}(\lambda_X^{r_2}(Z) \circ \gamma_2) (\gamma_1(x)) = 1 &&\text{(Def. of $\lambda^{r_1}$)}\\ 
&\Longleftrightarrow\bigvee_{r_1 \odot r_2 \geq r}\lambda_X^{r_1}(\lambda_X^{r_2}(S) \circ \gamma_2) \circ \gamma_1(x) = 1. &&\text{(Linearity of $\alg{A})$}
\end{align*}
This, as desired, shows that the template $\tau = \bigvee_{r_1 \odot r_2 \geq r} \diam{1}^{r_1} \diam{2}^{r_2} \omega$ witnesses that $;$ is reducible with respect to $\lambda^r$. 
\hfill $\blacksquare$
\end{exa}

For an example where $\func{F}$ is \emph{not} a monad, we consider Instantial PDL of Example~\ref{example:CoalgebraicDynamicLogics-IPDL}. The reduction axioms used here are due to \cite[Section 3.3]{vanBenthemBezhanishviliEnqvist2019}. 

\begin{exa}\label{example:ReductionCompositionInstantialNBH}
In Instantial PDL, the sequential composition $;$ is reducible with respect to every $\lambda^{k+1}$ via the $(2,k+1)$-template 
$$
\diam{1}^{k+1} \big(\diam{2}^2(\omega_1, \omega_{k+1}), \dots, \diam{2}^2(\omega_k, \omega_{k+1}), \diam{2}^1\omega_{k+1} \big)
$$
and the (alternative) sequential composition $\star$ defined in Example~\ref{example:Generalcompositions} is reducible with respect to $\lambda^{k+1}$ via the template
$$
\diam{1}^2\big( \diam{2}^{k+1}(\omega_1, \dots, \omega_k, \omega_{k+1}), \top \big).
$$
Thus, the corresponding reduction axioms 
$$
\diam{a_1 {;} a_2}^{k+1}(\varphi_1, \dots, \varphi_k, \varphi) \leftrightarrow \diam{a_1}^{k+1} \big(\diam{a_2}^2(\varphi_1, \varphi), \dots, \diam{a_2}^2(\varphi_k, \varphi), \diam{a_2}^1\varphi \big)
$$
and 
$$
\diam{a_1 \star a_2}^{k+1}(\varphi_1, \dots, \varphi_k, \varphi) \leftrightarrow \diam{a_1}^2\big( \diam{a_2}^{k+1}(\varphi_1, \dots, \varphi_k, \varphi), \top \big) 
$$
hold in all standard models.

Reducibility for $;$ is shown similarly to the proof of \cite[Proposition 3]{vanBenthemBezhanishviliEnqvist2019} as follows.
Untangling the definitions, for $S_1, \dots, S_k, S \in 2^X$ and $\gamma_1,\gamma_2 \colon X \to \mathcal{P}\mathcal{P}X$ we get for the left-hand side 
\begin{align*}
&\lambda^{k+1}_X(S_1, \dots, S_k, S)\big(\gamma_1 {;} \gamma_2(x)\big) = 1 \\ 
&\Longleftrightarrow \exists Z \in \gamma_1{;}\gamma_2(x) \colon Z \subseteq S \wedge \forall i\colon Z \cap S_i \neq \varnothing  \\
&\Longleftrightarrow \exists Z' \in \gamma_1(x), G\subseteq \mathcal{P}X:G \subseteq \bigcup_{x'\in Z'} \gamma_2(x') \wedge \bigcup G \subseteq S \wedge \forall i \colon \bigcup G \cap S_i \neq \varnothing,
\end{align*}
and for the right-hand side 
\begin{align*}
& \lambda_X^{k+1}\big(\lambda_2(S_1, S) \circ \gamma_2, \dots, \lambda_2(S_k, S) \circ \gamma_2, \lambda_X^1(S) \circ \gamma_2\big)\big(\gamma_1(x)\big) = 1 \\
&\Longleftrightarrow \exists Z' \in \gamma_1(x) \colon Z' \subseteq \lambda_X^1(S) \circ \gamma_2 \wedge \forall i \colon Z' \cap \lambda^2_X(S_i, S) \circ \gamma_2 \neq \varnothing \\
& \Longleftrightarrow  \exists Z' \in \gamma_1(x) \colon \forall x' \in Z' \colon \exists U \in \gamma_2(x') \colon U \subseteq S \wedge\\
& \phantom{\Longleftrightarrow \exists Z' \in \gamma_1(x) \colon i} \forall i \colon \exists x_i' \in Z', U_i \in \gamma_2(x_i') \colon U_i \cap S_i \neq \varnothing. 
\end{align*}
The former immediately implies the latter with the same $Z'$. The converse is obtained setting $G = \{ U \mid U \subseteq S \wedge \exists y \in Z' \colon U \in \gamma_2(y) \}$ as in \cite[Proposition 3]{vanBenthemBezhanishviliEnqvist2019}. 

To see reducibility for $\star$, we get for the left-hand side 
\begin{align*}
&\lambda^{k+1}_X(S_1, \dots, S_k, S)\big(\gamma_1 \star \gamma_2(x)\big) = 1 \\ 
&\Longleftrightarrow \exists Z \in \gamma_1 \star \gamma_2(x) \colon Z \subseteq S \wedge \forall i\colon Z \cap S_i \neq \varnothing  \\
&\Longleftrightarrow \exists Z' \in \gamma_1(x), y \in Z', Z\in \gamma_2(y) \colon Z \subseteq S \wedge \forall i\colon Z \cap S_i \neq \varnothing, 
\end{align*}
and for the right-hand side
\begin{align*}
&\lambda^{2}_X\big(\lambda_X^{k+1}(S_1, \dots, S_k, S)\circ \gamma_2\big)\big(\gamma_1(x)\big) = 1 \\ 
&\Longleftrightarrow \exists Z' \in \gamma_1(x) \colon Z' \cap \lambda_X^{k+1}(S_1, \dots, S_k, S)\circ \gamma_2 \neq \varnothing   \\
&\Longleftrightarrow \exists Z' \in \gamma_1(x), y \in Z', Z\in \gamma_2(y) \colon Z \subseteq S \wedge \forall i\colon Z \cap S_i \neq \varnothing, 
\end{align*}
which finishes the proof. \hfill $\blacksquare$
\end{exa}

Other reducible operations, such as non-deterministic choice $\cup, \vee$ or dual $\dual$, are treated in the following subsection, as they are instances of \emph{independently reducible operations}. 
\subsection{Independent reducibility}\label{subsection:IndependentlyReducibleExamples}

Most remaining coalgebra operations discussed in Subsection~\ref{subsection:Coalg/Test-Operations}, in particular the ones induced by natural transformations, fall under the scope of the following `well-behaved' form of reducibility. 

\begin{defi}[\textsf{Independent reducibility}]\label{definition:independently-reducible}
Let $O$ be an $n$-ary coalgebra operation and let $\lambda$ be a $k$-ary predicate lifting. We call $O$ \emph{independently reducible with respect to $\lambda$} if $O$ is reducible with respect to $\lambda$ and this can be witnessed by a $(n,k)$-template in which no $\diam{j}^{\lambda'}$ with $\lambda' \neq \lambda$ occurs.    
\end{defi}

Clearly this notion only becomes relevant if $|\Pred| \geq 2$. Note that in Examples \ref{example:ReductionKleisli2-FiniteLinear} \& \ref{example:ReductionCompositionInstantialNBH} we already encountered operations which are reducible but not independently reducible.

Independent reducibility allows us to shift our attention towards coalgebra operations for the functor $\kNA$ and $\lambda^\ev$ which takes evaluations $\ev_{(\sigma_1, \dots, \sigma_k)} \colon A^{(A^X)^k} \to A$ similarly to Example~\ref{example:PredicateLiftings}(4). 

\begin{prop}\label{proposition:IndependentReducibility-neighbourhood}
An $n$-ary coalgebra operation $O \in \CoOp$ is independently reducible with respect to a $k$-ary predicate lifting $\lambda \in \Pred$ if and only if there is an $n$-ary coalgebra operation $\Omega$ for $\kNA$ that is independently reducible with respect to $\lambda^\ev$ and satisfies 
$\widehat{\lambda} \circ O = \Omega \circ \widehat{\lambda}^n$ (\emph{i.e.}, the diagram below commutes).
$$ 
\begin{tikzcd}[row sep = large, column sep = large]
\Coalg{\func{F}^n} \arrow[r, "O"] \arrow[d, "\widehat{\lambda}^n \circ (\text{--})"']
&\Coalg{\func{F}} \arrow[d, "\widehat{\lambda} \circ (\text{--})"] \\
  \Coalg{\kNA^n} \arrow[r, "\Omega"']
& \Coalg{\kNA}
\end{tikzcd}  
$$ 
Reducibility of $O$ and $\Omega$ is witnessed by the same templates (up to swapping $\diam{j}^\lambda$ and $\diam{j}^{\lambda^\ev}$). 
\end{prop}

\begin{proof}
`($\Rightarrow$)': Let $\tau$ be a $(n,k)$-reduction template in which no $\diam{j}^{\lambda'}$ with $\lambda' \neq \lambda$ occurs (as required in the definition of independent reducibility), and let $\tau^\ev$ denote the corresponding template replacing all $\diam{j}^\lambda$ by $\diam{j}^{\lambda^\ev}$. Note that $\tau^\ev$ defines itself a $n$-ary coalgebra operation $\Omega^\tau \colon \Coalg{\kNA^n} \to \Coalg{\kNA}$ via $\Omega^\tau(\vec{\xi})(x)(\vec{\sigma}) = \tau^\ev[\vec{\xi}, \vec{\sigma}](x)$. 

To show that the above diagram commutes, we need to show that for all $\vec{\gamma} \colon X \to \func{F}^nX$ and $\vec{\sigma}\in (A^X)^k$ it holds that 
$$
\tau[\vec{\gamma}, \vec{\sigma}] = \tau^\ev[\widehat{\lambda}^n_X \circ \vec{\gamma}, \vec{\sigma}],
$$     
which we prove by induction on the structure of the template $\tau$. The cases $\tau = \omega_i$ and $\tau = o(\tau_1, \dots, \tau_{\arity(o)})$ are straightforward by definitions, thus let $\tau = \diam{j}^\lambda(\tau_1, \dots, \tau_k)$ and suppose the statement already holds for $\tau_1, \dots, \tau_k$. We compute 
\begin{align*}
\tau[\vec{\gamma}, \vec{\sigma}](x) &= \lambda_X\big(\tau_1[\vec{\gamma}, \vec{\sigma}], \dots, \tau_k[\vec{\gamma}, \vec{\sigma}]\big) \big( \gamma_j (x)\big) \\
& = \lambda_X\big(\tau^\ev_1[\widehat{\lambda}_X^n \circ \vec{\gamma}, \vec{\sigma}], \dots, \tau^\ev_k[\widehat{\lambda}_X^n \circ \vec{\gamma}, \vec{\sigma}]\big) \big( \gamma_j (x) \big)\\
& = (\widehat{\lambda}_X\circ \gamma_j) (x)\big(\tau^\ev_1[\widehat{\lambda}_X^n \circ \vec{\gamma}, \vec{\sigma}], \dots, \tau^\ev_k[\widehat{\lambda}_X^n \circ \vec{\gamma}, \vec{\sigma}]\big)\\ 
& = \big( \lambda^\ev_X\big(\tau^\ev_1[\widehat{\lambda}_X^n \circ \vec{\gamma}, \vec{\sigma}], \dots, \tau^\ev_k[\widehat{\lambda}_X^n \circ \vec{\gamma}, \vec{\sigma}]\big) \circ \widehat{\lambda}_X \circ \gamma_j (x)\big) = \tau^\ev[\widehat{\lambda}^n_X \circ \vec{\gamma}, \vec{\sigma}](x)
\end{align*}
as desired.

`($\Leftarrow$)': Note that independent reducibility of $\Omega$ with respect to $\lambda^\ev$ yields
$$
\Omega(\vec{\xi})(x)(\vec{\sigma}) = \big(\lambda^\ev_X(\vec{\sigma}) \circ{O}(\vec{\xi})\big)(x) = \tau^\ev[\vec{\xi}, \vec{\sigma}]
$$
for some $(n,k)$-reduction template $\tau^\ev$. Let $\tau$ be the corresponding reduction template where all `$\lambda^\ev$' are replaced by `$\lambda$'. Since by assumption the above diagram commutes, we have 
$$
\widehat{\lambda}_X \circ O(\vec{\gamma}) = \Omega(\widehat{\lambda}_X^n \circ \vec{\gamma}),  
$$
and thus (the proof of the last equation was done in the previous part, and did not depend on the assumption of `($\Rightarrow$)')
$$
\big(\lambda_X(\vec{\sigma})\circ O(\gamma)\big)(x) = \big(\widehat{\lambda}_X\circ O(\gamma)\big)(x)(\vec{\sigma}) = \Omega(\widehat{\lambda}_X^n \circ \vec{\gamma})(x)(\vec{\sigma}) = \tau^\ev[\vec{\xi}, \vec{\sigma}] = \tau[\vec{\gamma}, \vec{\sigma}]
$$
shows independent reducibility of $O$ with respect to $\lambda$ via $\tau$, finishing the proof.
\end{proof}

The following examples illustrate that the reduction template for $O$ is usually obtained from a concrete description of the corresponding $\Omega$ in a straightforward manner. 

\begin{exa}\label{example:IndependentlyReducibleOperations}
The following are examples of independently reducible operations. 
The predicate liftings are as in Example~\ref{example:PredicateLiftings} and the coalgebra operations as in Example~\ref{example:CoalgebraOperations-Induced}.
\begin{enumerate}
\item Consider $\mathcal{P}$ with $O$ induced by $\cup \colon \mathcal{P}^2 \Rightarrow \mathcal{P}$. 
For $\lambda^\Diamond$, the corresponding $\Omega$ is induced by $\vee \colon \NA^2 \to \NA$, which is shown using the property $\bigvee_{i \in I} (r_i \vee s_i) = \bigvee_{i \in I} r_i \vee \bigvee_I s_i$ of the complete lattice $\alg{A}$. 
Hence, we find the familiar reduction axiom 
$$
\diam{a_1 \cup a_2 } \varphi \leftrightarrow \diam{a_1} \varphi \vee \diam{a_2} \varphi.
$$
Similarly, for $\lambda^\Box$ we find the reduction axiom 
$$
[a_1 \cup a_2] \varphi \leftrightarrow [a_1] \varphi \wedge [a_2] \varphi
$$
since the corresponding $\Omega$ is induced by $\wedge \colon \NA^2 \to \NA$.
\item Consider $\PA$ with $O$ induced by $\vee$. For $\lambda^\Diamond$ and $\lambda^\Box$ one can use the same $\Omega$ as in (1) and, therefore, the `same' reduction axioms
$$
\diam{a_1 \vee a_2 } \varphi \leftrightarrow \diam{a_1} \varphi \vee \diam{a_2} \varphi \quad \text{and} \quad [a_1 \vee a_2] \varphi \leftrightarrow [a_1] \varphi \wedge [a_2] \varphi
$$
are valid.  
This can be shown using that the following identities hold in $\FLew$-algebras (see, \emph{e.g.}, \cite[Lemma 2.6]{GalatosJipsen2007}): 
$
(r_1 \vee r_2) \odot s = (r_1 \odot s) \vee (r_2 \odot s) \text{ and } (r_1 \vee r_2) \rightarrow s = (r_1 \rightarrow s) \wedge (r_2 \rightarrow s). 
$
\item Consider again $\PA$ with $O$ induced by $\vee$. If $\alg{A}$ is \emph{linear}, then $O$ is independently reducible with respect to the ($\alg{2}$-valued) $\lambda^r$, with $\Omega$ again induced by $\vee$, due to the equivalence $\bigvee_{i \in I} s_i \vee s'_i \geq r \text{ iff } \bigvee_{i \in I} s_i \geq r \text{ or } \bigvee_{i \in I} s'_i \geq r$. Thus, we again get
$$
\diam{a_1 \vee a_2 }^r \varphi \leftrightarrow \diam{a_1}^r \varphi \vee \diam{a_2}^r \varphi
$$ 
as reduction axiom.
\item Consider $\MA$ with $O$ induced by dual $(\cdot)^\partial$. With $\lambda^\ev$, clearly $\Omega$ is itself induced by dual and we obtain the reduction axiom 
$$
\diam{a^\partial} \varphi \leftrightarrow \neg \diam{a} \neg \varphi.
$$
A similar argument works for $\vee$ and $\wedge$.  
\item Consider $\mathcal{P} \circ \mathcal{P}$ and $O$ induced by $\Cup$.
This operation is independently reducible with respect to $\lambda^{k+1}$ by taking $\Omega$ induced by $\theta_X^{k+1} \colon (k+1)\mathcal{N}^2 X \to (k+1)\mathcal{N} X$ with
$
\theta_X^{k+1}(N_1, N_2) \colon (S_1, \dots, S_k, S) \mapsto \bigvee_{K \subseteq \{ 1, \dots, k \} } N_1(S^K_1, \dots, S^K_k, S) \wedge N_2(S^{K^\mathrm{c}}_1, \dots, S^{K^\mathrm{c}}_k, S)   
$, 
where $S_i^K = S_i$ if $i \in K$, otherwise $S_i^K = X$ and $S_i^{K^\mathrm{c}}$ defined similarly for $K^\mathrm{c} \coloneq \{ 1, \dots, k \}{\setminus}K$ 
(shown as in \cite[Proposition 3]{vanBenthemBezhanishviliEnqvist2019}). 
Letting $\psi_i^K = \psi_i$ if $i \in K$, otherwise $\psi_i^K = \top$ and similarly with $\psi_i^{K^\mathrm{c}}$, this yields the corresponding reduction axiom $\diam{a_1 \Cup a_2}(\psi_1, \dots, \psi_k, \varphi) \leftrightarrow \bigvee_{K \subseteq \{ 1, \dots, k \} } \diam{a_1}(\psi^K_1, \dots, \psi^K_k, \varphi) \wedge \diam{a_2}(\psi^{K^\mathrm{c}}_1, \dots, \psi^{K^\mathrm{c}}_k, \varphi)$.  
\item As an example where $O$ and $\Omega$ are not induced operations, consider the \emph{counter-domain} $\sim$ for $\mathcal{P}$ of \autoref{example:CoalgebraOperations-CounterDomain}. In what follows, $\xi \colon X \to \NA X$, $x \in X$, $\sigma\in A^X$ and $0,1 \in A^X$ denote the corresponding constant maps. For $\lambda^\Box$, the operation $\sim$ is independently reducible via $\Omega(\xi)(x)(\sigma) = \xi(x)(0) \rightarrow \sigma(x)$ and for $\lambda^\Diamond$ it is independently reducible via $\Omega(\xi)(x)(\sigma) = \neg \xi(x)(1) \wedge \sigma(x)$, yielding the axioms 
$$
[{\sim}a] \varphi \leftrightarrow ([a]\bot \rightarrow \varphi) \quad \text{and} \quad \diam{{\sim}a}\varphi \leftrightarrow (\neg\diam{a}\top \wedge \varphi),
$$
respectively.
\item Similarly to the previous item, the \emph{counter-support} for $\PA$ (Example~\ref{example:CoalgebraOperations-CounterDomain}) is independently reducible with respect to $\lambda^\Box$ if the characteristic function $\chi_{ \{ 1 \} }$ is term-definable in $\alg{A}$ and independently reducible with respect to $\lambda^\Diamond$ if the characteristic function $\chi_{ \{ 0 \} }$ is term-definable in $\alg{A}$. The corresponding reduction axioms are then given by
$$
[{\sim}a] \varphi \leftrightarrow (\chi_{ \{ 1 \} }([a]\bot) \rightarrow \varphi) \quad \text{and} \quad \diam{{\sim}a}\varphi \leftrightarrow (\chi_{ \{ 0 \} }(\diam{a}\top) \wedge \varphi),
$$
respectively. The instantial counter-domain for $\mathcal{P} \circ \mathcal{P}$ (Example~\ref{example:CoalgebraOperations-CounterDomain}) is also independently reducible with respect to all predicate liftings $\lambda^{k+1}$ with 
$$
\diam{{\sim}a}^{k+1}(\varphi_1, \dots, \varphi_k, \varphi) \leftrightarrow (\neg\diam{a}^{k+1}(\top, \dots, \top) \wedge \varphi_1 \wedge \dots \wedge \varphi_k \wedge \varphi) 
$$
as corresponding reduction axioms. \hfill $\blacksquare$
\end{enumerate}
\end{exa} 

\begin{rem}[\textsf{Induced Operations and Relation to \cite{HansenKupkeLeal2014}}]\label{remark:relationTo2014}
Let $\lambda$ be unary and let $O = O^\theta$ be a coalgebra operation pointwise induced by a natural transformation $\theta \colon \func{F}^n \Rightarrow \func{F}$ (recall Example~\ref{example:CoalgebraOperations-Induced}). In order to find $\Omega$ as in Proposition~\ref{proposition:IndependentReducibility-neighbourhood} it suffices to find a natural transformation $\chi \colon \NA^n \Rightarrow \NA$ such that the following diagram commutes: 
$$ 
\begin{tikzcd}[row sep = large, column sep = large]
\func{F}^n \arrow[r, "\theta"] \arrow[d, "\widehat{\lambda}^n"']
&\func{F} \arrow[d, "\widehat{\lambda}"] \\
  \NA^n \arrow[r, "\chi"']
& \NA
\end{tikzcd}  
$$
This is the condition identified in \cite[Theorem 15]{HansenKupkeLeal2014} in the Boolean case $\alg{A} = \alg{2}$. Therein, it is also described how to systematically construct for every such $\chi$ a corresponding reduction axiom for $O^\theta$. For general $\alg{A}$, this is not possible if $\alg{A}$ is not functionally complete, \emph{i.e.}, if there exists some $f \colon A^n \to A$ which is not term-definable in $\alg{A}$. Indeed, in this case the $\NA$-operation induced by $\chi^f \colon \NA^n \Rightarrow \NA$ defined via $(\nu_1, \dots, \nu_n) \mapsto f^{\pw}(\nu_1, \dots, \nu_n)$ is not reducible with respect to $\lambda^\ev$.  
\hfill $\blacksquare$
\end{rem}

In the next subsection, we turn our attention towards the topic of reducibility of tests.

\subsection{Reducible tests}\label{subsection:ReducingTests}
Compared to coalgebra operations, reducibility for tests (recall Definition~\ref{definition:TestOperation}) is simpler to describe as it is captured by a \emph{term-definability} condition on $\alg{A}$.

\begin{defi}[\textsf{Reducible tests}]\label{definition:Reducibletest}
Let $\test \in \TeOp$ be a test and $\lambda \in \Pred$ be a $k$-ary predicate lifting. We call $\test$ \emph{reducible with respect to $\lambda$} if there is a $(k+1)$-ary $\alg{A}$-term function $t(x, x_1, \dots, x_k) \colon A^{k+1} \to A$ such that 
$$
\lambda_X(\sigma_1, \dots, \sigma_k) \circ \test(\sigma) = t^\mathrm{pw}(\sigma, \sigma_1, \dots, \sigma_k)$$
for all $\sigma, \sigma_i \in A^X$. We call $\test$ \emph{reducible} if it is reducible with respect to \emph{every} $\lambda \in \Pred$.  
\end{defi}

Similarly to Proposition~\ref{proposition:soundness-ReductionAxioms}, we get the following soundness result for tests.

\begin{prop}[\textsf{Soundness}]\label{proposition:SoundnessTests}
Let $t\colon A^{k+1} \to A$ witness that $\test$ is reducible with respect to $\lambda$. Then the formula 
$$
\diam{\test(\psi)}^\lambda (\varphi_1, \dots, \varphi_k) \leftrightarrow t(\psi, \varphi_1, \dots, \varphi_k)$$
is true at every state of any standard model.  
\end{prop}
\begin{proof}
In any standard model $\gamma\colon \Act \to (\func{F}X)^X$, we compute 
\begin{align*}
\sem{\diam{\test(\psi)}^\lambda(\varphi_1,\dots,\varphi_k)} & = \lambda_X(\sem{\varphi_1}, \dots, \sem{\varphi_k}) \circ \gamma_{\test(\psi)} \\
& = \lambda_X(\sem{\varphi_1}, \dots, \sem{\varphi_k}) \circ \test(\sem{\psi}) \\ 
& = t^\pw(\sem{\psi}, \sem{\varphi_1}, \dots, \sem{\varphi_k}) = \sem{t(\psi, \varphi_1, \dots, \varphi_k)} 
\end{align*}
as desired (here, we used that $\gamma_{\test(\psi)} = \test(\sem{\psi})$ holds in every standard model and reducibility of $\test$ with respect to $\lambda$). 
\end{proof}

We also get the following analogue of Proposition~\ref{proposition:ReducibleImpliesSafe} for tests, regarding the relationship between reducibility and \emph{safety} of tests (recall Definition~\ref{definition:SafeTestOperation}).

\begin{prop}\label{proposition:ReducibleTestsAreSafe}
Let the collection of predicate liftings $\Pred$ be separating and let $\test \in \TeOp$ be a test. If $\test$ is reducible, then it is safe.  
\end{prop}

\begin{proof}
To show this, we use the characterisation of safe tests of Lemma~\ref{lemma:SafetyTestOperations}(iii). Given $f \colon X \to Y$ and $\sigma \in A^Y$, by separation it suffices to show that for every $\lambda \in \Pred$ it holds that 
$$\widehat{\lambda}_Y \circ \func{F}f \circ \test(\sigma \circ f) = \widehat{\lambda}_Y \circ \test(\sigma) \circ f.$$
Say $\arity(\lambda) =: k$ and let $t$ be the $(k+1)$-ary term-function witnessing reducibility of $\test$ with respect to $\lambda$. Starting on the left-hand side of the above equation, note that $\widehat{\lambda}_Y \circ \func{F}f \circ \test(\sigma \circ f) = \kNA f \circ \widehat{\lambda}_X \circ \test(\sigma \circ f)$ by naturality. From here, for any $x\in X$ and $\sigma_1, \dots, \sigma_k \in A^Y$ we use reducibility to compute on the one hand
\begin{align*}
\big(\kNA f \circ \widehat{\lambda}_X \circ \test(\sigma \circ f)\big) (x)(\sigma_1, \dots, \sigma_k) 
& = \big(\lambda_X(\sigma_1 \circ f, \dots, \sigma_k \circ f) \circ \test(\sigma \circ f)\big)(x) \\ 
& = t^\mathrm{pw}(\sigma \circ f, \sigma_1 \circ f, \dots, \sigma_k \circ f)(x).
\end{align*}     
On the other hand we use reducibility as well and get
\begin{align*}
(\widehat{\lambda}_Y \circ \test(\sigma) \circ f)(x)(\sigma_1, \dots, \sigma_k) 
& = \big(\lambda_Y(\sigma_1, \dots, \sigma_k) \circ \test(\sigma)\big) \big(f(x)\big) \\ 
& = t^\mathrm{pw}(\sigma, \sigma_1, \dots, \sigma_k)\big(f(x)\big).
\end{align*}
Since both of these equal $t\big(\sigma(f(x)), \sigma_1(f(x)), \dots, \sigma_k(f(x))\big)$, this finishes the proof.
\end{proof}

In the following, we show that the tests we included in the coalgebraic logics of Subsection~\ref{subsection:SyntaxSemantics} are reducible.     

\begin{exa}\label{example:ReducibleTests}
The following are examples of reducible tests. The predicate liftings are as in Example~\ref{example:PredicateLiftings} and the tests as in Example~\ref{example:TestOperations}.  
\begin{enumerate}
\item Let $\func{F}$ be a pointed monad and $\test_P$ as in Example~\ref{example:TestOperations}(1), checking the `property' $P \subseteq A$. 
Furthermore, assume the characteristic function $\chi_P$ is $\alg{A}$-term definable (\emph{e.g.}, this holds for $\lucas_n$ and any other \emph{semi-primal} $\alg{A}$, see Remark~\ref{remark:SemiPrimal}), $\lambda$ is unary and $\widehat{\lambda}$ is a monad morphism. 
If $\widehat{\lambda}_X(\bot_{\func{F}X}) = 1$ for all $X$ (\emph{i.e.}, $\lambda$ is `$\Box$-like'), we can take 
$$
\diam{\test_P(\psi)}\varphi \leftrightarrow (\chi_P(\psi) \rightarrow \varphi)
$$ 
as reduction axiom. To see this, note that 
$$
\lambda_X(\sem{\varphi})\big(\test_P(\sem{\psi})(x)\big) = \begin{cases}
\lambda_X(\sem{\varphi})\big(\eta_X(x)\big) = \sem{\varphi}(x) & \text{if $\sem{\psi}(x) \in P$} \\
\lambda_X(\sem{\varphi})(\bot_{\func{F}X}) = 1 & \text{if $\sem{\psi}(x) \notin P$}
\end{cases}
$$
where in the first case we use that $\widehat{\lambda}$ is a monad morphism $\func{F} \Rightarrow \NA$. On the other hand, we have 
$$
\chi_P\big(\sem{\psi}(x)\big) \rightarrow \sem{\varphi}(x) = \begin{cases}
1 \rightarrow \sem{\varphi}(x) = \sem{\varphi}(x) & \text{if $\sem{\psi}(x) \in P$} \\
0 \rightarrow \sem{\varphi}(x) = 1 & \text{if $\sem{\psi}(x) \notin P$}
\end{cases}
$$ 
as well, finishing the proof. 
For example, with $P = \{ 1 \}$ and $\alg{A} = \lucas_n$ we recover the axiom $[\psi?]\varphi \leftrightarrow (\psi^n \rightarrow \varphi)$ used in \cite[Proposition 2.3]{Teheux2014}. 

If $\widehat{\lambda}_X(\bot_{\func{F}X}) = 0$ for all $X$ (\emph{i.e.}, $\lambda$ is `$\Diamond$-like'), we may take 
$$
\diam{\test_P(\psi)}\varphi \leftrightarrow (\chi_P(\psi) \wedge \varphi)
$$ as reduction axiom, which is shown similarly.    
\item Consider $\PA$ with the $\alg{A}$-test $\test$ as in Example~\ref{example:TestOperations}(2). For $\lambda^\Box$ corresponding to the $\alg{A}$-valued $\Box$-modality on $\alg{A}$-weighted frames we have the reduction axiom 
$$
\diam{\test(\psi)}^{\lambda^\Box}\varphi \leftrightarrow (\psi \rightarrow \varphi)
$$
and for $\lambda^\Diamond$ corresponding to the $\alg{A}$-valued $\Diamond$-modality on $\alg{A}$-weighted frames we have the reduction axiom 
$$
\diam{\test(\psi)}^{\lambda^\Diamond}\varphi \leftrightarrow (\psi \odot \varphi).
$$ 
To see this, one simply computes 
\begin{align*}
\lambda_X^\Box(\sem{\varphi})\big(\test(\sem{\psi})(x)\big) & = \lambda_X^\Box(\sem{\varphi})(e_x \odot^\pw \sem{\psi}) \\
& = \bigwedge_{x'\in X} \big(e_x(x') \odot \sem{\psi}(x')\big) \rightarrow \sem{\varphi}(x')\\
& = \sem{\psi}(x) \rightarrow \sem{\varphi}(x) 
\end{align*} 
and
\begin{align*}
\lambda_X^\Diamond(\sem{\varphi})\big(\test(\sem{\psi})(x)\big) & = \lambda_X^\Diamond(\sem{\varphi})(e_x \odot^\pw \sem{\psi}) \\
& = \bigvee_{x'\in X} e_x(x') \odot \sem{\psi}(x') \odot \sem{\varphi}(x') \\
& = \sem{\psi}(x) \odot \sem{\varphi}(x) 
\end{align*} 
as desired.
\item Again, consider the functor $\PA$ with $\test$ as in Example~\ref{example:TestOperations}(2), noting that this also defines a $2$-test for $\PA$. Let $r\in A{\setminus}\{ 0 \}$ and $\lambda^r$ be the corresponding $2$-predicate lifting as in Example~\ref{example:PredicateLiftings}(3). Then we find the reduction axiom 
$$
\diam{\test(\psi)}^r \varphi \leftrightarrow (\psi \wedge \varphi).
$$  
To see this, we observe that
$$
\lambda^r_X(\sem{\varphi})\big(\test(\sem{\psi})(x)\big) = 1 \Leftrightarrow \bigvee_{x' \in \sem{\varphi}} e_x(x') \odot \sem{\psi}(x') \geq r \Leftrightarrow \sem{\varphi}(x) \wedge \sem{\psi}(x) = 1  
$$
as desired.
\item For $\NA$ or $\MA$, predicate lifting $\lambda^\ev$ and $\test$ as in Example~\ref{example:TestOperations}(3), we again find the reduction axiom 
$$
\diam{\test(\psi)}\varphi \leftrightarrow (\psi \odot \varphi),
$$ 
which is shown by
\begin{align*}
\lambda^\ev_X(\sem{\varphi})\big(\test(\sem{\psi})(x)\big) & = \lambda^\ev(\sem{\varphi})\big(\ev_x((\text{-}) \odot^\pw \sem{\psi})\big) \\
& = \ev_x(\sem{\varphi} \odot^\pw \sem{\psi}) \\
&= \sem{\varphi}(x) \odot \sem{\psi}(x)
\end{align*}
as desired.
\item For $\mathcal{P} \circ \mathcal{P}$, $\lambda^{k+1}$ and $\test$ from Example~\ref{example:TestOperations}(4), we get $\diam{\test(\psi)}^{k+1} (\varphi_1, \dots, \varphi_k, \varphi) \leftrightarrow (\psi \wedge \varphi \wedge \varphi_1 \wedge \dots \wedge \varphi_k)$ analogously to \cite[Proposition 3]{vanBenthemBezhanishviliEnqvist2019}. \hfill $\blacksquare$  
\end{enumerate}     
\end{exa} 

With \emph{soundness} for reducible operations taken care of, next we deal with \emph{completeness}.

\section{Strong completeness via quasi-canonical models}\label{section:Completeness}
From now on, we assume that \emph{every} coalgebra operation $O \in \CoOp$ and test $\test \in \TeOp$ is reducible (Definitions~\ref{definition:ReducibleCoalgOperation} \& \ref{definition:Reducibletest}) and that we already know a sound and strongly complete axiomatization of the \emph{underlying $\alg{A}$-valued coalgebraic logic} of $\func{F}$ and $\Pred$, denoted by $\LT$ (some examples are found in Subsection~\ref{subsection:example-applications}). 

\begin{defi}[\textsf{Coalgebraic dynamic logic $\LTDyn$}]\label{definition: LTDyn}
The \emph{dynamic coalgebraic logic} $\LTDyn \subseteq \Formu$ is the smallest set of dynamic formulas satisfying the following conditions. 
\begin{itemize}
    \item For every action $a \in \Act$, the set $\LTDyn$ contains the axioms of $\LT$ with respect to the modalities $\{\diam{a}^\lambda\}_{\lambda\in \Lambda}$ and is closed under the corresponding rules of $\LT$.
    \item For every pair $(O, \lambda) \in \CoOp\times\Pred$, the set $\LTDyn$ contains all instances of the reduction axiom 
 $$\diam{O(\vec{a})}^\lambda(\vec{\varphi}) \leftrightarrow \tau[\vec{a}, \vec{\varphi}],$$
where the template $\tau$ witnesses reducibility of $O$ with respect to $\lambda$ (Definition~\ref{definition:ReducibleCoalgOperation}).
    \item For every pair $(\test, \lambda) \in \TeOp\times\Pred$, the set $\LTDyn$ contains all instances of the reduction axiom 
    $$\diam{\test(\psi)}^\lambda (\vec{\varphi}) \leftrightarrow t(\psi, \vec{\varphi}),$$
where the $\alg{A}$-term-function $t$ witnesses reducibility of $\test$ with respect to $\lambda$ (Definition~\ref{definition:Reducibletest}).
\end{itemize}   
\end{defi}

In this section, assuming finite one-step completeness of $\LT$ (Definition~\ref{definition:One-step-completeness}) we prove \emph{strong completeness} 
$$\Gamma \models \varphi \ \Longleftrightarrow \ \Gamma \proves_{\LTDyn} \varphi,$$ 
where the relation on the left-hand side is local semantic entailment in standard models (recall Subsection~\ref{subsection:SyntaxSemantics}) and the relation on the right-hand side intuitively holds if $\varphi$ can be derived from $\Gamma$ and $\LTDyn$ together with $\alg{A}$-reasoning (recall Subsection~\ref{subsection:Many-valuedModalLogic}), which is formally defined in terms of \emph{non-modal homomorphisms} in Subsection~\ref{subsection:CanonicalModels}.

\begin{rem}\label{remark:Section5Intro}
As mentioned in Remark~\ref{remark:Section4Intro}, if all operations and tests are reducible then dynamic formulas are provably equivalent to non-dynamic formulas. Thus, for the purpose of completeness the crucial result of this section is Theorem~\ref{theorem:FinitelyOneStepComplete-CanModel} establishing strong completeness of the underlying non-dynamic multi-modal logic. However, let us note that altogether we prove a somewhat stronger result in this section, namely the existence of a \emph{standard} canonical model which is \emph{coherent on all actions} (Subsection~\ref{subsection:CanonicalModels}), which provides a basis for future work aiming to establish (weak) completeness of coalgebraic dynamic logics including non-reducible operations similarly to \cite{HansenKupke2015}. \hfill $\blacksquare$ 
\end{rem}

\subsection{Non-modal homomorphisms and quasi-canonical models}\label{subsection:CanonicalModels}
In the two-valued setting, canonical models are usually based on the collection of \emph{maximally consistent theories}. The corresponding generalisation we use for our many-valued logics is given by \emph{non-modal homomorphisms} \cite{BouEstevaGodo2011} (let us mention here that non-modal homomorphisms have previously been used for many-valued PDL in \cite{Teheux2014,Sedlar2020} and many-valued coalgebraic logic in \cite{LinLiau2023}).

\begin{defi}[\textsf{Non-modal homomorphisms}]\label{definition:Non-modal-homomorphism}
A map $h \colon \Formu \to \alg{A}$ is a \emph{non-modal homomorphism} for $\LTDyn$ if it respects the $\alg{A}$-operations\footnote{\emph{I.e.}, $h$ is a homomorphism in the sense of $h(o(\varphi_1, \dots, \varphi_{\arity(o)})) = o^\alg{A}(h(\varphi_1), \dots, h(\varphi_{\arity(o)})$ for all $o \in \sign(\alg{A})$.} and satisfies $h(\LTDyn) = \{1 \}$. The set of all non-modal homomorphisms is denoted by $\NMH$.  
\end{defi} 

The local syntactic entailment is now defined via
$$
\Gamma \proves_{\LTDyn} \varphi \quad \text{if and only if} \quad h(\Gamma) = \{ 1 \} \Rightarrow h(\varphi) = 1 \text{ for all } h \in \NMH.
$$ 
The collection of non-modal homomorphisms $\NMH$ is the carrier set of the \emph{quasi-canonical models} we construct later on (generalising \cite{SchroederPattinson2009,HansenKupkeLeal2014}, the `\emph{quasi}' refers to possible non-uniqueness). Such models are defined as follows.

\begin{defi}[\textsf{Quasi-canonical models}]\label{definition:quasi-can-model}
A map $\zeta \colon \Act \to (\func{F} \NMH)^\NMH$ together with the \emph{canonical propositional valuation} $\sem{p}^\mathrm{c} \coloneq \ev_p$ (\emph{i.e.}, the evaluation map $h \mapsto h(p)$) is a \emph{quasi-canonical model for $\LTDyn$} if the \emph{coherence property} 
$$\sem{\varphi}^\mathrm{c} = \ev_\varphi$$
holds for \emph{all} dynamic formulas $\varphi \in \Formu$. 
\end{defi} 

The existence of a quasi-canonical standard model is sufficient for strong completeness, which is shown by the following standard argument. 

\begin{prop}[\textsf{Strong completeness}]
If there exists a quasi-canonical standard model $\zeta$ for $\LTDyn$, then for all $(\Gamma, \varphi) \in \mathcal{P}(\Formu) \times \Formu$:
$$\Gamma \proves_{\LTDyn} \varphi \ \text{ if and only if } \ \Gamma \models \varphi.$$  
\end{prop} 
\begin{proof}
Soundness `$\Rightarrow$' was established in Propositions~\ref{proposition:soundness-ReductionAxioms} \& ~\ref{proposition:SoundnessTests}. For completeness `$\Leftarrow$', proceed by contrapositive assuming that $\Gamma \not\proves \varphi$. Then there exists a non-modal homomorphism $h \in \NMH$ with $h(\Gamma) = \{ 1 \}$ but $h(\varphi) \neq 1$. Thus, in the quasi-canonical model $\zeta$ we have $\sem{\psi}^\mathrm{c}(h) = \ev_\psi(h) = h(\psi) = 1$ for all $\psi \in \Gamma$, while $\sem{\psi}^\mathrm{c}(h) = \ev_\psi(h) = h(\psi) \neq 1$. In other words, the standard model $\zeta$ witnesses $\Gamma \not\models \varphi$ as desired.       
\end{proof}  

The main result of this subsection shows that it suffices to establish the coherence property for \emph{atomic} actions in order to obtain a standard quasi-canonical model (we construct such a model which satisfies the coherence property on atomic actions using one-step completeness in Subsection~\ref{subsection:OneStepCompleteness}).

\begin{thm}\label{theorem:StandardQuasiCanModel}
Let $\zeta^0 \colon \AtAct \to (\func{F}\NMH)^{\NMH}$ be coherent for \emph{atomic} actions, that is, 
$$
\lambda_\NMH \big(\ev_{\varphi_1}, \dots, \ev_{\varphi_{\arity(\lambda)}}\big)\big(\zeta_{a_0}^0 (h)\big) = h\big(\diam{a_0}^\lambda (\varphi_1, \dots, \varphi_{\arity(\lambda)})\big) 
$$
holds for all $a_0 \in \AtAct$, $\lambda \in \Pred$, $\varphi_i \in \Formu$ and $h \in \NMH$. Let $\zeta \colon \Act \to (\func{F}\NMH)^\NMH$ be the standard extension of $\zeta^0$ as in Definition~\ref{definition:StandardModel}. Then $\zeta$ is a quasi-canonical standard model for $\LTDyn$.     
\end{thm} 

\begin{proof}
By mutual induction we show that $\sem{\varphi}^c = \ev_{\varphi} \colon \NMH \to A$ for all formulas $\varphi \in \Formu$ and coherence (as in the statement) for all actions $a \in \Act$ hold in $\zeta$. Here, only the case $\varphi = \diam{a}^\lambda (\varphi_1, \dots, \varphi_k)$ for some non-atomic action $a \in \Act$ and $k$-ary $\lambda \in \Pred$ is not immediate.  

\textbf{Coalgebra operations.} Let $a = O(a_1, \dots ,a_n)$ for some $n$-ary coalgebra operation $O \in \CoOp$ and inductively assume that coherence already holds for all $\zeta_{a_j}$. Let $\tau$ be the $(n,k)$-template corresponding to the reduction axiom for $(O, \lambda)$ which is contained in $\LTDyn$ (recall Definition~\ref{definition: LTDyn}). Since this reduction axiom is contained in $\LTDyn$, for every $h \in \NMH$ we have  
$$
h\big(\diam{O(a_1, \dots, a_n)}^\lambda (\varphi_1, \dots, \varphi_k)\big) = h\big(\tau[a_1, \dots, a_n, \varphi_1, \dots, \varphi_k]\big),$$ 
which shows $\ev_{\diam{O(a_1, \dots, a_n)}^\lambda (\varphi_1, \dots, \varphi_k)} = \ev_{\tau[a_1, \dots, a_n, \varphi_1, \dots, \varphi_k]}$. On the other hand, since $\tau$ witnesses reducibility, together with the inductive hypothesis on $\varphi_1, \dots, \varphi_k$ we get 
\begin{align*}
\sem{\diam{O(a_1, \dots, a_n)}^\lambda (\varphi_1, \dots, \varphi_k)}^\mathrm{c} & = \lambda_\NMH (\ev_{\varphi_1}, \dots, \ev_{\varphi_k}) \circ O(\zeta_{a_1}, \dots, \zeta_{a_n}) \\ 
& = \tau[\zeta_{a_1}, \dots, \zeta_{a_n}, \ev_{\varphi_1}, \dots, \ev_{\varphi_k}].
\end{align*}
Thus, it suffices to show that for every $\tau \in \nkTempl$ it holds that 
$$\ev_{\tau[\vec{a}, \vec{\varphi}]} = \tau[\zeta_{\vec{a}}, \ev_{\vec{\varphi}}],$$  
where we abbreviate $\zeta_{\vec{a}} = (\zeta_{a_1}, \dots, \zeta_{a_n})$ and $\ev_{\vec{\varphi}} = (\ev_{\varphi_1}, \dots, \ev_{\varphi_k})$. To show this, we proceed by induction on the structure of the reduction template $\tau$. 

For $\tau = \omega_i$, directly by definition we get 
$$
\ev_{\omega_i[\vec{a}, \vec{\varphi}]} = \ev_{\varphi_i} = \omega_i[\zeta_{\vec{a}}, \ev_{\vec{\varphi}}]
$$
as desired.

For $\tau = o(\tau_1 , \dots, \tau_{\arity(o)})$ we use the inductive hypothesis on $\tau_i$ for all $i = 1, \dots, \arity(o)$ to compute
$$
\tau[\zeta_{\vec{a}}, \ev_{\vec{\varphi}}] = o^\mathrm{pw} \big(\tau_1[\zeta_{\vec{a}}, \ev_{\vec{\varphi}}], \dots, \tau_{\arity(o)}[\zeta_{\vec{a}}, \ev_{\vec{\varphi}}]\big) = o^\mathrm{pw}(\ev_{\tau_1[\vec{a}, \vec{\varphi}]}, \dots, \ev_{\tau_{\arity(o)}[\vec{a}, \vec{\varphi}]}).
$$
Using the fact that $h$ (being a non-modal homomorphism) commutes with the $\alg{A}$-operation $o$ we observe that the above sends every $h \in \NMH$ to  
$$
o^\mathrm{pw}(\ev_{\tau_1[\vec{a}, \vec{\varphi}]}, \dots, \ev_{\tau_{\arity(o)}[\vec{a}, \vec{\varphi}]})(h) = h\big (o(\tau_1[\vec{a}, \vec{\varphi}], \dots, \tau_{\arity(o)}[\vec{a}, \vec{\varphi}]) \big) = h\big(\tau[\vec{a}, \vec{\varphi}]\big)
$$
as desired. 

Lastly, for the case $\tau = \diam{j}^{\lambda'}(\tau_1, \dots, \tau_{\arity(\lambda')})$, first we use the inductive hypothesis on the $\tau_i$ for all $i = 1,\dots, \arity(\lambda)$ and get 
\begin{align*}
\tau[\zeta_{\vec{a}}, \ev_{\vec{\varphi}}] & = \lambda_\NMH'\big(\tau_1[\zeta_{\vec{a}}, \ev_{\vec{\varphi}}], \dots, \tau_{\arity(\lambda')}[\zeta_{\vec{a}}, \ev_{\vec{\varphi}}]\big) \circ \zeta_{a_j} \\
& = \lambda_\NMH'(\ev_{\tau_1[\vec{a}, \vec{\varphi}]}, \dots, \ev_{\tau_{\arity(\lambda')}[\vec{a}, \vec{\varphi}]} ) \circ \zeta_{a_j}.
\end{align*}
Now we make use of the inductive hypothesis on coherence with respect to the action $a_j$, which yields that for every $h \in \NMH$ we have 
$$
\lambda_\NMH'(\ev_{\tau_1[\vec{a}, \vec{\varphi}]}, \dots, \ev_{\tau_{\arity(\lambda')}[\vec{a}, \vec{\varphi}]} ) \big( \zeta_{a_j}(h)\big) = h\big(\diam{a_j}^{\lambda'} (\tau_1[\vec{a}, \vec{\varphi}], \dots, \tau_{\arity(\lambda')}[\vec{a}, \vec{\varphi}])\big)
$$
as desired. This finishes the case of composite actions obtained via coalgebra operations.

\textbf{Tests.}
The last case we need to discuss is $a = \test(\psi)$, inductively assuming that $\sem{\psi} = \ev_\psi$ already holds. Let $t(x, x_1, \dots, x_k) \colon A^{k+1} \to A$ be the $\alg{A}$-term-function corresponding to the reduction axiom for $(\test, \lambda)$ which is included in $\LTDyn$. Then for any $h \in \NMH$ we get  
\begin{align*}
\sem{\diam{\test(\psi)}^\lambda(\vec{\varphi})}^\mathrm{c} & = \lambda_{\NMH}(\ev_{\varphi_1}, \dots, \ev_{\varphi_k})\big(\test(\ev_{\psi})(h)\big) && \text{(Def. \& Ind. Hyp.)} \\ 
& = t\big(\ev_{\psi}(h), \ev_{\varphi_1}(h), \dots, \ev_{\varphi_k}(h)\big) && \text{(Reducibility $\test$)} \\ 
& = t\big(h(\psi), h(\varphi_1), \dots,h(\varphi_k)\big)  = h\big(t(\psi, \vec{\varphi})\big). && \text{(Definitions $\ev$ \& $h \in \NMH$)}  
\end{align*} 
Finally, we have $h(t(\psi, \vec{\varphi})) = h(\diam{\test(\psi)}\vec{\varphi})$ since this reduction axiom is included in $\LTDyn$.   
\end{proof}

Note that this theorem does not require any extra assumptions on the algebra of truth-degrees $\alg{A}$. In the next subsection, we describe one way to prove \emph{existence} of quasi-canonical models via \emph{one-step completeness}, where we restrict our attention to \emph{finite} $\alg{A}$.

\subsection{Finite one-step completeness}\label{subsection:OneStepCompleteness}
Our strategy to obtain quasi-canonical models is inspired by \cite{SchroederPattinson2009,HansenKupkeLeal2014} and relies on the assumption that the underlying coalgebraic logic $\LT$ is \emph{strongly one-step complete over finite sets}. While \cite[Section 5]{HansenKupkeLeal2014} follows a similar strategy, here we adapt the proof of \cite[Theorem 2.5]{SchroederPattinson2009} more directly.

First, recall that a formula is \emph{rank-1} if every propositional variable is in the scope of precisely one modality. 
For any set $X$, we define the set of formal expressions 
$$\Pred(A^X) \coloneq \{ \heartsuit^\lambda (\sigma_1, \dots, \sigma_{\arity(\lambda)}) \mid \lambda \in \Pred, \sigma_i \in A^X  \}.$$
We call $H \colon \Pred(A^X) \to A$ a  \emph{rank-1 non-modal homomorphism for $\LT$} if $H$ (extended to $\alg{A}$-term combinations of $\Pred(A^X)$ in the obvious way) sends all instances of \emph{rank-1 formulas} of $\LT$ to $1$.  
We say that $H$ is \emph{(one-step) satisfiable} if there exists some $\alpha \in \func{F}X$ such that $\lambda_X(\sigma_1, \dots, \sigma_{\arity(\lambda)})(\alpha) = H\big(\heartsuit^\lambda(\sigma_1, \dots, \sigma_{\arity(\lambda)})\big)$ holds for all $\lambda \in \Pred$ and $\sigma_i \in A^X$. The following is a generalisation of \cite[Definition 2.4]{SchroederPattinson2009}.

\begin{defi}[\textsf{Finite one-step completeness}]\label{definition:One-step-completeness}
$\LT$ is \emph{strongly one-step complete over finite sets} if every rank-1 non-modal homomorphism $H \colon \Pred(A^X) \to A$ with finite $X$ is satisfiable.       
\end{defi}

Recalling more terminology from \cite{SchroederPattinson2009}, a \emph{surjective $\omega$-cochain of finite sets} is a sequence of finite sets $(X_m)_{m\in\mathbb{N}}$ together with surjective projection maps $p_m \colon X_{m+1} \to X_m$. Its \emph{inverse limit} $\varprojlim X_m$ consists of all sequences $(x_m) \in \prod X_m$ that are \emph{coherent}, that is $p_m(x_{m+1}) = x_m$. The functor $\func{F}$ \emph{weakly preserves inverse limits of surjective $\omega$-cochains of finite sets} if $\func{F}(\varprojlim X_i) \to \varprojlim \func{F} X_i$ is surjective. For brevity, we simply call functors $\func{F}$ \emph{weakly preserving} in this case.

The main theorem of this subsection is the following adaptation of \cite[Theorem 2.5]{SchroederPattinson2009} to the dynamic setting.
In order to effectively use Definition~\ref{definition:One-step-completeness} in the proof, we need $\alg{A}$ itself to be finite. Also note that in the following we only require the set $\Lambda$ of predicate liftings to be separating \emph{on finite sets} (see Definition~\ref{definition:SeparatingPredLifts}). 

\begin{thm}\label{theorem:FinitelyOneStepComplete-CanModel}
Let $\alg{A}$ be a finite $\FLew$-algebra, the collection of predicate liftings $\Pred$ be separating on finite sets, the functor $\func{F}$ be weakly preserving and the (non-dynamic) coalgebraic logic $\LT$ be strongly one-step complete over finite sets. Then there exists a quasi-canonical model for the coalgebraic dynamic logic $\LTDyn$.   
\end{thm}

\begin{proof}
Fix an arbitrary atomic action $a \in \AtAct$ and an arbitrary non-modal homomorphism $h \in \NMH$. It suffices to show that there exists some $\alpha \in \func{F}\NMH$ which satisfies  
$$
\lambda_\NMH (\ev_{\varphi_1}, \dots, \ev_{\varphi_{\arity(\lambda)}})(\alpha) = h\big(\diam{a}^\lambda (\varphi_1, \dots, \varphi_{\arity(\lambda)})\big) 
$$ 
for all $\lambda \in \Pred$ and all $\varphi_1, \dots, \varphi_{\arity(\lambda)} \in \Formu$, since then we may define
$\zeta_{a}^0(h) \coloneq \alpha$, which ensures that $\zeta_a^0$ is coherent on atomic actions
as required in Theorem~\ref{theorem:StandardQuasiCanModel}. Equivalently, setting 
$\mathcal{F} = \{ \ev_\varphi \mid \varphi \in \Formu \} \subseteq A^{\NMH}$
and 
\begin{align*}
H_a \colon \Pred(\mathcal{F}) &\to A \\ 
\heartsuit^\lambda(\ev_{\varphi_1}, \dots, \ev_{\varphi_{\arity(\lambda)}}) & \mapsto h(\diam{a}(\varphi_1, \dots, \varphi_{\arity(\lambda)}))
\end{align*}
we need to show that $H_a$ is one-step satisfiable in $\func{F}\NMH$. 

Enumerate the sets $\Prop$, $\Act$ and $\Pred$ (recall from Subsection~\ref{subsection:SyntaxSemantics} that these are assumed countable) and, 
for all $m \in \mathbb{N}$, let $\Formu^{(m)} \subseteq \Formu$ be the set of formulas $\phi$ such that $\phi$ only contains the first $m$-many propositional variables, actions and predicate liftings according to these enumerations, and $\phi$ has modal depth at most $m$ and at most $m$ occurrences of $\alg{A}$-connectives. Note that all $\Formu^{(m)}$ thus defined are finite.
Furthermore, for all $m \in \mathbb{N}$, we define the collection of restrictions of non-modal homomorphisms 
$$\NMH^{(m)} = \{ h{\restriction}_{\Formu^{(m)}} \mid h \in \NMH \}.$$ 
Since $\alg{A}$ and $\Formu^{(m)}$ are finite, the sets $\NMH^{(m)}$ are also finite. It is easily seen that we obtain the inverse limit $\NMH \cong \varprojlim \NMH^{(m)}$ where all projection maps are given by the obvious restrictions. Now let 
$$\mathcal{F}^{(m)} = \{ \ev_{\varphi} \mid \varphi \in \Formu^{(m)} \} \subseteq A^{\NMH^{(m)}}$$
and let $H^{(m)}_a \colon \Pred(\mathcal{F}^{(m)}) \to A$ be the corresponding restrictions of $H_a$. Since we included the axioms for $\LT$ (`instantiated' for the action $a$) in $\LTDyn$ (Definition~\ref{definition: LTDyn}) and $h$ is a non-modal homomorphism for $\LTDyn$, it follows that $H^{(m)}_a$ is a rank-1 non-modal homomorphism for $\LT$. Thus, by strong one-step completeness over finite sets of $\LT$, for every $m \in \mathbb{N}$ there exists some $\alpha^{(m)} \in \func{F}(\NMH^{(m)})$ which one-step satisfies $H^{(m)}_a$. 

Since the collection of predicate liftings $\Pred$ is separating on finite sets and all $\NMH^{(m)}$ are finite, the sequence $(\alpha^{(m)})$ is coherent, \emph{i.e.}, an element of $\varprojlim \NMH^{(m)}$. Lastly, since $\func{F}$ is weakly preserving, there exists some $\alpha \in \func{F}\NMH$ which induces this sequence. To check that this $\alpha$ one-step satisfies $H_a$, for $\lambda \in \Pred$ and $\varphi_1, \dots, \varphi_{\arity(\lambda)}$ choose $m$ sufficiently large to ensure $\diam{a}^\lambda (\varphi_1, \dots, \varphi_{\arity(\lambda)}) \in \Formu^{(m)}$. Then naturality of $\lambda$ applied to the limit projection $\NMH \to \NMH^{(m)}$ finally yields
\begin{align*}
\lambda_\NMH(\ev_{\varphi_1}, \dots, \ev_{\varphi_{\arity(\lambda)}})(\alpha) & = \lambda_{\NMH^{(m)}}(\ev_{\varphi_1}, \dots, \ev_{\varphi_{\arity(\lambda)}})(\alpha^{(m)}) \\ 
& = H_a^{(m)}\big(\heartsuit^\lambda(\ev_{\varphi_1}, \dots, \ev_{\varphi_{\arity(\lambda)}} )\big) \\ 
& = H_a \big(\heartsuit^\lambda(\ev_{\varphi_1}, \dots, \ev_{\varphi_{\arity(\lambda)}} )\big) = h\big(\diam{a}^\lambda (\varphi_1, \dots, \varphi_{\arity(\lambda)})\big)
\end{align*} 
as desired, due to our construction of the $\alpha^{(m)}$. This finishes the proof.                
\end{proof}

In the following subsection, we give some sample applications of this now established `recipe' to obtain strong completeness.

\subsection{Applications}\label{subsection:example-applications}
In order to illustrate the main results of this section we show how to apply them to obtain strong completeness for the iteration-free fragments of the coalgebraic dynamic logics of Subsection~\ref{subsection:SyntaxSemantics}. 

\begin{exa}[\textsf{Many-valued PDL with crisp accessibility relations}]\label{example: MVPDL-crisp} 
Consider the iteration-free fragment of the coalgebraic dynamic logic of Example~\ref{example:CoalgebraicDynamicLogics-PDL-Crisp}. 
Since $\mathcal{P}$ is weakly preserving \cite[Example 3.1]{SchroederPattinson2009} and $\Lambda \subseteq \{ \lambda^\Box, \lambda^\Diamond \}$ is separating (Example~\ref{example:PredicateLiftings}(1)), for finite $\alg{A}$ we can apply Theorem~\ref{theorem:FinitelyOneStepComplete-CanModel} with the reduction axioms 
\begin{align*}
[a \cup b]\varphi &\leftrightarrow [a]\varphi \wedge [b]\varphi  &  \diam{a \cup b}\varphi &\leftrightarrow \diam{a}\varphi \vee \diam{b}\varphi\\
[a{;}b]\varphi &\leftrightarrow [a][b]\varphi & \diam{a;b}\varphi &\leftrightarrow \diam{a}\diam{b}\varphi  
\end{align*}
of Examples~\ref{example:ReductionKleisliMonadMorphism} \& \ref{example:IndependentlyReducibleOperations}(1). For the tests $\TeOp \subseteq \{ \test_P \mid P \subseteq A \}$ the reduction axioms are given by  
\begin{align*}
[\test_P (\psi)]\varphi &\leftrightarrow \chi_P(\psi) \rightarrow \varphi & \diam{\test_P (\psi)}\varphi &\leftrightarrow \chi_P(\psi) \wedge \varphi 
\end{align*}
if the characteristic function $\chi_P$ is term-definable in $\alg{A}$ as discussed in Example~\ref{example:ReducibleTests}(1). Furthermore, we may also expand $\CoOp$ by the counter-domain operation $\sim$, where the corresponding reduction axioms are 
\begin{align*}
[{\sim}a]\varphi &\leftrightarrow [a]\bot \rightarrow \varphi & \diam{{\sim}a}\varphi &\leftrightarrow \neg\diam{a}\top \wedge \varphi 
\end{align*}
as discussed in Example~\ref{example:IndependentlyReducibleOperations}(6). 

For the underlying logic $\LT$, for $\Lambda = \{ \lambda^\Box \}$ one may consider the axiomatizations of \cite[Subsection 4.4]{BouEstevaGodo2011} 
if $\alg{A}$ is endowed with truth-constants and has a unique coatom. For $\alg{A} = \alg{G}_n$ a finite Gödel chain \emph{without} constants see \cite[Section 4]{CaicedoRodriguez2010} and for the bi-modal case \cite{RodriguezVidal2021}). If $\alg{A}$ is \emph{semi-primal}, axiomatizations are found in \cite[Subsection 4.1]{KurzPoigerTeheux2024}. In particular, this is the case for finite \L ukasiewicz-chains $\lucas_n$ (also see \cite{HansoulTeheux2013, BouEstevaGodo2011}), hence this may be used to obtain strong completeness for the iteration-free part of $\lucas_n$-valued PDL with crisp accessibility relations considered in \cite{Teheux2014}.      
\hfill $\blacksquare$
\end{exa}

In order to prepare for Examples~\ref{example:CoalgebraicDynamicLogics-PDL-Labelled1} \&~\ref{example:CoalgebraicDynamicLogics-PDL-Labelled2} which both use weighted accessibility relations (\emph{i.e.}, the coalgebraic signature functor $\PA$), we first prove the following (for the definition of weakly preserving recall the paragraph after Definition~\ref{definition:One-step-completeness}). 

\begin{lem}\label{lemma:PA-is-weaklypreserving}
For every linear $\alg{A}$, the functor $\PA$ is weakly preserving. 
\end{lem}  
\begin{proof}
Let $(\sigma_m)_{m\in \mathbb{N}} \in \varprojlim \PA X_m$ be coherent, \emph{i.e.}, $\sigma_m(x^m) = \bigvee \{ \sigma_{m+1}(x^{m+1}) \mid p_m(x^{m+1}) = x^m \}$. We claim that this sequence is induced by $\sigma \in \PA(\varprojlim X_m)$ defined by 
$$
\sigma(\vec{x}) = \bigwedge_{n\in \mathbf{N}} \sigma_n(x_n).
$$
That is, for $m \in \mathbb{N}$ and $x^m \in X_m$ we need to verify 
$$
\sigma_m(x^m) = \bigvee_{\pi_m(\vec{x}) = x^m} \bigwedge_{n \in \mathbb{N}} \sigma_n(x_n).
$$
The inequality `$\geq$' is obvious by definition. For the other inequality `$\leq$' we construct a sequence $\vec{x} \in \varprojlim X_n$ such that $\pi(\vec{x}) = x^m$ and $\bigwedge x_n = \sigma_m(x^m)$. By coherence we have $\sigma_m(x^m) = \bigvee\{ \sigma_{m+1}(x^{m+1}) \mid p_m(x^{m+1} = x^m)\}$. Thus, by surjectivity of $p_m$ and linearity of $\alg{A}$ there exists some $x_{m+1} \in X_{m+1}$ with $p_m(x_{m+1})$ and $\sigma_{m+1}(x_{m+1}) = \sigma_{m}(x^m)$. Repeating this argument, we construct a coherent sequence $(x_n)_{n \in \mathbb{N}}$ with $\sigma_n(x_n) = \sigma_m(x^m)$ for all $n \geq m$. Since in every coherent sequence it holds that $n < n'$ implies $\sigma_{n'}(x_{n'}) \leq \sigma_n(x_n)$, for the sequence thus constructed we find 
$$
\bigwedge_{n\in\mathbb{N}} \sigma_n(x_n) = \sigma_m(x^m)
$$
as desired, finishing the proof.        
\end{proof}

\begin{exa}[\textsf{Many-valued PDL with weighted accessibility relations}]\label{example: MVPDL-labeled}
Consider the iteration-free part of the logic of Example~\ref{example:CoalgebraicDynamicLogics-PDL-Labelled1} with finite linear $\alg{A}$. Since $\PA$ is weakly preserving (Lemma~\ref{lemma:PA-is-weaklypreserving}), and $\Lambda \subseteq \{ \lambda^\Box, \lambda^\Diamond \}$ is separating (Example~\ref{example:PredicateLiftings}(2)), we can apply Theorem~\ref{theorem:FinitelyOneStepComplete-CanModel} with the reduction axioms 
\begin{align*}
[a \vee b]\varphi &\leftrightarrow [a]\varphi \wedge [b]\varphi  &  \diam{a \vee b}\varphi &\leftrightarrow \diam{a}\varphi \vee \diam{b}\varphi\\
[a{;}b]\varphi &\leftrightarrow [a][b]\varphi & \diam{a;b}\varphi &\leftrightarrow \diam{a}\diam{b}\varphi \\
[\test(\psi)]\varphi &\leftrightarrow (\psi \rightarrow \varphi) & \diam{\test(\psi)}\varphi &\leftrightarrow (\psi \odot \varphi)  
\end{align*}
of Example~\ref{example:ReductionKleisliMonadMorphism}, Example~\ref{example:IndependentlyReducibleOperations}(2) and Example~\ref{example:ReducibleTests}(2), respectively. Furthermore, we may also expand $\CoOp$ by the counter-support $\sim$, where the corresponding reduction axioms are 
\begin{align*}
[{\sim}a]\varphi &\leftrightarrow \chi_{ \{ 1 \} }([a]\bot) \rightarrow \varphi & \diam{{\sim}a}\varphi &\leftrightarrow \chi_{ \{ 0 \} }(\diam{a}\top) \wedge \varphi 
\end{align*}
given that $\chi_{ \{ 1\} }$ and $\chi_{ \{ 0\} }$ are term-definable in $\alg{A}$, as discussed in Example~\ref{example:IndependentlyReducibleOperations}(7).

If $\alg{A}$ is endowed with canonical truth-constants $\check{c}$ for all $c \in \alg{A}$, for an axiomatization of $\LT$ where $\Pred = \{\lambda^\Box\}$ consider the one from \cite[Subsection 4.2]{BouEstevaGodo2011} (for $\alg{A} = \lucas_n$ \emph{without} constants see \cite[Subsection 5.1]{BouEstevaGodo2011}). In particular, this may be used to obtain strong completeness for the iteration-free part of $\alg{A}$-valued PDL with $\alg{A}$-weighted accessibility relations considered in \cite{Sedlar2020} (with tests added). Using $\Pred = \{ \lambda^\Diamond \}$ instead (note that $\Box$ and $\Diamond$ are not necessarily inter-definable since $\neg$ need not be involutive), we claim that for $\LT$ one can use the axiomatization 
\begin{align*}
\Diamond(p\vee q) &\leftrightarrow \Diamond p \vee \Diamond q, \\ 
\Diamond (p \odot \check{c}) &\leftrightarrow \Diamond p \odot \check{c} \quad \text{for all } c \in \alg{A}. 
\end{align*} 
Indeed, to see strong one-step completeness over finite sets, take a rank-1 non-modal homomorphism $H \colon \Pred(A^X) \to A$ with $X$ finite and define $\alpha \in A^X$ by $\alpha(x) = H(\Diamond e_x)$. This $\alpha$ one-step satisfies $H$ since 
\begin{align*}
    \lambda^\Diamond(\sigma)(\alpha) &= \bigvee_{x \in X} \sigma(x) \odot \alpha(x) = \bigvee_{x\in X} \check{\sigma}_x \odot H(\Diamond e_x)\\
    & = \bigvee_{x \in X} H\big( \Diamond(\check{\sigma}_x \odot e_x)\big) = H\big( \Diamond \bigvee_{x \in X} (\check{\sigma}_x \odot e_x)\big) ,
\end{align*}
where $\check{\sigma}_x$ is the truth-constant of $\sigma(x)$. Note that $\bigvee_{x\in X} (\check{\sigma}_x \odot e_x) = \sigma$ as desired. 
\hfill $\blacksquare$
\end{exa}  

\begin{exa}[\textsf{Two-valued PDL with weighted accessibility relations}]\label{example: TWOPDL-labeled}
Consider the iteration-free part of the coalgebraic dynamic logic of Example~\ref{example:CoalgebraicDynamicLogics-PDL-Labelled2} with finite linear $\alg{A}$. Since $\PA$ is weakly preserving (Lemma~\ref{lemma:PA-is-weaklypreserving}) and $\Lambda$ is separating, we can apply Theorem~\ref{theorem:FinitelyOneStepComplete-CanModel} with the reduction axioms 
\begin{align*}
\diam{a_1 \vee a_2}^r \varphi &\leftrightarrow \diam{a_1}^r \varphi \vee \diam{a_2}^r \varphi  \\
\diam{a_1 {;} a_2}^{r} \varphi &\leftrightarrow \bigvee_{r_1 \odot r_2 \geq r} \diam{a_1}^{r_1} \diam{a_2}^{r_2} \varphi \\
\diam{\test(\psi)}^r \varphi &\leftrightarrow (\psi \wedge \varphi) 
\end{align*}
of Example~\ref{example:IndependentlyReducibleOperations}(3), Example~\ref{example:ReductionKleisli2-FiniteLinear} and Example~\ref{example:TestOperations}(1), respectively.  

An axiomatization of the underlying $\LT$ is found in \cite[Definition 4.5]{Ma2015}, given by 
$$
\Diamond^r\bot \leftrightarrow \bot, \quad \quad \Diamond^r(p \vee q) \leftrightarrow \Diamond^r p \vee \Diamond^r q, \quad \quad \Diamond^{r_1} p \rightarrow \Diamond^{r_2}p \text{ for all }r_2\leq r_1.
$$
For strong one-step completeness over finite sets observe that a rank-1 non-modal homomorphism $H \colon \Pred(2^X) \to 2$ is one-step satisfied by $\alpha \in A^X$ defined by 
$$
\alpha(x) = \bigvee \{ r \mid H(\Diamond^r \{ x \}) = 1)  \}.
$$
Indeed, for a non-empty subset $S \subseteq X$ (which is necessarily finite) we use the second axiom above to find  
$$
H(\Diamond^r S) = H(\Diamond^r \bigcup_{s \in S}\{ s \}) = \bigvee_{s\in S} H(\Diamond^r\{ s \})  
$$ 
and we can see that this equals $1$ if and only if $H(\Diamond^r \{ s \}) = 1$ for some $s \in S$.
On the other hand we have 
$$
\lambda^r(S)(\alpha) = \bigvee_{s \in S} \alpha(s) \geq r
$$ 
and this equals $1$ if and only if $\alpha(s) \geq r$ for some $s \in S$. Lastly, note that 
$$
H(\Diamond^r\{ s \}) = 1 \Leftrightarrow \alpha(s) \geq r
$$
where `$\Leftarrow$' follows directly from the definition of $\alpha$ and `$\Rightarrow$' follows from this definition combined with the third axiom above (\emph{i.e.}, if $\alpha(s) = r' \geq r$ then $H(\Diamond^{r'}\{ s \}) = 1$ implies $H(\Diamond^{r}\{ s \}) = 1$). Lastly, the special case where $S = \varnothing$ is covered by the first axiom above. 
\hfill $\blacksquare$
\end{exa} 

As preparation for our discussion of (iteration-free) $\alg{A}$-valued game logic as in Example~\ref{example:CoalgebraicDynamicLogics-GL}, we first prove the following about its coalgebraic signature functor (\emph{i.e.}, the monotone $\alg{A}$-neighbourhood functor $\MA$, recall Example~\ref{example:FunctorsAndTheirCoalgebras/Kleisli}(3)). 

\begin{lem}\label{lemma:MA-is-weaklypreserving}
For every $\alg{A}$, the functor $\MA$ is weakly preserving.  
\end{lem}  
\begin{proof}
Let $(N_m) \in \varprojlim \MA X_m$ be coherent, \emph{i.e.}, $N_m(\sigma_m) = N_{m+1}(\sigma_{m} \circ p_{m})$. We claim that this sequence is induced by 
$N \in \MA(\varprojlim X_m)$ defined by 
$$
N(\sigma) = \bigvee_{m\in \mathbb{N}} \bigvee \{ N_m(\rho^m) \mid \rho^m \circ \pi_m = \sigma \},
$$
where $\pi_m \colon \varprojlim X_n \to X_m$ are the usual projection maps. That is, we want to verify that 
$$
N_n(\sigma_n) = N(\sigma_n \circ \pi_n) = \bigvee_{m \in \mathbb{N}} \bigvee \{ N_m(\rho^m) \mid \rho^m \circ \pi_m = \sigma_n \circ \pi_n\}.
$$ 
The inequality `$\leq$' is obvious. For the other inequality `$\geq$', we show that $\rho^{m_1} \circ \pi_{m_1} = \rho^{m_2} \circ \pi_{m_2}$ always implies $N_{m_1}(\rho^{m_1}) = N_{m_2}(\rho^{m_2})$. Without loss of generality assume $m_1 \leq m_2$. Then from $\pi_{m_1} =p_{m_1} \circ \dots \circ p_{m_2 - 1} \circ \pi_{m_2}$ we get 
$$
\rho^{m_1} \circ p_{m_1} \circ \dots \circ p_{m_2 - 1} \circ \pi_{m_2} = \rho^{m_2} \circ \pi_{m_2}
$$
and since $\pi_{m_2}$ is surjective (because all $p_i$ are surjective) we get $\rho^{m_1} \circ p_{m_1} \circ \dots \circ  p_{m_2 - 1} = \rho^{m_2}$ as well. Now coherence of the sequence $(N_m)$ yields
$$
N_{m_1}(\rho^{m_1}) = N_{m_1 + 1}(\rho^{m_1} \circ p_{m_1}) = \dots = N_{m_2}(\rho^{m_1} \circ p_{m_1} \circ \dots \circ p_{m_2 - 1}) = N_{m_2}(\rho^{m_2})
$$
as desired, finishing the proof.    
\end{proof}

\begin{exa}[\textsf{Many-valued game logic}]\label{example: MVGL}
Consider the iteration-free part of the logic of Example~\ref{example:CoalgebraicDynamicLogics-GL}. 
Since $\MA$ is weakly preserving (Lemma~\ref{lemma:MA-is-weaklypreserving}) and $\Lambda$ is separating, assuming finite $\alg{A}$ we can apply Theorem~\ref{theorem:FinitelyOneStepComplete-CanModel} with the reduction axioms
\begin{align*}
\diam{a \vee b}\varphi &\leftrightarrow \diam{a}\varphi \vee \diam{b}\varphi \\
\diam{a \wedge b}\varphi &\leftrightarrow \diam{a}\varphi \wedge \diam{b}\varphi \\
\diam{a^\partial}\varphi &\leftrightarrow \neg \diam{a} \neg \varphi \\
\diam{a{;}b}\varphi &\leftrightarrow \diam{a}\diam{b}\varphi \\
\diam{\test(\psi)}\varphi &\leftrightarrow (\psi \odot \varphi) 
\end{align*}
of Example~\ref{example:IndependentlyReducibleOperations}(4), Example~\ref{example:ReductionKleisliMonadMorphism} and Example~\ref{example:ReducibleTests}(3).

For $\LT$ it suffices to ensure monotonicity, for example via the axiom 
$$
\Diamond p \rightarrow \Diamond (p \vee q).
$$ 
Strong one-step completeness over finite sets holds since every rank-1 non-modal homomorphism $H \colon \Lambda(A^X) \to A$ is obviously one-step satisfied by $\alpha \in \MA X$ defined by $\alpha(\sigma) \coloneq H(\Diamond \sigma)$ (the above axiom ensures that this $\alpha$ is indeed an $\alg{A}$-monotone neighbourhood assignment). 

For example, this may be used to obtain strong completeness for iteration-free $\lucas_n$-valued game logic, by adding the axioms $\diam{a} p \rightarrow \diam{a} (p \vee q)$ and the above reduction axioms to an axiomatization of $\lucas_n$-valued propositional logic (see, \emph{e.g.}, \cite{Cignoli2000}). 
\hfill $\blacksquare$
\end{exa}

Lastly, we show how Instantial PDL \cite{vanBenthemBezhanishviliEnqvist2019} fits into the coalgebraic picture. This uses the additional flexibility of our framework compared to that of \cite{HansenKupkeLeal2014}, allowing for coalgebraic signature functors which are not monads and countable collections of (not necessarily unary) predicate liftings. 

\begin{exa}[\textsf{Instantial PDL}]\label{example: INSTPDL}
Consider the iteration-free part of the logic of Example~\ref{example:CoalgebraicDynamicLogics-IPDL}. 
Since $\mathcal{P} \circ \mathcal{P}$ is weakly preserving (this follows from the fact that $\mathcal{P}$ is weakly preserving) and $\Lambda$ is separating on finite sets, we can apply Theorem~\ref{theorem:FinitelyOneStepComplete-CanModel} with the reduction axioms 
\begin{align*}
\diam{a \Cup b}^{k+1}(\varphi_1,\dots,\varphi_k,\varphi) &\leftrightarrow \bigvee_{K \subseteq \{ 1, \dots, k \} } \diam{a}(\varphi^K_1, \dots, \varphi^K_k, \varphi) \wedge \diam{b}(\varphi^{K^\mathrm{c}}_1, \dots, \varphi^{K^\mathrm{c}}_k, \varphi) \\
\diam{a ; b}^{k+1}(\varphi_1,\dots,\varphi_k,\varphi) &\leftrightarrow \diam{a}^{k+1} \big(\diam{b}^2(\varphi_1, \varphi), \dots, \diam{b}^2(\varphi_k, \varphi), \diam{b}^1\varphi \big) \\
\diam{\test(\psi)}^{k+1}(\varphi_1,\dots,\varphi_k,\varphi) &\leftrightarrow (\psi \wedge \varphi \wedge \varphi_1 \wedge \dots \wedge \varphi_k)
\end{align*}
found in \cite{vanBenthemBezhanishviliEnqvist2019}, also see Example~\ref{subsection:IndependentlyReducibleExamples}(5),
Example~\ref{example:ReductionCompositionInstantialNBH} and Example~\ref{example:ReducibleTests}(4), respectively. Furthermore, we may also expand $\CoOp$ by the composition $\star$ and the instantial counter-domain $\sim$, where the corresponding reduction axioms are 
\begin{align*}
\diam{a \star b}^{k+1}(\varphi_1,\dots,\varphi_k,\varphi) &\leftrightarrow \diam{a}^2 \big( \diam{b}(\varphi_1, \dots, \varphi_k, \varphi), \top \big) \\
\diam{{\sim}a}^{k+1}(\varphi_1,\dots,\varphi_k,\varphi) &\leftrightarrow \neg \diam{a}^1 \top \wedge \varphi \wedge \varphi_1 \wedge \dots \wedge \varphi_k
\end{align*}
as in Example~\ref{example:ReductionCompositionInstantialNBH} and Example~\ref{example:IndependentlyReducibleOperations}(7).
   
For the underlying coalgebraic logic $\LT$ the clear candidate is the (notably rank-1) axiomatization of \cite[Section 4]{vanBenthemBezhanishviliEnqvistYu2017}. 
\hfill $\blacksquare$   
\end{exa} 

\section{Concluding remarks and future research}\label{section:Conclusion}
We presented a general coalgebraic framework which can account for many-valued dynamic logics where truth-degrees are taken from an arbitrary $\FLew$-algebra $\alg{A}$. This framework generalises \cite{HansenKupkeLeal2014,HansenKupke2015} already in the 2-valued Boolean case by providing a more general treatment of coalgebra operations and tests, which does not require the coalgebra signature functor to be a monad, and accommodates countably many predicate liftings $\lambda$ of arbitrary finite arities.
We defined a notion of reducibility which generalises the reduction axioms of PDL and game logic, and showed that reducible operations and tests are safe for bisimulation/behavioural equivalence. Moreover, for finite $\alg{A}$, we proved a uniform strong completeness theorem  for many-valued coalgebraic dynamic logics of reducible operations and tests (note that iteration is not reducible).
Our framework encompasses dynamic logics such as many-valued PDL \cite{Teheux2014,Sedlar2016} and two-valued instantial PDL \cite{vanBenthemBezhanishviliEnqvist2019} that were not covered by \cite{HansenKupkeLeal2014,HansenKupke2015},
as well as two examples that, to our knowledge, are new:
two-valued PDL with weighted accessibility relations and many-valued game logic, see Examples~\ref{example:CoalgebraicDynamicLogics-PDL-Labelled2} \&~\ref{example:CoalgebraicDynamicLogics-GL}. For these logics, our results (re)establish strong completeness for their iteration-free fragments.     

In this article, we opted to work with $\FLew$-algebras of truth-degrees since this is the most common framework for many-valued modal logic (\emph{e.g.}, \cite{BouEstevaGodo2011}), simplifies the presentation and is required for the strong completeness proof via non-modal homomorphisms in Section~\ref{section:Completeness}. Nevertheless, let us mention that the general framework and the results of Sections~\ref{section:CoalgebraicDynmicLogic}--\ref{section:ReducibleOperations} work in more generality, that is, with \emph{arbitrary} algebras $\alg{A}$ of truth-degrees (indeed, while it is relevant in specific examples, the proofs in Sections~\ref{section:CoalgebraicDynmicLogic}--\ref{section:ReducibleOperations} never relied on the assumption that $\alg{A}$ is a $\FLew$-algebra). In particular, one can also cover prior instances of many-valued propositional dynamic logic not based on $\FLew$-algebras, such as \cite{Sedlar2020} which also includes $\mathsf{FL}$-algebras which are not necessarily commutative/integral and \cite{Sedlar2016} which is based on four-valued Belnap-Dunn logic.

An obvious question for future consideration is how to prove (weak) completeness with \emph{non-reducible} operations, in particular with \emph{iteration}. Since the current paper can be seen as generalising \cite{HansenKupkeLeal2014}, one would reasonably attempt to generalise its follow-up \cite{HansenKupke2015}. However, this could be challenging in the many-valued case, for example \cite{Sedlar2020} needs to replace $(\cdot)^\ast$ by $(\cdot)^+$ and omit tests for technical reasons. The coalgebraic analysis may shed light on the nature of these issues, helping to resolve them in the future. 

Other interesting questions arise regarding \emph{safety}. For PDL, \cite{vanBenthem1998} shows that all \emph{safe first-order definable} $\mathcal{P}$-operations are combinations of $\cup$, $;$ , $\sim$ and propositional tests. A similar result for game logic is shown in \cite{Pauly2000}. One could try to characterise safe coalgebra operations that are first-order definable via \emph{coalgebraic correspondence theory} \cite{SchroederPattinson2010} (also see \cite{Litak2018} using a less expressive language). In work on behavioural/logical distances \cite{Beohar2024}, algebras of truth-values are often assumed to be commutative integral quantales (that is, complete $\FLew$-algebras). In this context, safety of a coalgebra operation $O$ could perhaps be studied as non-expansion of behavioural distance in the sense that $d_\func{F}(O(\vec{\gamma})(x), O(\vec{\gamma})(y)) \leq \mathrm{min}\{ d_\func{F}(\gamma_i(x), \gamma_i(y)) \mid 1 \leq i \leq \arity(O) \}$ where $d_\func{F}$ is a corresponding pseudometric on $\func{F}X$ \cite{Baldan2018}.

While we only considered $\Set$-functors here, it would be interesting to explore the topics of this article over various other \emph{base categories}, such as the category of \emph{measurable spaces} to accommodate \emph{probabilistic dynamic logics} \cite{Kozen1985,Doberkat2015}. This category could also be \emph{quantale-enriched} (relating this to the previous paragraph) or \emph{poset-enriched} (as in \emph{positive} coalgebraic logic \cite{DahlqvistKurz2017}). Other interesting `candidates' for such base categories (and monads defined on them) can for example be found in \cite{Jacobs2017}. 
Generalisations of our work beyond $\Set$ could be investigated via the dual adjunctions approach as in 
the recent \cite{Schoen2024} which provides a general framework for studying alternation-free coalgebraic fixpoint logic, and includes an instantiation to positive, test-free PDL. Their main result is a general adequacy theorem; completeness is mentioned as future work.  

Lastly we mention that completeness (as in Section~\ref{section:Completeness})  
for coalgebraic logics valued in an \emph{infinite} $\alg{A}$ currently seems out of reach. One likely reason for this is the (current) lack of research in \emph{many-valued coalgebraic logic} (some of the few examples are \cite{PattinsonSchroeder2011, BilkovaDostal2013, BilkovaDostal2016, KurzPoiger2023, KurzPoigerTeheux2024,LinLiau2023}). Therefore, we also advocate for further future research in this area. 

\section*{Acknowledgements} 
\emph{Wolfgang Poiger:} {\tiny \euflag}
Funded by the European Union under the Operational Programme Jan Amos Komenský project no. {CZ.02.01.01/00/23\_025/0008724} (\emph{DEZINFO}) and the European Research Executive Agency ERA Fellowship project no. 101334342 (\emph{GLOWS}). 
Views and opinions expressed are however those of the author(s) only and do not necessarily reflect those of the European Union or European Research Executive Agency. Neither the European Union nor the granting authority can be held responsible for them.  

\bibliographystyle{alphaurl}
\bibliography{References}

\end{document}